\documentclass{article}
\usepackage{fullpage}
\usepackage{amsmath,amsfonts,amssymb,amsthm}
\usepackage{graphicx} 
\usepackage{hyperref}
\usepackage[nameinlink]{cleveref}
\usepackage[T1]{fontenc}
\usepackage{physics}
\usepackage{todonotes}
\usepackage{mdframed} 
\usepackage{lipsum}
\usepackage{multicol}
\allowdisplaybreaks

\hypersetup{colorlinks={true},linkcolor={blue},citecolor=magenta}

\usepackage{comment}

\usepackage[style=alphabetic,minalphanames=3,maxalphanames=4,maxnames=99,backref=true]{biblatex}
\title{Cryptography in the Common Haar State Model:\\ Feasibility Results and Separations\footnote{This subsumes~\cite{AGL24}.}}

\author{Prabhanjan Ananth\thanks{\texttt{prabhanjan@cs.ucsb.edu}}\\ {\small UCSB} \and Aditya Gulati\thanks{\texttt{adityagulati@ucsb.edu}} \\ {\small UCSB} \and Yao-Ting Lin\thanks{\texttt{yao-ting\_lin@ucsb.edu}} \\ {\small UCSB}}
\date{}

\newcommand{\YTnote}[1]{}
\newcommand{\pnote}[1]{}
\newcommand{\AG}[1]{}

\begin{document}

\maketitle

\newtheorem{theorem}{Theorem}[section]

\newtheorem{definition}[theorem]{Definition}

\newtheorem{lemma}[theorem]{Lemma}

\newtheorem{fact}[theorem]{Fact}

\newtheorem{proposition}[theorem]{Proposition}

\newtheorem{corollary}[theorem]{Corollary}

\newtheorem{remark}[theorem]{Remark}

\newtheorem{claim}[theorem]{Claim}

\Crefname{fact}{fact}{Fact}

\newtheorem{mytheorem}{Theorem}
\newtheorem{mycorollary}[mytheorem]{Corollary}


\newcommand{\etal}{et~al.\ }
\newcommand{\aka}{also known as,\ }
\newcommand{\resp}{resp.,\ }
\newcommand{\ie}{i.e.,\ }
\newcommand{\wolog}{w.l.o.g.\ }
\newcommand{\Wolog}{W.l.o.g.\ }
\newcommand{\eg}{e.g.,\ }
\newcommand{\wrt} {with respect to\ }
\newcommand{\cf}{{cf.,\ }}

\newcommand{\st}{\ \text{s.t.}\ }


\newcommand{\N}{\mathbb{N}}
\newcommand{\Q}{\mathbb{Q}}
\newcommand{\R}{\mathbb{R}}
\newcommand{\Z}{\mathbb{Z}}

\newcommand{\round}[1]{\lfloor #1 \rceil}
\newcommand{\ceil}[1]{\lceil #1 \rceil}
\newcommand{\floor}[1]{\lfloor #1 \rfloor}
\newcommand{\angles}[1]{\langle #1 \rangle}
\newcommand{\parens}[1]{( #1 )}
\newcommand{\bracks}[1]{[ #1 ]}
\renewcommand{\bra}[1]{\langle#1\rvert}
\renewcommand{\braket}[2]{\langle #1 \mid #2 \rangle}
\renewcommand{\ket}[1]{\lvert#1\rangle}
\newcommand{\set}[1]{\{ #1 \}}
\newcommand{\bit}{\{0,1\}}


\newcommand{\cA}{{\mathcal A}}
\newcommand{\cB}{{\mathcal B}}
\newcommand{\cC}{{\mathcal C}}
\newcommand{\cD}{{\mathcal D}}
\newcommand{\cE}{{\mathcal E}}
\newcommand{\cF}{{\mathcal F}}
\newcommand{\cG}{{\mathcal G}}
\newcommand{\cH}{{\mathcal H}}
\newcommand{\cI}{{\mathcal I}}
\newcommand{\cJ}{{\mathcal J}}
\newcommand{\cK}{{\mathcal K}}
\newcommand{\cL}{{\mathcal L}}
\newcommand{\cM}{{\mathcal M}}
\newcommand{\cN}{{\mathcal N}}
\newcommand{\cO}{{\mathcal O}}
\newcommand{\cP}{{\mathcal P}}
\newcommand{\cQ}{{\mathcal Q}}
\newcommand{\cR}{{\mathcal R}}
\newcommand{\cS}{{\mathcal S}}
\newcommand{\cT}{{\mathcal T}}
\newcommand{\cU}{{\mathcal U}}
\newcommand{\cV}{{\mathcal V}}
\newcommand{\cW}{{\mathcal W}}
\newcommand{\cX}{{\mathcal X}}
\newcommand{\cY}{{\mathcal Y}}
\newcommand{\cZ}{{\mathcal Z}}

\newcommand{\bfA}{\mathbf{A}}
\newcommand{\bfB}{\mathbf{B}}
\newcommand{\bfC}{\mathbf{C}}
\newcommand{\bfD}{\mathbf{D}}
\newcommand{\bfE}{\mathbf{E}}
\newcommand{\bfF}{\mathbf{F}}
\newcommand{\bfG}{\mathbf{G}}
\newcommand{\bfH}{\mathbf{H}}
\newcommand{\bfI}{\mathbf{I}}
\newcommand{\bfJ}{\mathbf{J}}
\newcommand{\bfK}{\mathbf{K}}
\newcommand{\bfL}{\mathbf{L}}
\newcommand{\bfM}{\mathbf{M}}
\newcommand{\bfN}{\mathbf{N}}
\newcommand{\bfO}{\mathbf{O}}
\newcommand{\bfP}{\mathbf{P}}
\newcommand{\bfQ}{\mathbf{Q}}
\newcommand{\bfR}{\mathbf{R}}
\newcommand{\bfS}{\mathbf{S}}
\newcommand{\bfT}{\mathbf{T}}
\newcommand{\bfU}{\mathbf{U}}
\newcommand{\bfV}{\mathbf{V}}
\newcommand{\bfW}{\mathbf{W}}
\newcommand{\bfX}{\mathbf{X}}
\newcommand{\bfY}{\mathbf{Y}}
\newcommand{\bfZ}{\mathbf{Z}}

\newcommand{\bfa}{\mathbf{a}}
\newcommand{\bfb}{\mathbf{b}}
\newcommand{\bfc}{\mathbf{c}}
\newcommand{\bfd}{\mathbf{d}}
\newcommand{\bfe}{\mathbf{e}}
\newcommand{\bff}{\mathbf{f}}
\newcommand{\bfg}{\mathbf{g}}
\newcommand{\bfh}{\mathbf{h}}
\newcommand{\bfi}{\mathbf{i}}
\newcommand{\bfj}{\mathbf{j}}
\newcommand{\bfk}{\mathbf{k}}
\newcommand{\bfl}{\mathbf{l}}
\newcommand{\bfm}{\mathbf{m}}
\newcommand{\bfn}{\mathbf{n}}
\newcommand{\bfo}{\mathbf{o}}
\newcommand{\bfp}{\mathbf{p}}
\newcommand{\bfq}{\mathbf{q}}
\newcommand{\bfr}{\mathbf{r}}
\newcommand{\bfs}{\mathbf{s}}
\newcommand{\bft}{\mathbf{t}}
\newcommand{\bfu}{\mathbf{u}}
\newcommand{\bfv}{\mathbf{v}}
\newcommand{\bfw}{\mathbf{w}}
\newcommand{\bfx}{\mathbf{x}}
\newcommand{\bfy}{\mathbf{y}}
\newcommand{\bfz}{\mathbf{z}}

\newcommand{\sfA}{\mathsf{A}}
\newcommand{\sfB}{\mathsf{B}}
\newcommand{\sfC}{\mathsf{C}}
\newcommand{\sfD}{\mathsf{D}}
\newcommand{\sfE}{\mathsf{E}}
\newcommand{\sfF}{\mathsf{F}}
\newcommand{\sfG}{\mathsf{G}}
\newcommand{\sfH}{\mathsf{H}}
\newcommand{\sfI}{\mathsf{I}}
\newcommand{\sfJ}{\mathsf{J}}
\newcommand{\sfK}{\mathsf{K}}
\newcommand{\sfL}{\mathsf{L}}
\newcommand{\sfM}{\mathsf{M}}
\newcommand{\sfN}{\mathsf{N}}
\newcommand{\sfO}{\mathsf{O}}
\newcommand{\sfP}{\mathsf{P}}
\newcommand{\sfQ}{\mathsf{Q}}
\newcommand{\sfR}{\mathsf{R}}
\newcommand{\sfS}{\mathsf{S}}
\newcommand{\sfT}{\mathsf{T}}
\newcommand{\sfU}{\mathsf{U}}
\newcommand{\sfV}{\mathsf{V}}
\newcommand{\sfW}{\mathsf{W}}
\newcommand{\sfX}{\mathsf{X}}
\newcommand{\sfY}{\mathsf{Y}}
\newcommand{\sfZ}{\mathsf{Z}}

\newcommand{\sfa}{\mathsf{a}}
\newcommand{\sfb}{\mathsf{b}}
\newcommand{\sfc}{\mathsf{c}}
\newcommand{\sfd}{\mathsf{d}}
\newcommand{\sfe}{\mathsf{e}}
\newcommand{\sff}{\mathsf{f}}
\newcommand{\sfg}{\mathsf{g}}
\newcommand{\sfh}{\mathsf{h}}
\newcommand{\sfi}{\mathsf{i}}
\newcommand{\sfj}{\mathsf{j}}
\newcommand{\sfk}{\mathsf{k}}
\newcommand{\sfl}{\mathsf{l}}
\newcommand{\sfm}{\mathsf{m}}
\newcommand{\sfn}{\mathsf{n}}
\newcommand{\sfo}{\mathsf{o}}
\newcommand{\sfp}{\mathsf{p}}
\newcommand{\sfq}{\mathsf{q}}
\newcommand{\sfr}{\mathsf{r}}
\newcommand{\sfs}{\mathsf{s}}
\newcommand{\sft}{\mathsf{t}}
\newcommand{\sfu}{\mathsf{u}}
\newcommand{\sfv}{\mathsf{v}}
\newcommand{\sfw}{\mathsf{w}}
\newcommand{\sfx}{\mathsf{x}}
\newcommand{\sfy}{\mathsf{y}}
\newcommand{\sfz}{\mathsf{z}}
\newcommand{\eps}{\varepsilon}
\newcommand{\veps}{\varepsilon}
\newcommand{\vphi}{\varphi}
\newcommand{\vsigma}{\varsigma}
\newcommand{\vrho}{\varrho}
\newcommand{\vpi}{\varpi}


\newcommand{\secp}{{\lambda}}
\newcommand{\poly}{\mathsf{poly}}
\newcommand{\polylog}{\operatorname{polylog}}
\newcommand{\loglog}{\operatorname{loglog}}
\newcommand{\GF}{\operatorname{GF}}
\newcommand{\Exp}{\operatorname*{\mathbb{E}}}
\newcommand{\Ex}{\Exp}
\newcommand{\View}[1]{\mathsf{View}\angles{#1}}
\newcommand{\Var}{\operatorname*{Var}}
\newcommand{\ShH}{\operatorname{H}}
\newcommand{\maxH}{\operatorname{H_0}}
\newcommand{\minH}{\operatorname{H_{\infty}}}
\newcommand{\Enc}{\operatorname{Enc}}
\newcommand{\Setup}{\operatorname{Setup}}
\newcommand{\KGen}{\operatorname{KeyGen}}
\newcommand{\StGen}{\operatorname{StateGen}}
\newcommand{\Dec}{\operatorname{Dec}}
\newcommand{\Sign}{\operatorname{Sign}}
\newcommand{\Samp}{\mathsf{Samp}}
\newcommand{\Ver}{\mathsf{Ver}}
\newcommand{\Gen}{\operatorname{Gen}}
\newcommand{\negl}{\mathsf{negl}}
\newcommand{\Supp}{\operatorname{Supp}}
\newcommand{\maj}{\operatorname*{maj}}
\newcommand{\argmax}{\operatorname*{arg\,max}}
\newcommand{\Image}{\operatorname{Im}}

\newcommand{\TD}{\mathsf{TD}}


\newcommand{\class}[1]{\ensuremath{\mathbf{#1}}}
\newcommand{\coclass}[1]{\class{co\mbox{-}#1}} 
\newcommand{\prclass}[1]{\class{Pr#1}}
\newcommand{\PPT}{\class{PPT}}
\newcommand{\BPP}{\class{BPP}}
\newcommand{\NC}{\class{NC}}
\newcommand{\AC}{\class{AC}}
\newcommand{\NP}{\ensuremath{\class{NP}}}
\newcommand{\coNP}{\ensuremath{\coclass{NP}}}
\newcommand{\RP}{\class{RP}}
\newcommand{\coRP}{\coclass{RP}}
\newcommand{\ZPP}{\class{ZPP}}
\newcommand{\IP}{\class{IP}}
\newcommand{\coIP}{\coclass{IP}}
\newcommand{\AM}{\class{AM}}
\newcommand{\coAM}{\class{coAM}}
\newcommand{\MA}{\class{MA}}
\renewcommand{\P}{\class{P}}
\newcommand\prBPP{\prclass{BPP}}
\newcommand\prRP{\prclass{RP}}
\newcommand{\Ppoly}{\class{P/poly}}
\newcommand{\DTIME}{\class{DTIME}}
\newcommand{\ETIME}{\class{E}}
\newcommand{\BPTIME}{\class{BPTIME}}
\newcommand{\AMTIME}{\class{AMTIME}}
\newcommand{\coAMTIME}{\class{coAMTIME}}
\newcommand{\NTIME}{\class{NTIME}}
\newcommand{\coNTIME}{\class{coNTIME}}
\newcommand{\AMcoAM}{\class{AM} \cap \class{coAM}}
\newcommand{\NPcoNP}{\class{N} \cap \class{coN}}
\newcommand{\EXP}{\class{EXP}}
\newcommand{\SUBEXP}{\class{SUBEXP}}
\newcommand{\qP}{\class{\wt{P}}}
\newcommand{\PH}{\class{PH}}
\newcommand{\NEXP}{\class{NEXP}}
\newcommand{\PSPACE}{\class{PSPACE}}
\newcommand{\NE}{\class{NE}}
\newcommand{\coNE}{\class{coNE}}
\newcommand{\SharpP}{\class{\#P}}
\newcommand{\prSZK}{\prclass{SZK}}
\newcommand{\BPPSZK}{\BPP^{\SZK}}
\newcommand{\TFAM}{\class{TFAM}}
\newcommand{\RTFAM}{{\R\text{-}\class{TUAM}}}
\newcommand{\FBPP}{\class{BPP}}
\newcommand{\BPPR}{\FBPP^{\RTFAM}}
\newcommand{\TFNP}{\mathbf{TFNP}}
\newcommand{\MIP}{\mathbf{MIP}}
\newcommand{\ZK}{\class{ZK}}
\newcommand{\SZK}{\class{SZK}}
\newcommand{\SZKP}{\class{SZKP}}
\newcommand{\SZKA}{\class{SZKA}}
\newcommand{\CZKP}{\class{CZKP}}
\newcommand{\CZKA}{\class{CZKA}}
\newcommand{\coSZKP}{\coclass{SZKP}}
\newcommand{\coSZKA}{\coclass{SZKA}}
\newcommand{\coCZKP}{\coclass{CZKP}}
\newcommand{\coCZKA}{\coclass{CZKA}}
\newcommand{\NISZKP}{\class{NI\mbox{-}SZKP}}
\newcommand{\phcc}{\ensuremath \PH^{\mathrm{cc}}}
\newcommand{\pspacecc}{\ensuremath \PSPACE^{\mathrm{cc}}}
\newcommand{\sigcc}[1]{\ensuremath \class{\Sigma}^{\mathrm{cc}}_{#1}}
\newcommand{\aco}{\ensuremath \AC^0}
\newcommand{\ncone}{\ensuremath \NC^1}

\newcommand{\shadowgen}{\mathsf{ShadowGen}}
\newcommand{\E}{\mathop{\mathbb{E}}}
\newcommand{\commit}{{\mathsf{Commit}}}
\newcommand{\reveal}{{\mathsf{Reveal}}}
\newcommand{\aux}{{\mathsf{aux}}}
\newcommand{\com}{{\mathsf{com}}}
\newcommand{\swapt}{{\mathsf{SwapTest}}}
\newcommand{\accept}{{\mathsf{accept}}}
\newcommand{\C}{\mathbb{C}}
\newcommand{\Haar}{\mathcal{H}}
\newcommand{\sym}{\mathsf{Sym}}
\renewcommand{\ketbra}[2]{\ket{#1}\bra{#2}}
\renewcommand{\braket}[2]{\langle #1 | #2 \rangle}
\newcommand{\inner}[2]{\langle #1, #2 \rangle}
\newcommand{\ugets}{\xleftarrow{\$}}
\newcommand{\ol}[1]{\overline{#1}}

\newcommand{\gap}{k} 
\newcommand{\keylength}{m} 

\newcommand{\supp}{\mathsf{supp}}
\newcommand{\freq}[2]{\mathsf{freq}_{#2} (#1)}
\newcommand{\hamming}{\mathsf{size}}
\newcommand{\type}{\mathsf{type}}
\newcommand{\bintype}{\mathsf{bintype}}
\newcommand{\setT}{\mathsf{mset}}
\newcommand{\rep}{\mathsf{rep}}
\newcommand{\goodpre}[4]{\mathcal{I}^{(#3)}_{#1,#2} (#4)}

\newcommand{\good}{\mathsf{Good}}
\newcommand{\bad}{\mathsf{Bad}}

\newcommand{\alice}{A}
\newcommand{\bob}{B}
\newcommand{\charlie}{C}
\newcommand{\eve}{E}

\newcommand{\wt}{\widetilde}
\newcommand{\wh}{\widehat}

\newcommand{\Des}{\mathsf{Des}}
\newcommand{\Com}{\mathsf{Com}}
\newcommand{\KA}{\mathsf{KA}}

\newcommand{\puz}{\mathsf{Puz}}
\newcommand{\sol}{\mathsf{Sol}}

\newcommand{\Real}{\mathsf{Real}}
\newcommand{\Sim}{\mathsf{Sim}}
\newcommand{\Hyb}{\mathsf{Hyb}}

\newcommand{\dtv}{\mathsf{d_{TV}}}
\newcommand{\ch}{\mathbf{ch}}

\newcommand{\plain}{\mathsf{plain}}
\newcommand{\adv}{\mathsf{adv}}
\newcommand{\complete}{\mathsf{Complete}}
\newcommand{\tomography}{\mathsf{Tomography}}
\newcommand{\LOCC}{\mathsf{LOCC}}
\begin{abstract}
\noindent Common random string model is a popular model in classical cryptography. We study a quantum analogue of this model called the common Haar state (CHS) model. In this model, every party participating in the cryptographic system receives many copies of one or more i.i.d Haar random states. 
\par We study feasibility and limitations of cryptographic primitives in this model and its variants:
\begin{itemize}
    \item We present a construction of pseudorandom function-like states with security against computationally unbounded adversaries, as long as the adversaries only receive (a priori) bounded number of copies. By suitably instantiating the CHS model, we obtain a new approach to construct pseudorandom function-like states in the plain model. 
    \item We present separations between pseudorandom function-like states (with super-logarithmic length) and quantum cryptographic primitives, such as interactive key agreement and bit commitment, with classical communication. To show these separations, we prove new results on the indistinguishability of identical versus independent Haar states against LOCC (local operations, classical communication) adversaries.  
\end{itemize}
\end{abstract}
\newpage 
\tableofcontents 
\newpage

\section{Introduction} 
In classical cryptography, the common random string and the common reference string models were primarily introduced to tackle cryptographic tasks that were impossible to achieve in the plain model. In the common reference string model, there is a trusted setup who produces a string that every party has access to. In the common random string model, the common string available to all the parties is sampled uniformly at random. Due to the lack of structure required from the common random string model, it is in general the more desirable model of the two. There have been many constructions proposed over the years in these two models, including non-interactive zero-knowledge~\cite{BFM19}, secure computation with universal composition~\cite{CF01,CLOS02} and two-round secure computation~\cite{GS19,BL19}. 
\par It is a worthy pursuit to study similar models for quantum cryptographic protocols. In the quantum world, there is an option to define models that are intrinsically quantum in nature. For instance, we could define a model wherein a trusted setup produces a quantum state and every party participating in the cryptographic system receives one or more copies of this quantum state. Indeed, two works by Morimae, Nehoran and Yamakawa~\cite{MNY23} and Qian~\cite{Qian23} consider this model, termed as the {\em common reference quantum state model} (CRQS). They proposed a construction of unconditionally secure commitments in this model.  Quantum commitments is a foundational notion in quantum cryptography. In recent years, quantum commitments have been extensively studied~\cite{AQY21,MY21,AGQY22,MY23onewayness,BCQ23,Bra23} due to its implication to secure computation~\cite{BartusekCKM21a,GLSV21}. The fact that information-theoretically secure commitments are impossible in the plain model~\cite{LoChau97,May97,CLM23} renders the contributions of~\cite{MNY23,Qian23} particularly interesting. 

\paragraph{Common Haar State Model.} While CRQS is a quantum analogue of the common reference string model, in a similar vein, we can ask if there is a quantum analogue of the common random string model. We consider a novel model called the {\em common Haar state model} (CHS). In this model, every party in the system (including the adversary) receives many copies of many i.i.d Haar states. We believe that the CHS model is more pragmatic than the CRQS model owing to the fact that we do not require any structure from the common public state. This raises the possibility of avoiding a trusted setup altogether and instead we could rely upon naturally occuring physical processes to obtain the Haar states. This model was also recently introduced in an independent and concurrent recent work\footnote{We refer the reader to~\Cref{sec:cgg24} for a comparison with CCS. } by Chen, Coladangelo and Sattath~\cite{chen2024power} (henceforth, referred to as CCS). 
\par There are three reasons to study this model. Firstly, this model allows us to bypass impossibility results in the plain model. For instance, as we will see later, primitives that require computational assumptions in the plain model, can instead be designed with information-theoretic security in the CHS model. Second, perhaps a less intuitive reason, is that the constructions proposed in this model can, in some cases, be adopted to obtain constructions in the plain model by instantiating the Haar states either using state designs or pseudorandom state generators (PRSGs)~\cite{JLS18}. This leads to a modular approach of designing cryptographic primitives from PRS: first design the primitive in the CHS model and then instantiate the common Haar state using PRS. Finally, this model can be leveraged to demonstrate separations between different quantum cryptographic primitives. 

\subsection{Our Results}
\noindent We explore both feasibility results and black-box separations in the CHS model. 

\subsubsection{Feasibility Results}
\label{sec:intro:feasibility}

\paragraph{Pseudorandom Function-Like States with Statistical Security.} We study the possibility of designing pseudorandom function-like state generators (PRFSGs), introduced by Ananth, Qian and Yuen~\cite{AQY21}, with statistical security in the CHS model. Roughly speaking, a PRFSG is an efficient keyed quantum circuit that can be used to produce many pseudorandom states. We refer the reader to~\Cref{sec:qpseudorandomness} for a detailed discussion on the different notions of pseudorandomness in the quantum world. 
\par We are interested in designing $(\lambda,m,n,t)$-PR\underline{F}SGs in the setting when $n \geq \lambda$ and $m=\Omega(\log(\lambda))$, where $\lambda$ is the key length, $m$ is the input length, $n$ is the output length (and also the number of the qubits in the common Haar state) and $t$ is the maximum number of queries that can be requested by the adversary. However, in the CHS model, we can in fact achieve statistical security. 
\par We show the following. 

\newcommand{\secparam}{\lambda}
\begin{theorem}[Informal]
\label{thm:intro:prsgs}
There is a statistically secure $(\secparam,m,n,\ell)$-PR\underline{F}SG in the CHS model, for $m=\secparam^c$, $n \geq \secparam$ and $\ell=O\left( \frac{\secparam^{1-c}}{\log(\secparam)^{1+\varepsilon}} \right)$, for any constant $\varepsilon > 0$ and for all $c \in [0,1)$. 
\end{theorem} 

\noindent CCS is the only other work that has studied pseudorandomness in the CHS model. There are a few advantages of our result over CCS:
\begin{itemize}
    \item Our theorem subsumes and generalizes the result of CCS who showed $(\secparam,n,t)$-PRSGs exists in their model, where the output length is larger than the key length, i.e., $n > \secparam$ and moreover, when $t=1$ with $t$ being the number of copies of the PRS state given to the adversary.
    \item Our construction, when restricted to the case of PRSGs, is slightly simpler than CCS: in CCS, on a subset of qubits of the Haar state, a random Pauli operator is applied whereas in our case a random Pauli $Z$ operator is applied. Our construction of PR\underline{F}SG uses the seminal Goldreich-Goldwasser-Micali approach~\cite{GGM86} to go from one-query security to many-query security. 
    \item They propose novel sophisticated tools in their analysis whereas our analysis is arguably more elementary using well known facts about symmetric subspaces.
    \item Finally, we can achieve arbitrary stretch whereas it is unclear whether this is also achieved by CCS. 
\end{itemize}

\noindent As a side contribution, the proof of our PRSG construction also simplifies the proof of the quantum public-key construction of Coladangelo~\cite{Col23}; this is due to the fact the core lemma proven in~\cite{Col23} is implied by the above theorem. 
\par Interestingly, the above theorem has implications for computationally secure pseudorandomness in the plain model. Specifically, we obtain the following corollary by instantiating the CHS model using stretch PRSGs: 

\begin{corollary}
{\em Assuming $(\secparam,n,\ell)$-PRSGs, there exists $(\secparam',m,n,t)$-PR\underline{F}SGs, where $n > \secparam' > \secparam$, $m=\secparam^c$ and $\ell=O\left( \frac{\secparam^{1-c}}{\log(\secparam)^{1+\varepsilon}} \right)$, for any constant $\varepsilon > 0$ and $c \in [0,1)$.} 
\end{corollary}

\noindent Prior to our work, stretch PRFSGs for super-logarithmic input length, even in the bounded query setting, was only known from one-way functions~\cite{AQY21}. This complements the work of~\cite{AQY21} who showed a construction of PRFSGs for logarithmic input length from PRSGs.  

\par Interestingly, the state generators in both works (CCS and ours) only consume one copy of a single Haar state. In this special case, it is interesting to understand whether we can extend our result to the setting when the adversary receives $ \frac{\secparam}{\log(\secparam)}$ copies or more. We show this is not possible. 

\begin{theorem}[Informal]
 There does not exist a secure $(\secparam,m,n,\ell)$-PR\underline{F}SG, for any $m \geq 1$, in the CHS model, where $n=\omega(\log(\secparam))$ and $\ell=\Omega\left( \frac{\secparam}{\log(\secparam)} \right)$.    
\end{theorem}

\noindent CCS also proved a lower bound where they showed that unbounded copy pseudorandom states do not exist. Their negative result is stronger in the sense that they rule out PRSGs who use up many copies of the Haar states from the CHRS and thus, their work gives a clean separation between 1-copy stretch PRS and unbounded copy PRS which was not known before. On the other hand, for the special case when the PR\underline{F}SG takes only one copy of the Haar state, we believe our result yields better parameters. 

\paragraph{Commitments.} In addition to pseudorandomness, we also study the possibility of constructing other cryptographic primitives in the CHS model. We show the following:  

\begin{theorem}[Informal]
There is an unconditionally secure bit commitment scheme in the CHS model.  
\end{theorem}

\noindent Both our construction and the commitments scheme proposed by CCS are different although they share strong similarities. 

\subsubsection{Black-Box Separations} 
\paragraph{LOCC Indistinguishability.} We separate pseudorandom function-like states and quantum cryptographic primitives with classical communication using a variant of the CHS model. At the heart of our separations is a novel result that proves indistinguishability of identical versus independent Haar states against LOCC (local operations, classical communication) adversaries. More precisely, $(A,B)$ is an LOCC adversary if $A$ and $B$ are quantum algorithms who can communicate with each other via only classical communication channels. It is important that $A$ and $B$ do not share any entanglement. Moreover, we restrict our attention to LOCC distinguishers which are LOCC adversaries of the form $(A,B)$ where $A$ does not output anything whereas $B$ outputs a single bit. We say that a LOCC distinguisher $(A,B)$ can distinguish two states $\rho_{\sfA \sfB}$ and $\sigma_{\sfA \sfB}$ with probability at most $\varepsilon$, referred to as $\varepsilon$-LOCC indistinguishability, where $A$ receives the register $\sfA$ and $B$ receives the register $\sfB$, if $|{\sf Pr}\left[ 1 \leftarrow (A,B)(\rho_{\sfA \sfB}) \right] - {\sf Pr}\left[1 \leftarrow (A,B)(\sigma_{\sfA \sfB}) \right]| = \varepsilon$. Of particular interest is the case when 
$$\rho_{\sfA \sfB} = \Ex_{\ket{\psi} \leftarrow \Haar_n} \left[ (\ket{\psi}^{\otimes t})_{\sfA} \otimes (\ket{\psi}^{\otimes t})_{\sfB} \right],\ \sigma_{\sfA \sfB} = \Ex_{\substack{\ket{\psi} \leftarrow \Haar_n,\\ \ket{\phi} \leftarrow \Haar_n}} \left[ (\ket{\psi}^{\otimes t})_{\sfA} \otimes (\ket{\phi}^{\otimes t})_{\sfB} \right]$$  
Here, $\Haar_n$ denotes the Haar distribution on $n$-qubit quantum states and $t$ is polynomial in $n$. A couple of works by Harrow~\cite{Har23} and Chen, Cotler, Huang and Li~\cite{chen2022exponential} prove that the LOCC indistinguishability of $\rho_{\sfA \sfB}$ and $\sigma_{\sfA \sfB}$ is negligible in $n$ in the case when $t=1$. In this work, we extend to the case when $t$ is arbitrary. 

\begin{theorem}
\label{sec:thm:locc:intro}
$\rho_{\sfA \sfB}$ and $\sigma_{\sfA \sfB}$ (defined above) are $\varepsilon$-LOCC indistinguishable, where $\varepsilon=O\left( \frac{t^2}{2^n} \right)$. 
\end{theorem}

\noindent We also show that the above bound is tight by demonstrating an LOCC distinguisher whose distinguishing probability is $\Theta(\frac{t^2}{2^n})$.  
\par Recently, Ananth, Kaleoglu and Yuen~\cite{AKY24} prove the indistinguishability of $\rho_{\sfA \sfB}$ and $\sigma_{\sfA \sfB}$ in the dual setting, against non-local adversaries that can share entanglement but cannot communicate. 
\par The above theorem can easily be extended to the multi-party setting where either all the parties get (many copies of) the same Haar state or they receive i.i.d Haar states. 

\paragraph{Separations.} We use~\Cref{sec:thm:locc:intro} to show that some quantum cryptographic primitives with classical communication are impossible in the CHS model. Let us develop some intuition towards proving such a statement. Suppose there are two or more parties participating in a quantum cryptographic protocol with classical communication in the CHS model. By definition, all the parties would receive many, say $t$, copies of $\ket{\psi}$, where $\ket{\psi}$ is sampled from the Haar distribution. Since the parties can only exchange classical messages, thanks to~\Cref{sec:thm:locc:intro}, without affecting correctness or security we can modify the protocol wherein for each party, say $P_i$, a Haar state $\ket{\psi_i}$ is sampled and $t$ copies of $\ket{\psi_i}$ is given to $P_i$. From this, we can extract a quantum cryptographic primitive in the plain model since each party can sample a Haar state on its own. In conclusion, quantum cryptographic primitives with classical communication in the CHS model can be turned into their counterparts in the plain model. 
\par This gives a natural recipe for proving impossibility results in the CHS model. We apply this recipe to obtain impossibility results for interactive key agreements and interactive commitments. 

\begin{theorem}
\label{thm:impossibility:intro}
Interactive quantum key agreement and interactive quantum commitment protocols, with classical communication, are impossible in the CHS model. 
\end{theorem}

\noindent We extend the above theorem to separate interactive quantum key agreement and interactive quantum commitments from pseudorandom function-like state generators. The separations are obtained by considering a variant of the CHS model where the adversary does not get access to many copies of one Haar state but instead gets access to infinitely many input-less oracles\footnote{We note that~\cite{Kretschmer21} made similar use of infinitely many oracles to prove a separation between pseudorandom states and one-way functions.} $\left\{\{G_{k,x}\}_{k,x \in \{0,1\}^{\secparam}}\right\}_{\secparam \in \mathbb{N}}$ such that each $G_{k,x}$ produces a copy of a Haar state $\ket{\psi_{k,x}}$. In this model, it is easy to construct pseudorandom function-like states. However, an extension of~\Cref{thm:impossibility:intro} rules out the possibility of interactive quantum key agreement and quantum commitments with classical communication in this variant. Thus, we have the following. 

\begin{theorem}
There does not exist a black-box reduction from interactive quantum key agreement and quantum commitments with classical communication to pseudorandom function-like states. 
\end{theorem}

\noindent Prior work by Chung, Goldin and Gray~\cite{CGG24} extensively studies the separations between quantum cryptographic primitives with classical communication and different quantum pseudorandomness notions. However, their framework did not capture the above result. 
\par Prior works by~\cite{ACC+22,CLM23,LLLL24} ruled out quantum key agreements and non-interactive commitments with classical communication from post-quantum one-way functions. However, their separation was either based on a conjecture or in a restricted setting whereas our result is unconditional. This makes our result incomparable with the results from~\cite{ACC+22,CLM23,LLLL24}. Our work follows a long line of recent works~\cite{HY20,ACC+22,AHY23,CLM23,afshar2023possibility,bouaziz2023towards,bouaziz2024quantum,coladangelo2024black} that make progress in understanding the landscape of black-box separations in quantum cryptography.

\section{Technical Overview}
\subsection{Pseudorandomness in the CHS Model}

\paragraph{Warmup: Pseudorandom State Generators (PRSGs).} As a warmup, we first study 1-copy PRSG in the CHS model. Consider the following construction: $G_k(\ket{\vartheta}) := (Z^{k}\otimes I_{n-\secp})\ket{\vartheta}$, where $Z^{k} = Z^{k_1} \otimes \cdots \otimes Z^{k_{\secparam}}$, $k= k_1 \cdots k_{\secparam} \in \{0,1\}^{\secparam}$ and $I_{n - \secparam}$ is an identity operator on $n-\secparam$ qubits. In other words, $G_k$ applies a random Pauli $Z$ operator only on the first $\secparam$ qubits and does not touch the rest. Note that this construction already satisfies the stretch property (i.e. the output length is larger than the key length). 
\par Let us consider the case when the adversary receives just one copy of $\ket{\vartheta}$ and is expected to distinguish $G_k(\ket{\vartheta})$ versus an independent Haar state $\ket{\varphi}$. Formally, we would like to argue that the following states are close. 
\[
\rho := \E_{\substack{k\gets\set{0,1}^{\secp} \\ \ket{\vartheta}\gets\Haar_{n}}}
\left[
G_k(\ket{\vartheta}) \otimes\ketbra{\vartheta}{\vartheta}
\right]
\text{ and }
\sigma := \frac{I}{2^{n}} \otimes \frac{I}{2^n}.
\]
By the properties of the symmetric subspace, the following holds:
$$\E_{ \ket{\vartheta}\gets\Haar_{n}}
\left[
\ketbra{\vartheta}{\vartheta}^{\otimes 2}
\right] \approx_{\varepsilon} \E_{x,y \leftarrow [2^n], x^1 \neq y^1} \left[ \frac{1}{2} \left( \ketbra{xy}{xy} + \ketbra{xy}{yx} + \ketbra{yx}{xy} + \ketbra{yx}{yx} \right) \right],$$ where $\varepsilon$ is negligible in $n$ and the notation $x^1$ (respectively, $y^1$) denotes the first $\secparam$ bits of $x$ (respectively, $y$). Now, applying a random $Z$ operator on the first $\secparam$ qubits tantamounts to measuring the first $\secparam$ qubits in the computational basis. Given the fact that $x^1 \neq y^1$, this measurement unentangles the last $n$ qubits. Thus, the result is a state of the form $\E_{x,y \leftarrow [2^n], x^1 \neq y^1} \left[ \frac{1}{2} \ketbra{x}{x} \otimes \ketbra{y}{y} + \frac{1}{2} \ketbra{y}{y} \otimes \ketbra{x}{x} \right]$. This state is in turn close to $\frac{I}{2^{n}} \otimes \frac{I}{2^n}$. \\

\noindent \textsc{Generalizing to Many Copies of the CHS.} 
Next, we to generalize the above approach to even when polynomially many copies of the CHS are provided. Formally, we would like to argue that the following two states are close. 
\[
\rho := \E_{\substack{k\gets\set{0,1}^{\secp} \\ \ket{\vartheta}\gets\Haar_{n}}}
\left[
G_k(\ket{\vartheta}) \otimes\ketbra{\vartheta}{\vartheta}^{\otimes t}
\right]
\text{ and }
\sigma := \Ex_{\substack{\ket{\varphi}\gets\Haar_n \\ \ket{\vartheta}\gets\Haar_{n}}}
\left[
\ketbra{\varphi}{\varphi}\otimes\ketbra{\vartheta}{\vartheta}^{\otimes t}
\right],
\]
where $t$ is some polynomial of $n$. Note that, by the  property of the Haar distribution, we can simplify $\sigma$ to $$\sigma = \frac{I}{2^{n}} \otimes \Ex_{T\leftarrow [0:t]^N} \ketbra{T}{T},$$ where $\ket{T}$ is a type state\footnote{We encourage readers unfamiliar with type states to refer to~\Cref{def:type_states}.} and $N=2^n$. Note that by the properties of the symmetric subspace, $$\E_{ \ket{\vartheta}\gets\Haar_{n}}
\left[
\ketbra{\vartheta}{\vartheta}^{\otimes t+1}
\right] \approx_{\varepsilon} \Ex_{\substack{T\leftarrow [0:t+1]^N\\ T\text{ is }\secparam\text{-prefix collision-free}}}\ketbra{T}{T},$$ where $\varepsilon$ is negligible in $n$ and $T$ is $\secparam$-prefix collision-free if $T\in\bit^{N}$ and for any $x,y\in T$\footnote{Since $T\in\bit^{N}$, we can treat it as a set, in particular the set associated to $T$ is $\set{i: T[i]=1}$.} with $x\neq y$ implies $x^1\neq y^1$, where the notation $x^1$ (respectively, $y^1$) denotes the first $\secparam$ bits of $x$ (respectively, $y$). Note that, any $\secparam$-prefix collision-free type $T$, $$\ket{T} = \frac{1}{\sqrt{{t+1 \choose t}}}\sum_{x\in T}\ket{x}\ket{T\setminus\set{x}}.$$ Again, applying a random $Z$ operator on the first $\secparam$ qubits tantamounts to measuring the first $\secparam$ qubits in the computational basis. Given the fact that $T$ is $\secparam$-prefix collision-free, this measurement unentangles the first $n$ qubits. Thus, the result is a state of the form $$\Ex_{\substack{T\leftarrow [0:t+1]^N\\ T\text{ is }\secparam\text{-prefix collision-free}\\ x\leftarrow T}}\left[\ketbra{x}{x}\otimes\ketbra{T\setminus\set{x}}{T\setminus\set{x}}\right].$$ This state is in turn close to $\frac{I}{2^{n}} \otimes \Ex_{T\leftarrow [0:t]^N} \ketbra{T}{T}$.\\
\\
\noindent \textsc{Generalizing to $\ell$-copy PRSG.}
Finally, we generalize this $\ell$-copy PRSG. Formally, we would like to argue that the following two states are close. 
\[
\rho := \E_{\substack{k\gets\set{0,1}^{\secp} \\ \ket{\vartheta}\gets\Haar_{n}}}
\left[
G_k(\ket{\vartheta})^{\otimes \ell} \otimes\ketbra{\vartheta}{\vartheta}^{\otimes t}
\right]
\text{ and }
\sigma := \Ex_{\substack{\ket{\varphi}\gets\Haar_n \\ \ket{\vartheta}\gets\Haar_{n}}}
\left[
\ketbra{\varphi}{\varphi}^{\otimes \ell}\otimes\ketbra{\vartheta}{\vartheta}^{\otimes t}
\right],
\]
where $\ell,t$ is some polynomial of $n$. Note that, by the  property of the Haar distribution, we can simplify $\sigma$ to $$\sigma = \Ex_{T_1\leftarrow [0:\ell]^N} \ketbra{T_1}{T_1} \otimes \Ex_{T_2\leftarrow [0:t]^N} \ketbra{T_2}{T_2},$$ where $\ket{T_1},\ket{T_2}$ are type states and $N=2^n$. Note that, similar to the  last case, we can still write, 
$$\E_{ \ket{\vartheta}\gets\Haar_{n}}\left[ \ketbra{\vartheta}{\vartheta}^{\otimes t+\ell} \right] \approx_{\varepsilon} \Ex_{\substack{T\leftarrow [0:t+\ell]^N\\ T\text{ is }\secparam\text{-prefix collision-free}}}\ketbra{T}{T},$$
and any $\secparam$-prefix collision-free type $T$,
$$\ket{T} = \frac{1}{\sqrt{{t+\ell \choose \ell}}}\sum_{\substack{T_1\subset T\\ |T_1|=\ell}}\ket{T_1}\ket{T\setminus T_1}.$$
Ideally, we would want the application of $(Z^k\otimes I_{n-\secparam})^{\otimes \ell}$ to unentangle $\ket{T_1}$ from  $\ket{T\setminus T_1}$. This is equivalent to measuring the first $\ell$ registers in the type basis. This is in general not true, not true. Hence, we settle for the next best thing, which is finding a ``dense-enough''\footnote{Here, by dense-enough, we mean when picking a random type from $\secparam$-prefix collision-free, it lies in this subset with probability $1-\negl$.} subset of $\secparam$-prefix collision-free type such that $(Z^k\otimes I_{n-\secparam})^{\otimes \ell}$ to unentangle $\ket{T_1}$ from $\ket{T\setminus T_1}$. We find this subset to be ``$\secparam$-prefix $\ell$-fold collision-free'' types. 

\par We say that a $\secparam$-prefix collision-free type $T$ is ``$\secparam$-prefix $\ell$-fold collision-free'' if for all pairs of $\ell$ sized subsets $T_1,T_2\subset T$, $\oplus_{x\in T_1} x=\oplus_{x\in T_2} x$ only if $T_1=T_2$. We start by noting that this subset is only ``dense-enough'' if $\ell = O\left( \frac{\secparam}{\log(\secparam)^{1+\varepsilon}} \right)$, for any constant $\varepsilon > 0$.\footnote{Later, in the impossibility result, we show that this is in fact the best we can hope for as a larger subset would bypass the impossibility result.} 
\par Next, we show that for these $\secparam$-prefix $\ell$-fold collision-free types states, applying a random $(Z^k\otimes I_{n-\secparam})^{\otimes \ell}$ is equivalent to meauring the first $\ell$ registers in the type basis. This is because $(Z^k\otimes I_{n-\secparam})^{\otimes \ell}$ on a type state $\ket{T_1}$ is equivalent to adding a phase of $(-1)^{k\cdot(\oplus_{x\in T_1}x)}$. Hence, $$\Ex_{k}\left[(Z^k\otimes I_{n-\secparam})^{\otimes \ell}\otimes I_{tn}\ketbra{T}{T}(Z^k\otimes I_{n-\secparam})^{\otimes \ell}\otimes I_{tn}\right] = \Ex_{k}\left[\frac{1}{{t+\ell \choose \ell}}\sum_{\substack{T_1,T_2\subset T\\ |T_1|=|T_2|=\ell}}(-1)^{k\cdot(\oplus_{x\in T_1} x \bigoplus \oplus_{y\in T_2}y)}\ket{T_1}\ket{T\setminus T_1}\bra{T_2}\bra{T\setminus T_2}\right],$$ which for $\secparam$-prefix $\ell$-fold collision-free types states is non-zero only if $T_1=T_2$, giving us 
$$\Ex_{k}\left[(Z^k\otimes I_{n-\secparam})^{\otimes \ell}\otimes I_{tn}\ketbra{T}{T}(Z^k\otimes I_{n-\secparam})^{\otimes \ell}\otimes I_{tn}\right] = \Ex_{\substack{T_1\subset T\\ |T_1|=\ell}}\left[\ketbra{T_1}{T_1}\otimes\ketbra{T\setminus T_1}{T\setminus T_1}\right].$$ Over expectation over all $\secparam$-prefix $\ell$-fold collision-free types states, this state is close to $\Ex_{T_1\leftarrow [0:\ell]^N} \ketbra{T_1}{T_1} \otimes \Ex_{T_2\leftarrow [0:t]^N} \ketbra{T_2}{T_2}.$

\paragraph{Limitations.} To complement our result, we show that a $t$-copy PRSG is impossible in the CHS model, for $\ell = O\left( \frac{\secparam}{\log(\secparam)} \right)$ (for a restricted class of PRSG constructs which only takes one copy of the common Haar state). We show this by showing that the rank of $\sigma$ grows much faster than the rank of $\rho$, hence, a simple distinguisher is a projector on the eigenspace of $\rho$. In particular, let $\tilde{G}_k(\vartheta)$ be the PRSG. Then define \[
\rho := \E_{\substack{k\gets\set{0,1}^{\secp} \\ \ket{\vartheta}\gets\Haar_{n}}}
\left[
\tilde{G}_k(\ket{\vartheta})^{\otimes \ell} \otimes\ketbra{\vartheta}{\vartheta}^{\otimes t}
\right]
\text{ and }
\sigma := \Ex_{\substack{\ket{\varphi}\gets\Haar_n \\ \ket{\vartheta}\gets\Haar_{n}}}
\left[
\ketbra{\varphi}{\varphi}^{\otimes \ell}\otimes\ketbra{\vartheta}{\vartheta}^{\otimes t}
\right]
\]
Now since $\tilde{G}_k(\ket{\vartheta})$ is a PRSG, its output is negligibly close to a pure state. This means that the rank of $\rho \leq 2^\secparam {2^n+t+\ell-1 \choose t+\ell}$. In contrast, the rank of $\sigma = {2^n+\ell-1 \choose \ell}{2^n+t-1 \choose t}$. Note that, for $t = \secp^3$ and $\ell = \secp/\log(\secp)$, $\rank(\rho)/\rank(\sigma) = \negl$. Hence, we can find a distinguisher. Here the distinguisher just projects onto the eigenspace of $\rho$, $\rho$ gets accepted with probability $1$  but $\sigma$ gets accepted with probability $\negl$, hence giving a disguisher. Since PRFSs imply PRSs (by setting $c = 0$), achieving an $\ell$-query statistical PRFS in the CHS model for $\ell = \Omega(\secp/\log(\secp))$ is impossible.

\paragraph{Pseudorandom Function-like State Generators.} Next we extend this idea from PRSGs to achieve PR\underline{F}SGs. We take inspiration from the seminal Goldreich-Goldwasser-Micali approach~\cite{GGM86}. In particular, on the key $K = (k_1^0,\ldots,k_m^0, k_1^1,\ldots,k_m^1)\in\bit^{2\secp' m}$ and the input $\bfx = (x_1,\ldots,x_m)\in\bit^{m}$, define the PRFSG $G_K(\bfx,\ket{\vartheta})$ as follows: $G_K(\bfx,\ket{\vartheta}) = (Z^{\bigoplus_{i=1}^m k^{x_i}_i}\otimes I_{n-\secp'}) \ket{\vartheta}.$ Formally, 
the following two states are close: 
\[
\rho := \E_{\substack{K\gets\set{0,1}^{2m\secp'} \\ \ket{\vartheta}\gets\Haar_{n}}}
\left[
\otimes_{i=1}^q G_K(\bfx^i,\ket{\vartheta})^{\otimes \ell_i}\otimes\ketbra{\vartheta}{\vartheta}^{\otimes t}
\right],
\]
and
\[
\sigma := \Ex_{\substack{\forall i\in[q], \ket{\varphi_i}\gets\Haar_n \\ \ket{\vartheta}\gets\Haar_{n}}}
\left[
\otimes_{i=1}^{q}\ketbra{\varphi_i}{\varphi_i}^{\otimes \ell_i}\otimes\ketbra{\vartheta}{\vartheta}^{\otimes t}
\right],
\]
for all $\bfx^1,\ldots,\bfx^q\in\bit^{m}$ and $\ell_1,\ldots,\ell_q$ such that $\sum_{i=1}^{q}\ell_i = \ell$, for $\ell=O\left( \frac{\secparam^{1-c}}{\log(\secparam)^{1+\varepsilon}} \right)$ and $m=\secparam^c$, for any constant $\varepsilon > 0$ and $c \in [0,1)$. 
\par Just as before, we can write $\sigma$ as follows: 
$$\sigma = \bigotimes^{q}_{i=1}\Ex_{T_i\leftarrow [0:\ell_i]^N} \ketbra{T_i}{T_i} \otimes \Ex_{\tilde{T}\leftarrow [0:t]^N} \ketbra{\tilde{T}}{\tilde{T}},$$ where $T_i$'s and $\tilde{T}$ are type states and $N=2^n$. Note that, similar to the  last case, we can still write, 
$$\E_{ \ket{\vartheta}\gets\Haar_{n}}\left[ \ketbra{\vartheta}{\vartheta}^{\otimes t+\ell} \right] \approx_{\varepsilon} \Ex_{\substack{T\leftarrow [0:t+\ell]^N\\ T\text{ is }\secparam\text{-prefix }\ell\text{-fold collision-free}}}\ketbra{T}{T},$$
and any $\secparam$-prefix $\ell$-fold collision-free type $T$,
$$\ket{T} = \frac{1}{\sqrt{{t+\ell \choose \ell}}}\sum_{\substack{T_1\subset T\\ |T_1|=\ell}}\ket{T_1}\ket{T\setminus T_1}.$$ Now, after application of one layer of $(Z^{k}\otimes I_{n-\secp})^{\otimes \ell}$, we know that $\ket{T_1}$ unentagles from $\ket{T\setminus T_1}$. We extend this idea to show that even for a tensor of type states, applying $(Z^{k}\otimes I_{n-\secp})^{\otimes \tilde{\ell_i}}$ on parts of each type state still unentangles each of them as long as all the type states are $\secparam$-prefix $\ell$-fold collision-free type and their combined set is still $\secparam$-prefix $\ell$-fold collision-free. Formally, we show the following: 
Let $\tilde{\ell_1},\ldots,\tilde{\ell_q}\in\N$, and $t_1,\ldots,t_q\in\N$ such that $\sum_{i=1}^{q} \tilde{\ell_i} = \tilde{\ell}$ and $\sum_{i=1}^{q} t_i = t$. Then for any $\secparam$-prefix $\tilde{\ell}$-fold collision-free type $T$ and any mutually disjoint sets $T_1,\ldots,T_q$ satisfying $\bigcup_{i=1}^{q} T_i = T$ and $|T_i|=t_i+\tilde{\ell_i}$ for all $i\in [q]$,
\begin{multline*}
\Ex_{k\gets\bit^{n}}\left[\bigotimes_{i=1}^{q}  \left( \left(Z^{k}\otimes I_m\right)^{\otimes \ell_i}\otimes I_{n+m}^{\otimes t_i} \right) \ketbra{T_i}{T_i} \left(\left(Z^{k}\otimes I_m\right)^{\otimes \tilde{\ell_i}}\otimes I_{n+m}^{\otimes t_i}\right) \right] \\
= \bigotimes_{i=1}^{q}\Ex_{\substack{X_i\subset T_i\\ |X_i|=\tilde{\ell_i}}} \left[ \ketbra{X_i}{X_i} \otimes \ketbra{T_i\setminus X_i}{T_i\setminus X_i} \right].
\end{multline*}

\par Hence, applying each layer $(Z^{k^b_i}\otimes I_{n-\secp})$ unentagles all type states into two halfs. Hence, by repeated application, we get 
$$\rho \approx_{\varepsilon} \Ex_{\substack{T\leftarrow [0:t+\ell]^N\\ T\text{ is }\secparam\text{-prefix }\ell\text{-fold collision-free}}}\Ex_{(T_1, T_2, \dots, T_q, \hat{T})}
\left[
\bigotimes_{i=1}^{q}\ketbra{T_i}{T_i}\otimes\ketbra{\hat{T}}{\hat{T}}
\right],
$$
where $(T_1, T_2, \dots, T_q, \hat{T})$ are sampled as follows: for $i = 1, 2, \dots, q$, sample an $\ell_i$-subset from $T\setminus(\bigcup_{j=1}^{i-1} T_j)$ uniformly and let $\hat{T}:= T \setminus (\bigcup_{j=1}^{q} T_j)$. Over expectation over all $\secparam$-prefix $\ell$-fold collision-free types states, this state is close to $\sigma$.

\subsection{Quantum Bit Commitments}
With $t$-copy PRSG in hand, we construct a statistically-hiding, statistically-binding commitment scheme in the CHS model. Our scheme draws inspiration from the quantum commitment scheme introduced in~\cite{MY21,MNY23} that builds quantum bit commitments from $t$-copy PRSG. 
\par In particular, to commit to $b=0$, the committer creates a superposition over all keys of the PRSG in the decommitment register and runs the PRSG in superposition over this register. The committer sets this as the commitment register. To commit to $b=1$, the committer creates a maximally entangled state over the commitment and the decommitment register. Formally, 
        \[
            \ket{\psi_0}_{\sfC_i\sfR_i} 
            := \frac{1}{\sqrt{2^\secp}} \sum_{k \in \bit^\secp} G_k(\ket{\vartheta})_{\sfC_i} \ket{k||0^{n-\secp}}_{\sfR_i}
        \]
        and
        \[
            \ket{\psi_1}_{\sfC_i\sfR_i} 
            := \frac{1}{\sqrt{2^n}} \sum_{j \in \bit^n} \ket{j}_{\sfC_i} \ket{j}_{\sfR_i},
        \]
where, $(\sfC_1,\ldots,\sfC_p)$ is the commitment register and $(\sfR_1,\ldots,\sfR_p)$ is the reveal register.

\par To achieve hiding, our scheme relies on the pseudorandomness property of the PRSG. In particular, the commitment is very close to one where the keys are distinct for all $(\sfC_i,\sfR_i)$, in this case, one copy of PRS is indistinguishable from a maximally mixed state.\footnote{Note that this still needs multi-key security which is not trivial in the CHS model, since all the PRS generators share the same Haar state for randomness. But we prove that our construction satisfies multikey security.}

\par Unlike the approach in~\cite{MY21}, our construction is not of the canonical form~\cite{Yan22}. To achieve binding, the receiver performs multiple SWAP tests. In particular, we show that since the rank of the commitment registers is exponentially separated, multiple SWAP tests can distinguish between the two. 

\subsection{Black-Box Separations}

\paragraph{LOCC Indistinguishability.}
The notion of LOCC indistinguishability is well-studied and is referred to as quantum data hiding by quantum information theorists~\cite{BDF+99, DLT02, EW02, GB02, HLS05, MWW09, CLMO13, PNC14, CH14, CLMOW14, HBAB19}. In this setting, there is a challenger, two (possibly entangled and mixed) bipartite quantum states $\rho_{\sfA\sfB}$ and $\sigma_{\sfA\sfB}$, and a computationally unbounded, two-party distinguisher (Alice, Bob) who are spatially separated and without pre-shared entanglement. The challenger picks a quantum state from $\set{\rho_{\sfA\sfB}, \sigma_{\sfA\sfB}}$ uniformly at random and sends register $\sfA$ to Alice and register $\sfB$ to Bob respectively. The task of Alice and Bob is to distinguish whether they are given $\rho_{\sfA\sfB}$ or $\sigma_{\sfA\sfB}$ by performing local operations and communicating classically. We call such distinguishers \emph{LOCC adversaries}.

We focus on the case where Alice and Bob each receive $t = \text{poly}(\secp)$ copies of $\ket{\psi}_\sfA$ and $\ket{\phi}_\sfB$, where $\ket{\psi}$ and $\ket{\phi}$ are either two identical or i.i.d. Haar states of length $n = \omega(\log(\secp))$. Explicitly, the two input states are
\[
\rho_{\sfA\sfB} = \Ex_{\ket{\psi} \gets \Haar_n} \left[ \ketbra{\psi}{\psi}^{\otimes t}_{\sfA} \otimes \ketbra{\psi}{\psi}_{\sfB}^{\otimes t} \right],
\]
\[
\sigma_{\sfA\sfB} = \Ex_{\ket{\psi} \gets \Haar_n} \left[ \ketbra{\psi}{\psi}^{\otimes t}_{\sfA} \right]
\otimes \Ex_{\ket{\phi} \gets \Haar_n} \left[ \ketbra{\phi}{\phi}^{\otimes t}_{\sfB} \right].
\]
Note that if global measurements are allowed, performing SWAP tests can easily distinguish them. As one of our main technical contributions, we show that for any LOCC adversary, the advantage of distinguishing $\rho_{\sfA\sfB}$ from $\sigma_{\sfA\sfB}$ is negligible in $\secp$. Before we explain the proof, we compare our theorem with~\cite[Theorem~8]{Har23}. In short, the theorems are incomparable. Our setting is stronger in the sense that the LOCC adversary both obtain polynomial copies of the input, while~\cite[Theorem~8]{Har23} studies the single-copy setting. However, \cite[Theorem~8]{Har23} is more general since it holds for a family of input states, whereas the input in our setting is fixed to $\rho_{\sfA\sfB}$ and $\sigma_{\sfA\sfB}$, which are belong to the family. We refer the readers to~\Cref{remark:Harrow} for a detailed discussion.

Toward the proof, we start by using the following common technique in proving LOCC indistinguishability: the set of LOCC measurements is a (proper) subset of the set of all positive partial transpose (PPT) measurements~\cite{CLMOW14}. Hence, it is sufficient to upper bound the maximum distinguishing advantage over two-outcome PPT measurements, \ie $\set{M_{\sfA\sfB}, I_{\sfA\sfB} - M_{\sfA\sfB}}$ such that $0\preceq M_{\sfA\sfB} \preceq I_{\sfA\sfB}$ and $0\preceq M_{\sfA\sfB}^{\Gamma_B} \preceq I_{\sfA\sfB}$, where $M_{\sfA\sfB}^{\Gamma_B}$ denote the partial transpose of $M_{\sfA\sfB}$ \wrt $\sfB$. Next, from the basic properties of partial transpose and trace norm, we show that the distinguishing advantage is bounded by the trace norm between $\rho_{\sfA\sfB}^{\Gamma_\sfB}$ and $\sigma_{\sfA\sfB}^{\Gamma_\sfB}$.

The most technical part of the proof is to upper bound the quantity $\norm{\rho_{\sfA\sfB}^{\Gamma_\sfB} - \sigma_{\sfA\sfB}^{\Gamma_\sfB}}_1$. We point out that the partial transpose of a density matrix might \emph{not} be a positive semidefinite matrix. Our first step is to expand $\rho_{\sfA\sfB}$ and $\sigma_{\sfA\sfB}$ in the \emph{type basis} as follows:
\[
\rho_{\sfA\sfB} = \Ex_{T\gets[0:2t]^d} \left[ \ketbra{T}{T}_{\sfA\sfB} \right],
\]
\[
\sigma_{\sfA\sfB} = \Ex_{S_A\gets[0:t]^d} \left[ \ketbra{S_A}{S_A}_\sfA \right] \otimes \Ex_{S_B\gets[0:t]^d} \left[ \ketbra{S_B}{S_B}_\sfB \right],
\]
where $d := 2^n$. Next, we further conditioned on the events that (1)~$T, S_A$ and $S_B$ each have no repeated elements (2)~$S_A$ and $S_B$ have no identical elements. From the collision bound, doing so only incurs an additional error of $O(t^2/d) = \negl(\secp)$. Therefore, we can now treat $T, S_A$ and $S_B$ as \emph{sets}. It suffices to prove that $\norm{\Tilde{\rho}_{\sfA\sfB}^{\Gamma_\sfB} - \Tilde{\sigma}_{\sfA\sfB}^{\Gamma_\sfB}}_1$ is negligible in $\secp$, where
\[
\Tilde{\rho}_{\sfA\sfB} := \Ex_{T\gets\binom{[d]}{2t}} \left[ \ketbra{T}{T}_{\sfA\sfB} \right],
\] 
\[
\Tilde{\sigma}_{\sfA\sfB} := \Ex_{\substack{S_A,S_B\gets\binom{[d]}{t}: \\ S_A \cap S_B = \emptyset}} \left[ \ketbra{S_A}{S_A}_\sfA \otimes \ketbra{S_B}{S_B}_\sfB \right].
\]
Observe that the $\Tilde{\sigma}^{\Gamma_B}_{\sfA\sfB} = \Tilde{\sigma}_{\sfA\sfB}$. To obtain a simpler expression of $\Tilde{\rho}^{\Gamma_B}_{\sfA\sfB}$, we rely on the following useful identity for bi-partitioning the type states:
\[
\ket{T}_{\sfA\sfB} 
= \sum_{X\in\binom{T}{t}} \frac{1}{\sqrt{\binom{2t}{t}}} \ket{T\setminus X}_\sfA \otimes  \ket{X}_\sfB.
\]
Hence, the partial transpose of $\Tilde{\rho}_{\sfA\sfB}$ can be written as
\[
\Tilde{\rho}^{\Gamma_B}_{\sfA\sfB}
= \Ex_{T\gets\binom{[d]}{2t}} \left[ \frac{1}{\binom{2t}{t}} \sum_{\substack{X,Y\in\binom{T}{t}}} \ketbra{T\setminus X}{T\setminus Y}_\sfA \otimes \ketbra{Y}{X}_\sfB \right].
\]
If $X = Y$, then the term is the tensor product of two \emph{disjoint} sets $\ketbra{T\setminus X}{T\setminus X}_\sfA \otimes \ketbra{X}{X}_\sfB$. Such a term will be canceled out by the corresponding term in $\Tilde{\sigma}^{\Gamma_B}_{\sfA\sfB}$ since they have equal coefficients. Therefore, the difference between them is the following matrix with mismatched $X$ and $Y$:
\[
\Tilde{\rho}^{\Gamma_B}_{\sfA\sfB} - \Tilde{\sigma}^{\Gamma_B}_{\sfA\sfB}
= \Ex_{T\gets\binom{[d]}{2t}} \left[ \frac{1}{\binom{2t}{t}} \sum_{\substack{X,Y\in\binom{T}{t}}: \\ X\neq Y} \ketbra{T\setminus X}{T\setminus Y}_\sfA \otimes \ketbra{Y}{X}_\sfB \right].
\]
We continue to simplify it by applying a double-counting argument. Every tuple of sets $(T,X,Y)$ uniquely determines a tuple of mutually disjoint sets $(C,I,X',Y')$ satisfying $C = T \setminus (X \cup Y)$ ($C$ for the complement of $X\cup Y$), $I = X \cap Y$ ($I$ for intersection), $X' = X \setminus I$ and $Y' = Y \setminus I$. Hence, $T \setminus X = C \uplus Y'$, $Y = I \uplus Y'$, $T \setminus Y = C \uplus X'$, and $X = I \uplus X'$ where $\uplus$ denotes the disjoint union. By further classifying the summands according to $s := |C| = |I| \in \set{0,1,\dots,t-1}$ (note that then $|X'| = |Y'| = t-s$), we have
\begin{align*}
& \norm{\Tilde{\rho}_{\sfA\sfB}^{\Gamma_\sfB} - \Tilde{\sigma}_{\sfA\sfB}^{\Gamma_\sfB}}_1 
=  \frac{1}{\binom{d}{2t}\binom{2t}{t}} \norm{ \sum_{s = 0}^{t-1} \sum_{C\in\binom{[d]}{s}} \sum_{I \in \binom{[d]\setminus C}{s}} 
 \sum_{\substack{X',Y' \in \binom{[d]\setminus (C\uplus I)}{t-s}: \\ X' \cap Y' = \emptyset}} \ket{C\uplus Y'}_\sfA \ket{I\uplus Y'}_\sfB \bra{C\uplus X'}_\sfA \bra{I\uplus X'}_\sfB }_1 \\
& \leq \frac{1}{\binom{d}{2t}\binom{2t}{t}} \sum_{s = 0}^{t-1} \sum_{C\in\binom{[d]}{s}} \sum_{I \in \binom{[d]\setminus C}{s}} 
\Bigg\| \underbrace{\sum_{\substack{X',Y' \in \binom{[d]\setminus (C\uplus I)}{t-s}: \\ X' \cap Y' = \emptyset}} \ket{C\uplus Y'}_\sfA \ket{I\uplus Y'}_\sfB \bra{C\uplus X'}_\sfA \bra{I\uplus X'}_\sfB}_{=: K_{C,I}} \Bigg\|_1,
\end{align*}
where the inequality follows from the triangle inequality. We observe that the matrix $K_{C,I}$ has the same structure as the adjacency matrix of \emph{Kneser graphs}. Here, we recall the definition of Kneser graphs. For $v,k\in\N$, the Kneser graph $K(v, k)$ is the graph whose vertices correspond to the $k$-element subsets of the set $[v]$, and two vertices are adjacent if and only if the two corresponding sets are disjoint. Therefore, for every $(C,I)$, the matrix $K_{C,I}$ is isospectral to the adjacency matrix of the Kneser graph $K(d-|C|-|I|, t-|I|)$. Finally, we employ the well-studied spectral property of Kneser graphs as a black box to obtain an $O(t^2/d) = \negl(\secp)$ upper bound for $\norm{\Tilde{\rho}_{\sfA\sfB}^{\Gamma_\sfB} - \Tilde{\sigma}_{\sfA\sfB}^{\Gamma_\sfB}}_1$.

Furthermore, we show the tightness of the theorem by constructing an optimal LOCC distinguisher that achieves the same advantage. The strategy is simple: Alice and Bob each individually measure every copy of their input in the computational basis and obtain a total of $2t$ outcomes. Then, they output $1$ if there is any collision among these $2t$ outcomes.

\paragraph{Impossibility Results in the CHS model.}

With the LOCC Haar indistinguishability theorem in hand, we investigate the limits of the CHS model when the communication between the parties is classical. We show that the several impossibility results of information-theoretically secure schemes in the plain model can be generically lifted to the CHS model, even when the adversary does not receive any common Haar state. We emphasize that there is no classical counterpart in the CRS model. If the adversary is not given the CRS, then many information-theoretically secure schemes exist, such as key agreements. 

As common in proving impossibilities, our approach is to convert schemes in the CHS model to those in the plain model. The transform is simple: in the new scheme, the parties each sample polynomially many copies of the Haar state \emph{independently} and run the original scheme. Crucially, despite the inconsistency in their Haar states, the new scheme still satisfies completeness thanks to the LOCC Haar indistinguishability. A caveat is that sampling Haar states is time-inefficient. However, since the impossibilities in the plain model are still valid if the (honest) algorithms in the scheme are time-inefficient, doing so is acceptable for the sake of showing impossibilities.

\paragraph{Separation Results.}  We separate many important primitives from $(\secp,\omega(\log(\secp)))$-PRSG.  Since $(\secp,\omega(\log(\secp)))$-PRSGs do not exist in the CHS model, we need to ``strengthen'' the oracle in order to prove separations. For every security parameter $\secp\in\N$, we define the oracle as $\set{G_k}_{k\in\bit^\secp}$ where each $G_k$ is an isometry that takes no input and outputs an i.i.d. Haar state $\ket{\psi_k}$. 

Relative to this oracle, the implementation of the PRSG is straightforward: the output on $k$ of any length $\secp\in\N$ is $\ket{\psi_k}$. The security directly follows from the hardness of unstructured search. To prove the non-existence of QCCC schemes, we employ a two step approach. First, showing that a scheme with respect to this oracle can be transformed to schemes with respect to a much weaker oracle. Second, showing that this much weaker oracle does not give much extra power over the plain model. Formally:  First, similar to the previous section, we show that due to the LOCC indistinguishability, the parties can sample all ``large'' quantum states on their own, and the correctness and security is only ``polynomially'' affected\footnote{Since the Haar indistinguishability has a factor of $O(t^2/d)$, as long as $t^2/d$ is inverse-polynomial, we do not incur a lot of loss.}. This means that any scheme with respect this this oracle can be turned into a scheme with respect to an oracle with only short (constant times logarithmic) Haar states. Second, for short (constant times logarithmic) quantum states, we show that this oracle does not give much extra power since an adversary can learn the oracle completely. This is because for short-enough states, the adversary can run tomography on polynomial queries and learn the state with up to inverse polynomial error. Hence, the adversary can simulate both parties post-selecting on a transcript to learn any secret\footnote{Note that since the adversary does not need to be efficient, as long as they have the description of this oracle, they can post-select on the transcript.}. This means that any scheme secure in the presence of this oracle can be transformed into another scheme that is secure in the plain model.  

Lastly, we observe that by considering a generalized oracle, namely $\set{G_{k,x}}_{k,x\in\bit^\secp}$, we can show that (classically accessible) PRFSGs with super-logarithmic input length exist. We can extend the impossibility of QCCC commitments to hold in the presence of the generalized oracle as well. Thus, we can separate PRFS and QCCC commitments. 
\section{Preliminaries}

We denote the security parameter by $\secp$. We assume that the reader is familiar with the fundamentals of quantum computing covered in~\cite{nielsen_chuang_2010}.

\subsection{Notation}
\begin{itemize}
\item We use $[n]$ to denote $\{1,\ldots,n\}$ and $[0:n]$ to denote $\{0,1,\ldots,n\}$.
\item For any finite set $T$ and any integer $0\leq k\leq |T|$, we denote by $\binom{T}{k}$ the set of all $k$-size subsets of $T$.
\item For any finite set $T$, we use the notation $x\gets T$ to indicate that $x$ is sampled uniformly from $T$. 
\item We denote by $S_t$ the symmetric group of degree $t$.
\item For any set $A$ and $t \in \N$, we denote by $A^t$ the $t$-fold Cartesian product of $A$. 
\item For $\sigma\in S_t$ and $\bfv = (v_1,\ldots,v_t)$, we define $\sigma(\bfv) := (v_{\sigma(1)},\ldots,v_{\sigma(t)})$.
\item We denote by ${\cal D}(H)$ the set of density matrices in the Hilbert space $H$. 
\item Let $\rho_{AB} \in \cD(H_A\otimes H_B)$, by $\Tr_{B}(\rho_{AB}) \in \cD(H_A)$ we denote the reduced density matrix by taking partial trace over $B$.
\item We denote by $\TD(\rho, \rho') := \frac{1}{2}\| \rho - \rho' \|_1$ the trace distance between quantum states $\rho, \rho'$, where $\norm{X}_1 = \Tr(\sqrt{X^\dagger X})$ denotes the trace norm.
\item For any matrices $A,B$, we write $A\preceq B$ to indicate that $B-A$ is positive semi-definite.
\item For any Hermitian matrix $O$, the trace norm of $O$ has the following variational definition: 
\[
\norm{O}_1 = \max_{-I\preceq M\preceq I} \Tr(MO).
\]
Furthermore, if $\Tr(O) = 0$ then $\norm{O}_1 = 2\cdot\max_{0 \preceq M\preceq I} \Tr(MO)$.
\item We denote the Haar measure over $n$ qubits by $\Haar_n$.
\item For any matrix $M_{\sfA\sfB} = \sum_{i,j,k,\ell} \alpha_{ijk\ell} \ketbra{i}{j}_\sfA \otimes \ketbra{k}{\ell}_\sfB$ on registers $(\sfA,\sfB)$, by $M_{\sfA\sfB}^{\Gamma_\sfB}$ we denote its \emph{partial transpose} \wrt register $\sfB$, \ie $M_{\sfA\sfB}^{\Gamma_\sfB} = \sum_{i,j,k,\ell} \alpha_{ijk\ell} \ketbra{i}{j}_\sfA \otimes \ketbra{\ell}{k}_\sfB$.\footnote{Note that the (partial) transpose operation needs to be defined \wrt to an orthogonal basis. Throughout this work, it is always defined \wrt to the computational basis.}
\end{itemize}

\subsection{Common Haar State Model}
The Common Haar State (CHS) model is related to the Common Reference Quantum State (CRQS) model~\cite{MNY23}. In this model, all parties receive polynomially many copies of a \emph{single} quantum state sampled from the Haar distribution. Recently, another work of Chen et.al.~\cite{chen2024power} studied a similar model called the Common Haar Random State (CHRS) model. In the CHRS model, every party receives polynomially many copies of \emph{polynomially many} i.i.d. Haar states.\\

\noindent We define another variant of the CHS model called the \emph{Keyed} Common Haar State Model. In this model, all parties (once the security parameter is set to $\secp$) have access to  the oracle (called the \emph{Keyed} Common Haar State Oracle) $G^{\secp} := \set{G_k}_{k \in \bit^\secp}$ as follows. For every $k\in\bit^\secp$, the oracle $G_k$ is a Haar isometry that maps any state $\ket{\psi}$ to $\ket{\psi}\ket{\vartheta_k}$, where $\ket{\vartheta_k}$ is a Haar state of length $n(\secp) = \omega(\log(\secp))$.\\

\noindent While the above variant is harder to instantiate (hence not useful for constructions), is a natural candidate for black-box separations as seen is~\Cref{sec:QBB_QCCC}.
 
\subsubsection{Pseudorandom State (PRS) Generators in the CHS model}

\begin{definition}[Statistically secure $(\secparam,n,\ell)$-pseudorandom state generators in the CHS model]
We say that a QPT algorithm $G$ is a \emph{statistically secure $(\secparam,n,\ell)$-pseudorandom state generator (PRSG)} in the CHS model if the following holds:
    \begin{itemize}
        \item \textbf{State Generation}:
        For any $\secp\in\N$ and $k\in\set{0,1}^{\secp}$, the algorithm $G_k$ (where $G_k$ denotes $G(k,\cdot)$) is a quantum channel such that for every $n(\secp)$-qubit state $\ket{\vartheta}$,
        \[
        G_k(\ketbra{\vartheta}{\vartheta}) = \ketbra{\vartheta_k}{\vartheta_k},
        \]
        for some $n(\secp)$-qubit state $\ket{\vartheta_k}$. We sometimes write $G_k(\ket{\vartheta})$ for brevity.\footnote{More generally, the generation algorithm could take multiple copies of the common Haar state as input or output a state of different size compared to the common Haar state. Here, we focus on a restricted class of generators that only require a single copy of the common Haar state as input, and the output of the generator matches the size of the common Haar states.}
        \item \textbf{$\ell$-copy Pseudorandomness}: 
        For any polynomial $t(\cdot)$ and any non-uniform, unbounded adversary $\alice = \set{\alice_\secp}_{\secp\in\N}$, there exists a negligible function $\negl(\cdot)$ such that: 
        \begin{multline*}
        \Bigg|
        \Pr_{\substack{k\gets\set{0,1}^{\secp} \\ \ket{\vartheta}\gets\Haar_{n(\secp)}}}
        \left[ \alice_{\secp}\left(G_k(\ket{\vartheta})^{\otimes \ell(\secp)} \otimes \ketbra{\vartheta}{\vartheta}^{\otimes t(\secp)} \right) = 1 \right] \\
        - \Pr_{\substack{\ket{\varphi}\gets\Haar_{n(\secp)} \\ \ket{\vartheta}\gets\Haar_{n(\secp)}}}
        \left[ \alice_{\secp}\left(\ketbra{\varphi}{\varphi}^{\otimes \ell(\secp)} \otimes \ketbra{\vartheta}{\vartheta}^{\otimes t(\secp)} \right) = 1 \right]
        \Bigg|
        \leq \negl(\secp).
        \end{multline*}
    \end{itemize}
\noindent If $G$ satisfies $\ell$-copy pseudorandomness for every polynomial $\ell(\cdot)$ then we drop $\ell$ from the notation and simply denote it to be a $(\secparam,n)$-PRSG. 
\end{definition}
\noindent We define a stronger definition below called \emph{multi-key $\ell$-copy PRS generators}. Looking ahead, our construction of PRS in~\Cref{sec:prs-con} satisfies this definition.

\begin{definition}[Multi-key statistically secure $(\secparam,n,\ell)$-pseudorandom state generators in the CHS model]
We say that a QPT algorithm $G$ is a \emph{multi-key statistically secure $(\secparam,n,\ell)$-pseudorandom state generator} in the CHS model if the following holds:
    \begin{itemize}
        \item \textbf{State Generation}:
        For any $\secp\in\N$ and $k\in\set{0,1}^{\secp}$, the algorithm $G_k$ (where $G_k$ denotes $G(k,\cdot)$) is a quantum channel such that for every $n(\secp)$-qubit state $\ket{\vartheta}$,
        \[
        G_k(\ketbra{\vartheta}{\vartheta}) = \ketbra{\vartheta_k}{\vartheta_k},
        \]
        for some $n(\secp)$-qubit state $\ket{\vartheta_k}$. We sometimes write $G_k(\ket{\vartheta})$ for brevity.
        \item \textbf{Multi-key $\ell$-copy Pseudorandomness}: For any polynomial $t(\cdot)$, $p(\cdot)$ and any non-uniform, unbounded adversary $\alice = \set{\alice_\secp}_{\secp\in\N}$, there exists a negligible function $\negl(\cdot)$ such that: 
        \begin{multline*}
        \Bigg|
        \Pr_{\substack{k_1,\ldots,k_{p(\secp)}\gets\set{0,1}^{\secp}\\ \ket{\vartheta}\gets\Haar_{n(\secp)}}}
        \left[ \alice_{\secp}\left(\bigotimes_{i=1}^{p(\secp)} G_{k_i}(\ket{\vartheta})^{\otimes \ell(\secp)} \otimes \ketbra{\vartheta}{\vartheta}^{\otimes t(\secp)} \right) = 1 \right] \\
        - \Pr_{\substack{\ket{\varphi_1},\ldots,\ket{\varphi_{p(\secp)}}\gets\Haar_{n(\secp)}\\ \ket{\vartheta}\gets\Haar_{n(\secp)}}}
        \left[ \alice_{\secp}\left(\bigotimes_{i=1}^{p(\secp)}\ketbra{\varphi_i}{\varphi_i}^{\otimes \ell(\secp)} \otimes \ketbra{\vartheta}{\vartheta}^{\otimes t(\secp)} \right) = 1 \right]
        \Bigg|
        \leq \negl(\secp).
        \end{multline*}
    \end{itemize}
\noindent If $G$ satisfies multi-key $\ell$-copy pseudorandomness for every polynomial $\ell(\cdot)$ then we drop $\ell$ from the notation and simply denote it to be a multi-key $(\secparam,n)$-PRSG. 
\end{definition}

\begin{remark}
Note that in the plain model, PRS implies multi-key PRS because the pseudorandom state generator does not share randomness for different keys. It is not clear whether this holds in the CHS model as the different executions of the pseudorandom state generator share the same common Haar state.
\end{remark}

\subsubsection{Pseudorandom Function-Like State (PRFS) Generators in the CHS model}

\begin{definition}[Statistical selectively secure $(\secparam,m,n,\ell)$-PRFS generators]
We say that a QPT algorithm $G$ is a \emph{statistical selectively secure $(\secparam,m,n,\ell)$-PRFS generator} in the CHS model if the following holds:
\begin{itemize}
\item \textbf{State Generation}:
For any $\secp\in\N$, $k\in\set{0,1}^{\secp}$ and $x\in\bit^{m(\secp)}$, where $m(\secp)$ is the input length, the algorithm $G_{k,x}$ (where $G_{k,x}$ denotes $G(k,x,\cdot)$) is a quantum channel such that for every $n(\secp)$-qubit state $\ket{\vartheta}$,
\[
G_{k,x}(\ketbra{\vartheta}{\vartheta}) = \ketbra{\vartheta_{k,x}}{\vartheta_{k,x}},
\]
for some $n(\secp)$-qubit state $\ket{\vartheta_{k,x}}$. We sometimes write $G_{k,x}(\ket{\vartheta})$ or $G_k(x,\ket{\vartheta})$ for brevity.
\item \textbf{$\ell$-query Selective Security}:
For any polynomial $t(\cdot)$, any non-uniform, unbounded adversary $\alice = \set{\alice_\secp}_{\secp\in\N}$, and any tuple of (possibly repeated) $m(\secp)$-bit indices 
$(x_1, \dots, x_{\ell(\secp)})$, there exists a negligible function $\negl(\cdot)$ such that for all $\secp\in\N$,
\begin{multline*}
\Bigg\lvert
\Pr_{k\gets\bit^\secp, \ket{\vartheta}\gets\Haar_{n(\secp)}}
\left[
A_\secp \left( x_1,\dots,x_{\ell(\secp)}, \bigotimes_{i = 1}^{\ell(\secp)} G(k,x_i,\ket{\vartheta}) \otimes \ketbra{\vartheta}{\vartheta}^{\otimes t(\secp)} \right) = 1
\right] \\
- 
\Pr_{\substack{\forall x \in \bit^{m(\secp)},\ \ket{\varphi_x}\gets\Haar_{n(\secp)}, \\ 
\ket{\vartheta}\gets\Haar_{n(\secp)}}}
\left[
A_\secp \left( x_1,\dots,x_{\ell(\secp)},\bigotimes_{i = 1}^{\ell(\secp)} \ketbra{\varphi_{x_i}}{\varphi_{x_i}} \otimes \ketbra{\vartheta}{\vartheta}^{\otimes t(\secp)} \right) = 1
\right]
\Bigg\rvert
\leq \negl(\secp).
\end{multline*}
\end{itemize}
\noindent If $G$ satisfies $\ell$-query selective security for every polynomial $\ell(\cdot)$, we drop $\ell$ from the notation and say that $G$ is a $(\secparam,m,n)$-PRFS generator. 
\end{definition}

\subsubsection{Quantum Commitments in the CHS model}
\begin{definition}[Quantum commitments in the CHS model]
A (non-interactive) quantum commitment scheme in the CHS model is given by a tuple of the committer $C$ and receiver $R$ parameterized by a polynomial $p(\cdot)$, both of which are uniform QPT algorithms. Let $\ket{\vartheta}$ be the $n(\secp)$-qubit common Haar state. The scheme is divided into two phases: the commit phase, and the reveal phase as follows:
    \begin{itemize}
        \item Commit phase: $C$ takes $\ket{\vartheta}^{\otimes p(\secp)}$ and a bit $b \in \set{0, 1}$ to commit as input, generates a quantum state on registers $\sfC$ and $\sfR$, and sends the register $\sfC$ to $R$.
        \item Reveal phase: $C$ sends $b$ and the register $\sfR$ to $R$. $R$ takes $\ket{\vartheta}^{\otimes p(\secp)}$ and $(b, \sfC, \sfR)$ given by $C$ as input, and outputs $b$ if it accepts and otherwise outputs $\bot$.
    \end{itemize}
\end{definition}

\begin{definition}[Poly-copy statistical hiding]
A quantum commitment scheme $(C, R)$ in the CHS model satisfies \emph{poly-copy statistical hiding} if for any non-uniform, unbounded malicious receiver $R^* = \set{R^*_\secp}_{\secp\in\N}$, and any polynomial $t(\cdot)$, there exists a negligible function $\negl(\cdot)$ such that
\begin{multline*}
\Bigg\vert \Pr[R^*_\secp(\ket{\vartheta}^{\otimes t(\secp)},\Tr_{\sfR}(\sigma_{\sfC\sfR})) = 1:
\substack{\ket{\vartheta}\gets \Haar_{n(\secp)},\\ \sigma_{\sfC\sfR}\gets C_{\com}(\ket{\vartheta}^{\otimes p(\secp)},0)}] \\
- \Pr[R^*_\secp(\ket{\vartheta}^{\otimes t(\secp)},\Tr_{\sfR}(\sigma_{\sfC\sfR})) = 1
:\substack{\ket{\vartheta}\gets \Haar_{n(\secp)}, \\ \sigma_{\sfC\sfR}\gets C_{\com}(\ket{\vartheta}^{\otimes p(\secp)},1)}] \Bigg\vert
\leq \negl(\secp),
\end{multline*}
where $C_{\com}$ is the commit phase of $C$.
\end{definition}

\begin{definition}[Statistical sum-binding] A quantum commitment scheme $(C, R)$ in the CHS model satisfies \emph{statistical sum-binding} if the following holds. For any pair of non-uniform, unbounded malicious senders $C_0^{\ast}$ and $C_1^{\ast}$ that take $\ket{\vartheta}^{\otimes T(\secp)}$ for arbitrary large $T(\cdot)$ as input and work in the same way in the commit phase, if we let $p_b$ to be the probability that $R$ accepts the revealed bit $b$ in the interaction with $C_b^{\ast}$ for $b \in \set{0, 1}$, then we have
$$p_0+p_1\leq 1+\negl(\secp).$$
\end{definition}

\subsection{Symmetric Subspaces, Type States, and  Haar States}

\noindent The proofs of facts and lemmas stated in this subsection can be found in~\cite{Harrow13church}. Let $\bfv = (v_1,\ldots,v_t)\in A^t$ for some finite set $A$. Let $|A| = N$. Define $\type(\bfv)\in [0:t]^N$ to be the \emph{type vector} such that the $i^{th}$ entry of $\type(\bfv)$ equals the number of occurrences of $i\in[N]$ in $\bfv$.\footnote{We identify $[0:t]^N$ as $[0:t]^A$.} In this work, by $T\in [0:t]^N$ we implicitly assume that $\sum_{i\in[N]}T_i = t$. 

For $T\in[0:t]^N$, we denote by $\setT(T)$ the \emph{multiset} uniquely determined by $T$. That is, the multiplicity of $i\in\setT(T)$ equals $T_i$ for all $i\in[N]$.
We write $T\gets [0:t]^N$ to mean sampling $T$ uniformly from $[0:t]^N$ conditioned on $\sum_{i\in[N]}T_i = t$. We write $\bfv\in T$ to mean $\bfv \in A^t$ satisfies $\type(\bfv) = T$.

In this work, we will focus on \emph{collision-free} types $T$ which satisfy $T_i\in\bit$ for all $i\in[N]$. A collision-free type $T$ can be naturally treated as a \emph{set} and we write $\bfv\gets T$ to mean sampling a uniform $\bfv$ conditioned on $\type(\bfv) = T$.

\begin{definition}[Type states] \label{def:type_states}
Let $T\in[0:t]^N$, we define the \emph{type states}:
\[
\ket{T} := \sqrt{\frac{\prod_{i\in[N]} T_i !}{t!}} \sum_{\bfv \in T} \ket{\bfv}.
\]
If $T$ is collision-free, then it can be simplified to 
\[
\ket{T} = \frac{1}{\sqrt{t!}} \sum_{\bfv \in T} \ket{\bfv}.
\]
Furthermore, it has the following useful expression
\begin{equation} \label{eq:useful}
\ketbra{T}{T} 
= \frac{1}{t!} \sum_{\bfv, \bfu \in T} \ketbra{\bfv}{\bfu}
= \Ex_{\bfv\gets T}\left[ \sum_{\sigma\in S_t}\ketbra{\bfv}{\sigma(\bfv)} \right].
\end{equation}
\end{definition} 

\begin{lemma}[Average of copies of Haar-random states]
\label{fact:avg-haar-random}
For all $N,t \in \N$, we have
\[
\E_{\ket{\vartheta} \leftarrow \Haar(\C^N)} \ketbra{\vartheta}{\vartheta}^{\otimes t}
= \E_{T \leftarrow [0:t]^N} \ketbra{T}{T}.
\]
\end{lemma}

\subsection{Quantum Black-Box Reductions}
We recall the definition of fully black-box reductions~\cite{RTV04,BBF13} and their quantum analogue. The definitions below are taken verbatim from~\cite{HY20}.

\begin{definition}[Quantum primitives]
A quantum primitive $\cP$ is a pair $(\cF_\cP, \cR_\cP)$, where $\cF_\cP$ is a set of quantum algorithms $\cI$, and $\cR_\cP$ is a relation over pairs $(\cI, \cA)$ of quantum algorithms $\cI \in \cF_\cP$ and $\cA$. A quantum algorithm $\cI$ implements $\cP$ or is an implementation of $\cP$ if $\cI \in \cF_\cP$. If $\cI \in \cF_\cP$ is efficient, then $\cI$ is an efficient implementation of $\cP$. A quantum algorithm $\cA$ $\cP$-breaks $\cI \in \cF_\cP$ if $(\cI, \cA) \in \cR_\cP$. A secure implementation of $\cP$ is an implementation $\cI$ of $\cP$ such that no efficient quantum algorithm $\cP$-breaks $\cI$. The primitive $\cP$ quantumly exists if there exists an efficient and secure implementation of $\cP$.
\end{definition}

\begin{definition}[Quantum primitives relative to oracle]
Let $\cP = (\cF_\cP, \cR_\cP)$ be a quantum primitive, and $O$ be a quantum oracle. An oracle quantum algorithm $\cI$ implements $\cP$ relative to $O$ or is an implementation of $\cP$ relative to $O$ if $\cI^O \in \cF_\cP$. If $\cI^O \in \cF_\cP$ is efficient, then $\cI$ is an efficient implementation of $\cP$ relative to $O$. A quantum algorithm $\cA$ $\cP$-breaks $\cI \in \cF_\cP$ relative to $O$ if $(I^O, \cA^O) \in \cR_\cP$. A secure implementation of $\cP$ is an implementation $\cI$ of $\cP$ relative to $O$ such that no efficient quantum algorithm $\cP$-breaks $\cI$ relative to $O$. The primitive $\cP$ quantumly exists relative to $O$ if there exists an efficient and secure implementation of $\cP$ relative to $O$.
\end{definition}

\begin{definition}[Quantum fully black-box reductions]
A pair $(C, S)$ of efficient oracle quantum algorithms is a \emph{quantum fully-black-box reduction} from a quantum primitive $\cP = (\cF_\cP, \cR_\cP)$ to a quantum primitive $\cQ = (\cF_\cQ, \cR_\cQ)$ if the following two conditions are satisfied:
\begin{enumerate}
\item \textbf{(Correctness.)} For every implementation $\cI \in \cF_\cQ$, we have $C^\cI \in \cF_\cP$.
\item \textbf{(Security.)} For every implementation $\cI \in \cF_\cQ$ and every quantum algorithm $\cA$, if $\cA$ $\cP$-breaks $C^\cI$, then $S^{\cA,\cI}$ $\cQ$-breaks $\cI$.
\end{enumerate}
\end{definition}

\section{Warmup: Statistical Stretch PRS Generators in the CHS model} \label{sec:PRS_CHS}

\noindent We present a construction of multi-key PRS generator with statistical security in the CHS model.  

\begin{theorem} \label{thm:PRS_CHS}
There exists a multi-key $(\secparam,n,\ell)$-statistical PRS generator in the CHS model, where $n \geq \secparam$ and $\ell = O(\secp/\log(\secp)^{1+\eps})$ for any constant $\eps > 0$. 
\end{theorem}

\noindent The proof can be found in~\Cref{sec:prs-con}. Later, we prove the optimality of our construction in~\Cref{sec:prs-imp}. Specifically, we show that any $(\secparam,n,\ell)$-statistical PRS generator cannot simultaneously satisfy $n=\omega(\log(\secparam))$ and $\ell=\Omega(\secp/\log(\secp))$. 

\subsection{Useful Lemmas}

At a high level, the proof follows the template of~\cite{AGQY22,AGKL}: we do the analysis in the symmetric subspace. First, we identify a nice property of type vectors such that (1) a randomly sampled type satisfies this property with overwhelming probability and (2) the PRS generation algorithm behaves well on every type state having this property. We identify these type vectors as \emph{$\ell$-fold collision-free} types (which are a generalization of distinct types~\cite{AGQY22,AGKL}). 
\begin{definition}[$\ell$-fold $n$-prefix collision-free types]
\label{def:l_fold_collisions}
Let $n,m,t,\ell\in\N$ such that $t \geq \ell$ and $T \in [0:t]^{2^{n+m}}$ is a type vector. We say that $T$ is \emph{$\ell$-fold $n$-prefix collision-free} if for all pairs of $\ell$-subsets\footnote{Here we allow the subsets to contain duplicate elements.} $\cS,\cT \subseteq \setT(T)$, the first $n$ bits of $\bigoplus_{x\in\cS} x \in \bit^{n+m}$ is identical to that of $\bigoplus_{y\in\cT} y \in \bit^{n+m}$ if and only if $\cS = \cT$. We define $\goodpre{n}{m}{\ell}{t} := \set{ T \in [0:t]^{2^{n+m}}: T \text{ is $\ell$-fold $n$-prefix collision-free} }$ as the set of all $\ell$-fold $n$-prefix collision-free type vectors.
\end{definition}

\noindent When $t > \ell$, one can easily verify that $\ell$-fold $n$-prefix collision-freeness implies the standard collision-freeness. Also note that when $t > 2\ell$, $\ell$-fold $n$-prefix collision-freeness implies $i$-fold $n$-prefix collision-freeness for all $i\leq \ell$.

Next, we show that a random type is $\ell$-fold $n$-prefix collision-free with high probability.

\begin{lemma} \label{fact:random_type_l_fold_collision}
$\Pr_{T\gets [0:t]^{2^{n+m}}}[ T \in \goodpre{n}{m}{\ell}{t}] 
= 1- O(t^{2\ell}/(2^n-2\ell))$.
\end{lemma}

\begin{proof}
First, sampling $T\gets [0:t]^{2^{n+m}}$ uniformly is $O(t^2/2^{n+m})$-close to sampling a uniform collision-free $T$ from $[0:t]^{2^{n+m}}$ by the collision bound. \\

\noindent Furthermore, sampling a uniform collision-free $T$ from $[0:t]^{2^{n+m}}$ is equivalent to sampling $t$ elements $x_1,x_2,\dots,x_t$ one by one from $\bit^{n+m}$ conditioned on them being distinct and setting $T$ such that $\setT(T) = \set{x_1,\ldots,x_t}$. Hence, it suffices to show that sampling $t$ elements $x_1,x_2,\dots,x_t$ one by one from $\bit^{n+m}$ conditioned on them being distinct results in an $\ell$-fold $n$-prefix collision-free set with probability $1- O(t^{2\ell}/2^n)$. \\

\noindent For any two distinct $\ell$-subsets of indices $\cS\neq\cT \subseteq [t]$, let $\bad_{\cS,\cT}$ denote the event that the first $n$ bits of $\bigoplus_{i\in\cS} x_i$ is the same as that of $\bigoplus_{j\in\cT}x_j$. Then the following holds:

$$\Pr[\bad_{\cS,\cT}: \substack{x_1,x_2,\dots,x_t\gets \bit^{n+m}\\ x_1,x_2,\dots,x_t\text{ are distinct}}] = O(1/(2^n-2\ell)).$$

\noindent This is because we can first sample $|\cS\cup\cT|-1$ elements (in $\cS\cup\cT$) except one with indices in $\cS\setminus\cT$. Then $\bad_{\cS,\cT}$ occurs only if the first $n$ bits of the last sample is equal to the first $n$ bits of the bitwise XOR of all other elements in $\cS$ with all elements in $\cT$, which happens with probability at most $O(1/(2^n-2\ell))$. \\

\noindent By a union bound, we have $T \in \goodpre{n}{m}{\ell}{t}$ with probability at least $1 - \left(O(t^2/2^{n+m}) + \binom{t}{\ell}^2\cdot O(1/(2^n-2\ell))\right) = 1 - O(t^{2\ell}/(2^n-2\ell))$.
\end{proof}

\noindent Finally, the following two lemmas show that applying random Pauli-$Z$ on any $\ell$-fold $n$-prefix collision-free type state is equivalent to a ``classical'' probabilistic process\footnote{We say that this is a ``classical'' probabilistic process because we can write the resulting density matrix as direct sum of matrices with classical descriptions with weights chosen by a completely classical process. This means that we can simualte this process by first doing a completely classical sampling process followed by a state preparation.}. 
\begin{lemma} \label{lem:perm_split}
For any $\bfv\in\bit^{(n+m)(t+\ell)}$ such that $\type(\bfv) \in \goodpre{n}{m}{\ell}{t+\ell}$ and $\sigma\in S_{t+\ell}$, define 
\[
A_{\bfv,\sigma} 
:= \E_{k\gets\bit^{n}}\left[ \left(\left(Z^{k}\otimes I_m\right)^{\otimes \ell}\otimes I_{n+m}^{\otimes t}\right)\ketbra{\bfv}{\sigma(\bfv)}\left(\left(Z^{k}\otimes I_m\right)^{\otimes \ell}\otimes I_{n+m}^{\otimes t}\right) \right].
\]
Then $A_{\bfv,\sigma} = \ketbra{\bfv}{\sigma(\bfv)}$ if $\sigma$ maps $[\ell]$ to $[\ell]$; otherwise, $A_{\bfv,\sigma} = 0$. 
\end{lemma}
\begin{proof}
Suppose $\bfv = (v_1||w_1,\ldots,v_{t+\ell}||w_{t+\ell})\in\bit^{(n+m)(t+\ell)}$ with $v_i\in\bit^n$ and $w_i\in\bit^m$ for all $i\in[t]$. First, a direct calculation yields:
\[
\left(\left(Z^{k}\otimes I_m\right)^{\otimes \ell}\otimes I_{n+m}^{\otimes t}\right)\ketbra{\bfv}{\sigma(\bfv)}\left(\left(Z^{k}\otimes I_m\right)^{\otimes \ell}\otimes I_{n+m}^{\otimes t}\right) 
= (-1)^{\langle k,\bigoplus_{i=1}^{\ell} (v_i\oplus v_{\sigma(i)})\rangle}\ketbra{\bfv}{\sigma(\bfv)}.
\]
Therefore, after averaging over $k$,
\[
A_{\bfv,\sigma} 
= \E_{k\gets\bit^{n}}\left[(-1)^{\langle k,\bigoplus_{i=1}^{\ell} (v_i\oplus v_{\sigma(i)})\rangle}\right] \ketbra{\bfv}{\sigma(\bfv)}
= \begin{cases}
\ketbra{\bfv}{\sigma(\bfv)} & \text{ if } \bigoplus_{i=1}^{\ell} (v_i\oplus v_{\sigma(i)}) = 0 \\
0 & \text{ otherwise.}
\end{cases}
\]
Since $\type(\bfv) \in \goodpre{n}{m}{\ell}{t+\ell}$, the condition $\bigoplus_{i=1}^{\ell} v_i  = \bigoplus_{i=1}^\ell v_{\sigma(i)}$ holds if and only if the two sets $\set{1,2,\dots,\ell}$ and $\set{\sigma(1),\sigma(2),\dots,\sigma(\ell)}$ are identical. 
\end{proof}

\noindent The following lemma lies at the technical heart of this section. It states that the action of applying random $Z^k$ on $\ell$-fold $n$-prefix collision-free types $T$\footnote{Since $T$ is collision-free, we will treat it as a set.} has the following ``classical'' probabilistic interpretation: the output is identically distributed to first uniformly sampling an $\ell$-subset $X$ from $T$ and then generating $\ketbra{X}{X} \otimes \ketbra{T\setminus X}{T\setminus X}$.
\begin{lemma} \label{lem:nice_T}
For any $T \in \goodpre{n}{m}{\ell}{t+\ell}$,
\[
\Ex_{k\gets\bit^{n}}\left[ \left( \left(Z^{k}\otimes I_m\right)^{\otimes \ell}\otimes I_{n+m}^{\otimes t} \right) \ketbra{T}{T} \left(\left(Z^{k}\otimes I_m\right)^{\otimes \ell}\otimes I_{n+m}^{\otimes t}\right) \right] 
= \Ex_{X\gets \binom{T}{\ell}} \left[ \ketbra{X}{X} \otimes \ketbra{T\setminus X}{T\setminus X} \right].
\]
\end{lemma}
\begin{proof}
We first use the expression in~\Cref{eq:useful} on the left-hand side:
\begin{align} \label{eq:split}
LHS 
= \Ex_{\bfv\gets T} \left[ \sum_{\sigma\in S_t} \Ex_{k\gets\bit^{n}}\left[ \left( \left(Z^{k}\otimes I_m\right)^{\otimes \ell}\otimes I_{n+m}^{\otimes t} \right) \ketbra{\bfv}{\sigma(\bfv)} \left( \left(Z^{k}\otimes I_m\right)^{\otimes \ell}\otimes I_{n+m}^{\otimes t} \right) \right] \right].
\end{align}
Then from the previous lemma (\Cref{lem:perm_split})
\begin{align*}
\eqref{eq:split}
& = \Ex_{\bfv\gets T} \left[ \sum_{\sigma_1 \in S_\ell, \sigma_2 \in S_{t}} \ketbra{\bfv}{\sigma_1\circ \sigma_2(\bfv)} \right] \\
& = \Ex_{\bfv\gets T} \left[ \sum_{\sigma_1 \in S_\ell} \ketbra{\bfv_{[1:\ell]}}{\sigma_1(\bfv_{[1:\ell]})} \otimes \sum_{\sigma_2 \in S_{t}} \ketbra{\bfv_{[\ell+1:\ell+t]}}{\sigma_2(\bfv_{[\ell+1:\ell+t]})} \right] \\
& = \Ex \left[ \sum_{\sigma_1 \in S_\ell} \ketbra{\bfv_1}{\sigma_1(\bfv_1)} \otimes \sum_{\sigma_2 \in S_{t}} \ketbra{\bfv_2}{\sigma_2(\bfv_2)}: \substack{X\gets \binom{T}{\ell}, \\ \bfv_1\gets X, \\ \bfv_2\gets T\setminus X} \right] \\
& = \Ex_{X\gets \binom{T}{\ell}} \left[ \ketbra{X}{X} \otimes \ketbra{T\setminus X}{T\setminus X} \right].
\end{align*}
For the first equality, we use~\Cref{lem:perm_split} and decompose $\sigma = \sigma_1 \circ \sigma_2$ for some $\sigma_1,\sigma_2$ such that $\sigma_1(x) = x$ for all $x \in \set{\ell+1,\ell+2,\cdots,\ell+t}$ and $\sigma_2(y) = y$ for all $y \in \set{1,2,\cdots,\ell}$. Since all $\ell+1,\ell+2,\cdots,\ell+t$ are fixed points of $\sigma_1$, we can view it as an element in $S_\ell$. Similarly, we view $\sigma_2(y)$ as an element in $S_{t}$. The second equality follows by denoting the first $\ell$ part of $\bfv$ by $\bfv_{[1:\ell]}$ and the last $t$ part of $\bfv$ by $\bfv_{[\ell+1:\ell+t]}$. The third equality holds because sampling a tuple $\bfv$ from $T$ is equivalent to sampling an $\ell$-subset $X$ from $T$ followed by ordering to elements in $X$ and $T\setminus X$.
\end{proof}

\subsection{Construction}\label{sec:prs-con}
In this section, we assume that the length of the common Haar state satisfies $n = n(\secp) \geq \secp$ for all $\secp \in \N$. We define the construction as follows: on input $k\in\set{0,1}^{\secp}$ and a single copy of the common Haar state $\ket{\vartheta}$,
\[
G_k(\ket{\vartheta}) := (Z^{k}\otimes I_{n-\secp})\ket{\vartheta}.
\]

\begin{lemma}[$\ell$-copy pseudorandomness] \label{lem:prs-sec}
Let $G$ be as defined above. Let 
\[
\rho := \E_{\substack{k\gets\set{0,1}^{\secp} \\ \ket{\vartheta}\gets\Haar_{n}}}
\left[
G_k(\ket{\vartheta})^{\otimes \ell}\otimes\ketbra{\vartheta}{\vartheta}^{\otimes t}
\right]
\text{ and }
\sigma := \Ex_{\substack{\ket{\varphi}\gets\Haar_n \\ \ket{\vartheta}\gets\Haar_{n}}}
\left[
\ketbra{\varphi}{\varphi}^{\otimes \ell}\otimes\ketbra{\vartheta}{\vartheta}^{\otimes t}
\right].
\]
Then $\TD\left(\rho,\sigma\right) = O\left(\frac{(\ell+t)^{2\ell}}{2^{\secp}}\right)$.
\end{lemma}
\begin{proof}
We prove this via a hybrid argument:
\paragraph{Hybrid $1$.} 
Sample $T\gets [0:\ell+t]^{2^n}$. 
Sample $k \gets \set{0,1}^{\secp}$. 
Output $((Z^{k} \otimes I_{n-\secp})^{\otimes \ell}\otimes I_{n}^{\otimes t})\ket{T}$.
\paragraph{Hybrid $2$.} 
Sample $T\gets [0:\ell+t]^{2^n}$ uniformly conditioned on $T\in\goodpre{\secp}{n-\secp}{\ell}{\ell+t}$. 
Sample $k\gets\set{0,1}^{\secp}$. 
Output $((Z^{k}\otimes I_{n-\secp})^{\otimes \ell}\otimes I_{n}^{\otimes t})\ket{T}$.
\paragraph{Hybrid $3$:}
Sample $T\gets [0:\ell+t]^{2^n}$ uniformly conditioned on $T\in\goodpre{\secp}{n-\secp}{\ell}{\ell+t}$.
Sample a uniform $\ell$-subset $T_1$ from $T$.
Output $\ket{T_1}\otimes\ket{T\setminus T_1}$.
\paragraph{Hybrid $4$.} 
Sample $T\gets [0:\ell+t]^{2^n}$.
Sample a uniform $\ell$-subset $T_1$ from $T$.\footnote{Since $T$ might have collisions, $T_1$ is allowed to contain duplicate elements.}
Output $\ket{T_1}\otimes\ket{T\setminus T_1}$.
\paragraph{Hybrid $5$.} 
Sample a collision-free $T$ from $[0:\ell+t]^{2^n}$.
Sample a uniform $\ell$-subset $T_1$ from $T$.
Output $\ket{T_1}\otimes\ket{T\setminus T_1}$.
\paragraph{Hybrid $6$.} 
Sample a uniform collision-free $T_1$ from $[0:\ell]^{2^n}$.
Sample a uniform collision-free $T_2$ from $[0:t]^{2^n}$ conditioned on $T_1$ and $T_2$ have no common elements.
Output $\ket{T_1} \otimes \ket{T_2}$.
\paragraph{Hybrid $7$.} 
Sample a uniform collision-free $T_1$ from $[0:\ell]^{2^n}$.
Sample a uniform collision-free $T_2$ from $[0:t]^{2^n}$.
Output $\ket{T_1} \otimes \ket{T_2}$.
\paragraph{Hybrid $8$.} 
Sample $T_1\gets [0:\ell]^{2^n}$.
Sample $T_2\gets [0:t]^{2^n}$.
Output $\ket{T_1} \otimes \ket{T_2}$.\\

\paragraph{Indistinuishability of Hybrids.} 
\begin{itemize}
    \item By~\Cref{fact:random_type_l_fold_collision}, the trace distance between Hybrid~$1$ and Hybrid~$2$ is $O((t+\ell)^{2\ell}/2^{\secp})$.  
\item From~\Cref{lem:nice_T}, the output of Hybrid~$2$ is
\[
\Ex_{\substack{T \gets [0:\ell+t]^{2^n}: \\ T\in\goodpre{\secp}{n-\secp}{\ell}{\ell+t} }}
\Ex_{T_1\gets \binom{T}{\ell}} \left[ \ketbra{T_1}{T_1} \otimes \ketbra{T\setminus T_1}{T\setminus T_1} \right].
\]
Hence, Hybrid~$2$ is equivalent to Hybrid~$3$. 
\item Again by~\Cref{fact:random_type_l_fold_collision}, the trace distance between Hybrid~$3$ and Hybrid~$4$ is $O((t+\ell)^{2\ell}/2^{\secp})$. \item The trace distance between Hybrid~$4$ and Hybrid~$5$ is $O((t+\ell)^{2}/2^{n})$ by the collision bound. 
\item Hybrid~$5$ and Hybrid~$6$ are equivalent. 
\item The trace distance between Hybrid~$6$ and Hybrid~$7$ is $O(t\ell/2^{n})$. 
\item Finally, the trace distance between Hybrid~$7$ and Hybrid~$8$ is $O((t^2+\ell^2)/2^{n})$ by the collision bound.
\end{itemize} 
This completes the proof. 
\end{proof}

\noindent In the following, we show that our construction also satisfies multi-key $\ell$-copy pseudorandomness using~\Cref{lem:prs-sec}.

\begin{lemma}[Multi-key $\ell$-copy pseudorandomness]\label{lem:prs-multi-sec}
Let $G$ be defined as above. Let 
\[
\rho := 
\bigotimes_{i=1}^p \Ex_{\ket{\varphi_i} \gets \Haar_n} 
\left[ 
\ketbra{\varphi_i}{\varphi_i}^{\otimes \ell} 
\right]
\otimes
\Ex_{\ket{\vartheta}\gets\Haar_{n}}
\left[
\ketbra{\vartheta}{\vartheta}^{\otimes t}
\right]
\text{ and }
\sigma := \Ex_{\ket{\vartheta}\gets\Haar_{n}}
\left[
\bigotimes_{i=1}^p  \Ex_{k_i\gets \bit^\secp} 
\left[ 
G_{k_i}(\ket{\vartheta})^{\otimes \ell} 
\right]
\otimes \ketbra{\vartheta}{\vartheta}^{\otimes t}
\right].
\]
Then $\TD\left(\rho,\sigma\right) = O\left(\frac{p \cdot (p\ell+t)^{2\ell}}{2^{\secp}}\right)$.
\end{lemma}
\begin{proof}
For $j = 0,1,\dots,p$, we define the following (hybrid) density matrices:\footnote{Similar to proving the output of a classical PRG on polynomial i.i.d uniform keys is computationally indistinguishable from polynomial i.i.d uniform strings, we can construct a security reduction to simulate these hybrids. However, since we are in the information-theoretic setting, we instead calculate their trace distances directly.}
\begin{align*}
\xi_j :=
\bigotimes_{i=1}^j \Ex_{\ket{\varphi_i}\gets\Haar_{n}}
\left[ 
\ketbra{\varphi_i}{\varphi_i}^{\otimes \ell} 
\right]
\otimes
\Ex_{\ket{\vartheta}\gets\Haar_{n}}
\left[
\bigotimes_{i=j+1}^p \Ex_{k_i \gets \bit^\secp} 
\left[ 
G_{k_i}(\ket{\vartheta})^{\otimes \ell} 
\right]
\otimes \ketbra{\vartheta}{\vartheta}^{\otimes t}
\right].
\end{align*}
We will complete the poof by showing that $\TD(\xi_j,\xi_{j+1}) = O\left( \frac{ ((p-j) \cdot \ell + t)^{2\ell} }{ 2^{\secp} } \right)$ for $j = 0,1,\dots,p-1$. By the property that $\TD(A\otimes X, A\otimes Y) = \TD(X,Y)$, the trace distance between $\xi_j$ and $\xi_{j+1}$ is identical to that of 
\[
\xi'_j :=
\Ex_{\ket{\vartheta}\gets\Haar_{n}}\left[ \bigotimes_{i=j+1}^p \Ex_{k_i \gets \bit^\secp} [G_{k_i}(\ket{\vartheta})^{\otimes \ell}] \otimes \ketbra{\vartheta}{\vartheta}^{\otimes t} \right]
\]
\[
\xi'_{j+1} :=
\Ex_{\ket{\varphi_{j+1}}\gets\Haar_{n}} 
\left[ 
\ketbra{\varphi_{j+1}}{\varphi_{j+1}}^{\otimes \ell} 
\right]
\otimes
\Ex_{\ket{\vartheta}\gets\Haar_{n}}\left[ \bigotimes_{i=j+2}^p \Ex_{k_i \gets \bit^\secp} [G_{k_i}(\ket{\vartheta})^{\otimes \ell}] \otimes \ketbra{\vartheta}{\vartheta}^{\otimes t} \right].
\]
By the monotonicity of trace distance (\ie $\TD(\cE(X),\cE(Y))\leq \TD(X,Y)$ for any quantum channel $\cE$) and setting $\cE := \bigotimes_{i=j+2}^p \Ex_{k_i \gets \bit^\secp} [G_{k_i}(\cdot)^{\otimes \ell}]$,\footnote{The channel $\cE$ acts as the identity on unspecified registers.} we have
\begin{align*}
& \TD(\xi'_j,\xi'_{j+1}) \leq \\
& \TD\left(
\Ex_{ \substack{k_{j+1} \gets \bit^\secp, \\ \ket{\vartheta} \gets \Haar_{n}} }
\left[ G_{k_{j+1}}(\ket{\vartheta})^{\otimes \ell} 
\otimes \ketbra{\vartheta}{\vartheta}^{\otimes (p-j-1)\ell+t} \right], 
\Ex_{ \substack{\ket{\varphi_{j+1}}\gets\Haar_{n}, \\ \ket{\vartheta} \gets \Haar_{n}} }
\left[ \ketbra{\varphi_{j+1}}{\varphi_{j+1}}^{\otimes \ell} 
\otimes \ketbra{\vartheta}{\vartheta}^{\otimes (p-j-1)\ell+t} \right]
\right) \\
& = O\left( \frac{ \left( (p-j) \ell + t \right)^{2\ell} }{ 2^{\secp} } \right),
\end{align*}
where the last equality follows from~\Cref{lem:prs-sec}. Applying the triangle inequality completes the proof.
\end{proof}

\begin{proof}[Proof of~\Cref{thm:PRS_CHS}]
Our construction is a efficiently-implementable unitary channel and thus satisfies the state generation property. Pseudorandomness follows from~\Cref{lem:prs-multi-sec}.
\end{proof}

\noindent As a remark, \Cref{lem:prs-sec} gives a simpler proof of the following theorem regarding the one-wayness of an ensemble of quantum states in~\cite{Col23}:
\begin{lemma}[{\cite[Lemma~5]{Col23}}] \label{lem:Z_haar_indis}
Consider the ensemble of states:
\[
\set{\rho_x}_{x \in \bit^n} 
= \left\{ \Ex_{\ket{\psi}\gets \Haar_n} \left[ (Z^x \otimes I^{\otimes m}) \ketbra{\psi}{\psi}^{\otimes m+1} (Z^x \otimes I^{\otimes m}) \right] \right\}_{x \in \bit^n}.
\]
Then, there is a constant $C > 0$, such that, for any POVM $\set{M_x}_{x \in \bit^n}$,
\[
\Ex_{x\gets\bit^n} \Tr( M_x\rho_x ) = C \cdot \left( \frac{m}{2^n} + \frac{m^7}{2^{3n}} \right)^{\frac{1}{2}}.
\]
\end{lemma}
\noindent By setting $\ell = 1, t = m, \secp = n$ in~\Cref{lem:prs-sec}, the ensemble of states $\set{\rho_x}_{x \in \bit^n}$ is pseudorandom, which implies its one-wayness.

In~\Cref{app:simpleproof}, we further give another proof by simplifying the calculation in~\Cref{lem:Z_haar_indis}, which may be of independent interest. Moreover, we eliminate the $m^7/2^{3n}$ term.

\subsection{Optimality of Our PRSG Construction}
\label{sec:prs-imp}

In this section, if the PRS generation algorithm uses only \emph{one copy} of the common Haar state, we show that $\ell$-copy statistical PRS and multi-key $\ell$-copy statistical PRS are impossible for $\ell = \Omega(\secp/\log(\secp))$ and $n=\omega(\log(\secp))$.

\begin{theorem}\label{thm:prs-imp}
Statistically secure $(\secparam,n,\ell)$-PRS is impossible in the CHS model if (a) the generation algorithm uses only one copy of the common Haar state, (b) $n=\omega(\log(\secp))$, (c) $\ell = \Omega(\secp/\log(\secp))$ and, (d) the length of the common Haar state is $n = \omega(\log(\secp))$.
\end{theorem}
\begin{proof}
We provve this by contradiction. Let there is a construction of such PRS $G$. First, from the state generation requirement of PRS generators, $G$ is a quantum channel that on any key and any pure state, outputs a pure state. Hence, $G$ is either an isometry or a replacement channel (\ie it outputs a fixed pure state for any input state).\footnote{According to the Stinespring representation, the action of any quantum channel is equivalent to appending auxiliary registers, performing a unitary operation on the enlarged system, and (possibly) taking a partial trace over some registers. For a bipartite entangled state, taking the partial trace over one subsystem results in a mixed state. Hence, after applying a unitary operation, either (1) there is no partial trace and the quantum channel is an isometry, or (2) the registers over which the partial trace is taken are not entangled with other registers, and the quantum channel is a replacement channel.}

We prove~\Cref{thm:prs-imp} by showing that for $t(\secp) := \secp^3$ and $\ell(\secp) := \secp/\log(\secp)$, there exists a (computationally unbounded) adversary $\alice$ such that
\[
\left|
\Pr_{\substack{k\gets\set{0,1}^{\secp}\\ \ket{\vartheta}\gets\Haar_{n}}}[\alice(\ketbra{\vartheta}{\vartheta}^{\otimes t}\otimes G(k,\ketbra{\vartheta}{\vartheta})^{\otimes \ell}) = 1] 
- \Pr_{\substack{\ket{\varphi}\gets\Haar_{n}\\ \ket{\vartheta}\gets\Haar_{n}}}[\alice(\ketbra{\vartheta}{\vartheta}^{\otimes t}\otimes \ketbra{\varphi}{\varphi}^{\otimes \ell}) = 1]
\right|
\]
is non-negligible. For short, we use the following notation:
\[
\rho_0 
:= \Ex_{k\gets\set{0,1}^{\secp}, \ket{\vartheta}\gets\Haar_{n}}\left[ \ketbra{\vartheta}{\vartheta}^{\otimes t}\otimes G(k,\ketbra{\vartheta}{\vartheta})^{\otimes \ell} \right]
\]
\[
\rho_1 
:= \Ex_{\ket{\varphi}\gets\Haar_{n}, \ket{\vartheta}\gets\Haar_{n}}
\left[ \ketbra{\vartheta}{\vartheta}^{\otimes t}\otimes \ketbra{\varphi}{\varphi}^{\otimes \ell} \right].
\]
The adversary $\alice$ is simple: it performs a binary measurement $\set{\Pi, I-\Pi}$ on input $\rho_b$ for $b\in\bit$, where $\Pi$ is the projection onto the eigenspace of $\rho_0$. The rank of $\rho_0$ and $\rho_1$ satisfies
\[
\rank(\rho_0) 
\leq 2^{\secp} \cdot {2^n + \ell+t-1 \choose \ell+t}
\quad \text{and} \quad
\rank(\rho_1) 
= {2^n + \ell-1 \choose \ell} \cdot {2^n + t-1 \choose t}.
\]
Now, by construction, we have 
\[
\Pr_{\substack{k\gets\set{0,1}^{\secp}\\ \ket{\vartheta}\gets\Haar_{n}}}[\alice(\ketbra{\vartheta}{\vartheta}^{\otimes t}\otimes G(k,\ketbra{\vartheta}{\vartheta})^{\otimes \ell}) = 1]
= \Tr(\Pi\rho_0) 
= \Tr(\rho_0) 
= 1.
\]
On the other hand, suppose $\Pi = \sum_{i=1}^{\rank(\rho_0)} \ketbra{u_i}{u_i}$, then
\begin{align*}
& \Pr_{\substack{\ket{\varphi}\gets\Haar_{n}\\ \ket{\vartheta}\gets\Haar_{n}}}[\alice(\ketbra{\vartheta}{\vartheta}^{\otimes t}\otimes \ketbra{\varphi}{\varphi}^{\otimes \ell}) = 1]
= \Tr(\Pi\rho_1) \\
& \leq \sum_{i=1}^{\rank(\rho_0)} 
\frac{1}{\binom{2^n + \ell - 1}{\ell}\binom{2^n + t - 1}{t}} \cdot
\sum_{T_1 \in [0:\ell]^{2^n}, T_2 \in [0:t]^{2^n}} |(\bra{T_1}\otimes\bra{T_2})\ket{u_i}|^2 \\
& \leq \frac{\rank(\rho_0)}{\binom{2^n + \ell - 1}{\ell}\binom{2^n + t - 1}{t}}
= \frac{\rank(\rho_0)}{\rank(\rho_1)}.
\end{align*}
A direct calculation yields:
\begin{align*}
\frac{\rank(\rho_0)}{\rank(\rho_1)} 
& = \frac{2^\secp}{\binom{\ell+t}{\ell}} \cdot
\prod_{i = 0}^{\ell-1} \left(1 + \frac{t}{2^n + i}\right)
\leq \frac{2^\secp}{(1+\frac{t}{\ell})^\ell} \cdot \prod_{i = 0}^{\ell-1} \left(1 + \frac{t}{2^n + i}\right) \\
& = 2^\secp \cdot \prod_{i = 0}^{\ell-1} \left( \frac{1 + \frac{t}{2^n + i}}{1+\frac{t}{\ell}} \right) 
\leq 2^\secp \cdot \left( \frac{1 + \frac{t}{2^n}}{1+\frac{t}{\ell}} \right)^\ell,
\end{align*}
where the first inequality follows from $\binom{\ell+t}{\ell} \geq (\frac{\ell+t}{\ell})^\ell$.
For $n = \omega(\log(\secp)), t = \secp^3$ and $\ell = \secp/\log(\secp)$, we have
\[
2^\secp \cdot \left( \frac{1 + \frac{t}{2^n}}{1+\frac{t}{\ell}} \right)^\ell
= \left( \frac{\secp \cdot (1 + \frac{\secp^3}{\secp^{\omega(1)}})}{1+\secp^2\log(\secp)}  \right)^{\secp/\log(\secp)} \leq \left( \frac{\secp \cdot 2}{\secp^2\log(\secp)}  \right)^{\secp/\log(\secp)}
\leq 2^{-\secp}
\]
for sufficiently large $\secp$. Hence, the distinguishing advantage ($1-2^{-\secp}$) is non-negligible. This completes the proof.
\end{proof}

\noindent Since multi-key pseudorandomness is stronger, we have the following immediate corollary.
\begin{corollary}
Multi-key statistically secure $(\secparam,n,\ell)$-PRS is impossible in the CHS model if (a) the generation algorithm uses only one copy of the common Haar state, (b) $n=\omega(\log(\secp))$, (c) $\ell = \Omega(\secp/\log(\secp))$ and, (d) the length of the common Haar state is $n = \omega(\log(\secp))$.
\end{corollary}

\section{Statistical Stretch PRFS Generators in the CHS model}
In this section, we extend our techniques from~\Cref{sec:prs-con} to construct an $(\secp,m,n,\ell)$-statistical PRFS in the CHS model, where $m=\secp^{c}$, $\ell = \secp^{1-c}/\log(\secp)^{1+\eps}$, the length of the common Haar state is $n\geq \secp^{1-c}$, for any constant $\eps > 0$ and $c\in [0,1)$. In the case when $n>\secp$, the construction satisfies stretch property. We prove the following theorem in~\Cref{sec:prfs-con}.

\begin{theorem} \label{thm:PRFS_CHS}
There exists an $(\secp,m,n,\ell)$-statistical selectively secure PRFS generator in the CHS model where the length of the common Haar state is $n(\secp)$, $m(\secp) = \secp^c$, $\ell=O(\secp^{1-c}/\log(\secp)^{1+\eps})$ and $n(\secp) \geq \secp^{1-c}$, for any constant $\eps > 0$ and for any $c \in [0,1)$.
\end{theorem}

\noindent Note that since a PRS can be used to computationally instantiate CHS in the plain model, the above result also gives us a way to get bounded-query long-input PRFS from PRS in the plain model. In more detail, we can start with a PRS that has stretch (i.e. $n > \secparam$) and then we can bootstrap into a PRFS for large input length at the cost of a reduction in stretch.\footnote{Formally, let $G_{PRS}$ is a $(\secp,n,\ell)$-PRS and $G(k,x,\ket{\phi})$ is $(\secp,m,n,\ell)$-statistical selectively secure PRFS generator in the CHS model with $n>\secp$, $\ell = O(\secp^{1-c}/\log(\secp)^{1+\eps})$ and $m(\secp) = \secp^c$, then for $K = (k_1,k_2)\in\bit^{\secp}\times\bit^{\secp}$ we can define $G_{PRFS}(k,x):= G(k_1,x,G_{PRS}(k_2))$ as the $(2\secp,m,n,\ell)$-PRFS generator.}

\begin{corollary}
    Assuming the existence of $(\secp,n,\ell)$-PRS, for $n>\secp$ and $\ell = O(\secp^{1-c}/\log(\secp)^{1+\eps})$, there exists a  selectively secure $(2\secp,m,n,\ell)$-PRFS generator with $m(\secp) = \secp^c$, for any constant $\eps > 0$ and for any $c \in [0,1)$.
\end{corollary}

\noindent Furthermore, since PRFS imply PRS, achieving an $\ell$-query statistical PRFS in the CHS model for $\ell = \Omega(\secp/\log(\secp))$ is impossible from~\Cref{thm:prs-imp}.

\begin{corollary}\label{cor:prfs-imp}
$(\secp,m,n,\ell)$-statistical PRFS is impossible in the CHS model if (a) the generation algorithm uses only one copy of the common Haar state, (b) $\ell = \Omega(\secp/\log(\secp))$, (c) the length of the common Haar state is $n$ and, (d) $n=\omega(\log(\secp))$.
\end{corollary}

\noindent We introduce several lemmas before proving~\Cref{thm:PRFS_CHS}.

\subsection{Useful Lemmas}

\noindent The following two lemmas are generalizations of the lemmas in~\Cref{sec:PRS_CHS}. In particular, they state that even after splitting an $\ell$-fold $n$-prefix collision-free type vector into $q$ subvectors, the action of a random Pauli-$Z$ still can be seen as a ``classical'' probabilistic process. 

\begin{lemma}[Generalization of~\Cref{lem:perm_split}] \label{lem:perm_split_gen}
Let $\ell,n,m,q,t\in\N$, $\ell_1,\ldots,\ell_q\in\N$, and $t_1,\ldots,t_q\in\N$ such that $\sum_{i=1}^{q} \ell_i = \ell$ and $\sum_{i=1}^{q} t_i = t$. For any $\bfv\in\bit^{(n+m)(\ell+t)}$ such that $\type(\bfv) \in \goodpre{n}{m}{\ell}{\ell+t}$, where $\bfv = (\bfv^1,\ldots,\bfv^q)$ and $\bfv^i\in\bit^{(n+m)(\ell_i+t_i)}$ for $i\in [q]$, and any $\sigma_1\in S_{\ell_1+t_1}, \sigma_2\in S_{\ell_2+t_2}, \cdots, \sigma_q\in S_{\ell_q+t_q}$, define the matrix
\[
A_{\bfv,\{\sigma_i\}_{i\in [q]}} 
:= \E_{k\gets\bit^{n}}\left[\bigotimes_{i=1}^{q} \left(\left(Z^{k}\otimes I_m\right)^{\otimes \ell_i}\otimes I_{n+m}^{\otimes t_i}\right)\ketbra{\bfv^i}{\sigma_i(\bfv^i)}\left(\left(Z^{k}\otimes I_m\right)^{\otimes \ell_i}\otimes I_{n+m}^{\otimes t_i}\right) \right].
\]
Then $A_{\bfv,\{\sigma_i\}_{i\in [q]}} = \bigotimes_{i=1}^q\ketbra{\bfv^i}{\sigma_i(\bfv^i)}$ if for all $i\in [q]$, $\sigma_i$ maps $[\ell_i]$ to $[\ell_i]$; otherwise, $A_{\bfv,\{\sigma_i\}_{i\in [q]}} = 0$. 
\end{lemma}
\begin{proof}
Suppose for all $i\in [q]$ and $j\in[\ell_i+t_i]$, $\bfv^i = (v^i_1||w^i_1,\ldots,v^i_{\ell_i+t_i}||w^i_{\ell_i+t_i})\in\bit^{(n+m)(\ell_i+t_i)}$ with $v_j\in\bit^n$ and $w_j\in\bit^m$. A direct calculation yields
\[
\left(\left(Z^{k}\otimes I_m\right)^{\otimes \ell_i}\otimes I_{n+m}^{\otimes t_i}\right)\ketbra{\bfv^i}{\sigma_i(\bfv^i)}\left(\left(Z^{k}\otimes I_m\right)^{\otimes \ell_i}\otimes I_{n+m}^{\otimes t_i}\right) 
= (-1)^{\langle k,\bigoplus_{j=1}^{\ell_i} (v^i_j\oplus v^i_{\sigma_i(j)})\rangle}\ketbra{\bfv^i}{\sigma_i(\bfv^i)}.
\]
After averaging over $k$,
\begin{align*}
A_{\bfv,\{\sigma_i\}_{i\in [q]}} 
& = \E_{k\gets\bit^{n}}\left[(-1)^{\langle k,\bigoplus_{i=1}^{q}\bigoplus_{j=1}^{\ell_i} (v^i_j\oplus v^i_{\sigma_i(j)})\rangle}\right] \cdot \bigotimes_{i=1}^{q}\ketbra{\bfv^i}{\sigma_i(\bfv^i)} \\
& = \begin{cases}
\bigotimes_{i=1}^{q}\ketbra{\bfv^i}{\sigma_i(\bfv^i)} & \text{ if } \bigoplus_{i=1}^{q}\bigoplus_{j=1}^{\ell_i} (v^i_j\oplus v^i_{\sigma_i(j)}) = 0 \\
0 & \text{ otherwise.}
\end{cases}
\end{align*}
Since $\type(\bfv) \in \goodpre{n}{m}{\ell}{t+\ell}$, the condition $\bigoplus_{i=1}^{q}\bigoplus_{j=1}^{\ell_i} v^i_j  
= \bigoplus_{i=1}^{q}\bigoplus_{j=1}^{\ell_i} v^i_{\sigma_i(j)}$ holds if and only if the two sets $\set{(i,j): i\in[q], j\in[\ell_i]}$ and $\set{(i,\sigma_i(j)): i\in[q], j\in[\ell_i]}$ are identical. The latter is equivalent to the condition: $\set{\sigma_i(j)):j\in[\ell_i]} = [\ell_j]$ for every $i\in[q]$. The proof is now complete.
\end{proof}

\begin{lemma}[Generalization of~\Cref{lem:nice_T}] \label{lem:nice_T_gen}
Let $\ell,n,m,q,t\in\N$, $\ell_1,\ldots,\ell_q\in\N$, and $t_1,\ldots,t_q\in\N$ such that $\sum_{i=1}^{q} \ell_i = \ell$ and $\sum_{i=1}^{q} t_i = t$. For any $T \in \goodpre{n}{m}{\ell}{t+\ell}$ and any mutually disjoint sets $T_1,\ldots,T_q$ satisfying $\bigcup_{i=1}^{q} T_i = T$ and $|T_i|=t_i+\ell_i$ for all $i\in [q]$,
\begin{multline*}
\Ex_{k\gets\bit^{n}}\left[\bigotimes_{i=1}^{q}  \left( \left(Z^{k}\otimes I_m\right)^{\otimes \ell_i}\otimes I_{n+m}^{\otimes t_i} \right) \ketbra{T_i}{T_i} \left(\left(Z^{k}\otimes I_m\right)^{\otimes \ell_i}\otimes I_{n+m}^{\otimes t_i}\right) \right] \\
= \bigotimes_{i=1}^{q}\Ex_{X_i\gets \binom{T_i}{\ell_i}} \left[ \ketbra{X_i}{X_i} \otimes \ketbra{T_i\setminus X_i}{T_i\setminus X_i} \right].
\end{multline*}
\end{lemma}
\begin{proof}
By~\Cref{eq:useful}, the left-hand side equals
\begin{align} \label{eq:split_gen}
\Ex_{\forall i\in [q], \bfv^i\gets T_i} \left[ \sum_{\forall i\in [q], \sigma_i\in S_{t_i+\ell_i}} \Ex_{k\gets\bit^{n}}\left[\bigotimes_{i=1}^{q} \left( \left(Z^{k}\otimes I_m\right)^{\otimes \ell_i}\otimes I_{n+m}^{\otimes t_i} \right) \ketbra{\bfv^i}{\sigma_i(\bfv^i)} \left( \left(Z^{k}\otimes I_m\right)^{\otimes \ell_i}\otimes I_{n+m}^{\otimes t_i} \right) \right] \right].
\end{align}
Then from the previous lemma (\Cref{lem:perm_split_gen})
\begin{align*}
\eqref{eq:split_gen}
& = \Ex_{\forall i\in [q],\bfv^i\gets T_i} \left[ \sum_{\forall i\in [q],\sigma^1_i \in S_{\ell_i}, \sigma^2_i \in S_{t_i}} \bigotimes_{i=1}^{q}\ketbra{\bfv^i}{\sigma^1_i\circ \sigma^2_i(\bfv^i)} \right] \\
& = \bigotimes_{i=1}^{q}\Ex_{\bfv^i\gets T_i} \left[ \sum_{\sigma^1_i \in S_{\ell_i}, \sigma^2_i \in S_{t_i}} \ketbra{\bfv^i}{\sigma^1_i\circ \sigma^2_i(\bfv^i)} \right] \\
& = \bigotimes_{i=1}^{q}\Ex_{\bfv^i\gets T_i} \left[ \sum_{\sigma^1_i \in S_{\ell_i}} \ketbra{\bfv^i_{[1:\ell]}}{\sigma^1_i(\bfv^i_{[1:\ell_i]})} \otimes \sum_{\sigma^2_i \in S_{t_i}} \ketbra{\bfv^i_{[\ell_i+1:t_i+\ell_i]}}{\sigma^2_i(\bfv^i_{[\ell_i+1:t_i+\ell_i]})}  \right] \\
& = \bigotimes_{i=1}^{q}\Ex \left[ \sum_{\sigma^1_i \in S_{\ell_i}} \ketbra{\bfv^i_1}{\sigma^1_i(\bfv^i_1)} \otimes \sum_{\sigma^2_i \in S_{t_i}} \ketbra{\bfv^i_2}{\sigma^2_i(\bfv^i_2)}: \substack{X_i\gets \binom{T_i}{\ell_i}, \\ \bfv^i_1\gets X_i, \\ \bfv^i_2\gets T_i\setminus X_i} \right] \\
& = \bigotimes_{i=1}^{q}\Ex_{X_i\gets \binom{T_i}{\ell_i}} \left[ \ketbra{X_i}{X_i} \otimes \ketbra{T_i\setminus X_i}{T_i\setminus X_i} \right].
\end{align*}
For the first equality, we use~\Cref{lem:perm_split_gen} and decompose for each $i\in [q]$, $\sigma_i = \sigma^1_i \circ \sigma^2_i$ for some $\sigma^1_i,\sigma^2_i$ such that $\sigma^1_i(x) = x$ for all $x \in \set{\ell_i+1,\ell_i+2,\cdots,\ell_i+t_i}$ and $\sigma^2_i(y) = y$ for all $y \in \set{1,2,\cdots,\ell_i}$. Similar to~\Cref{lem:nice_T}, we can view them as elements in $S_{\ell_i}$ and $S_{t_i}$. The second equality follows from linearity of trace. The third equality follows by denoting for each $i\in [q]$, the first $\ell_i$ part of $\bfv^i$ by $\bfv^i_{[1:\ell]}$ and the last $t_i$ part of $\bfv^i$ by $\bfv^i_{[\ell_i+1:\ell_i+t_i]}$. The fourth equality holds because for each $i\in [q]$, sampling $\bfv_i$ from $T_i$ is equivalent to sampling an $\ell_i$-subset $X_i$ from $T_i$ followed by ordering the elements in $X_i$ and $T_i\setminus X_i$.
\end{proof}

\subsection{Construction} \label{sec:prfs-con}

We extend the techniques used in~\Cref{sec:prs-con} to construct a statistical PRFS in~\Cref{fig:prfs}. The construction samples a uniform key for each position of the input being zero or one. Applying this to the common Haar state gives us the output of the PRFS. The details can be seen in~\Cref{fig:prfs}. Thoughout this section, one should think of $m = \secp^{c}$ and $\secp' = \secp^{1-c}$ for some constant $c\in [0,1)$.

\begin{figure}[h]
   \begin{tabular}{|p{16cm}|}
   \hline \\
\par Given the common Haar state $\ket{\vartheta}$, on the key $K = (k_1^0,\ldots,k_m^0, k_1^1,\ldots,k_m^1)\in\bit^{2\secp' m}$ and the input $\bfx = (x_1,\ldots,x_m)\in\bit^{m}$, define $G(K,\bfx,\ket{\vartheta})$ as follows:
\begin{itemize}
    \item $\ket{\psi_{K,\bfx}} = G(K,\bfx,\ket{\vartheta}) = (Z^{\bigoplus_{i=1}^m k^{x_i}_i}\otimes I_{n-\secp'}) \ket{\vartheta}.$
    \item Output $\ket{\psi_{K,\bfx}}$.
\end{itemize}
\\
\hline
\end{tabular}
\caption{PRFS in the CHS model}
\label{fig:prfs}
\end{figure}

\noindent The main property of the construction that makes it a PRFS is its ability to \emph{disentangles} any type state in $\goodpre{\secp'}{n-\secp'}{\ell}{\ell+t}$ into a probabilistic mixture of disjoint subsets of the type. Formally, we show the following lemma:

\begin{lemma}\label{lem:prfs-type}
Let $G$ be defined as in~\Cref{fig:prfs}. Let $q,t\in\N$, $\ell_1,\ldots,\ell_q\in\N$ such that $\sum_{i=1}^{q}\ell_i = \ell$. Let $\bfx^1,\ldots,\bfx^q\in\bit^{m}$ with $\bfx^i\neq\bfx^j$ for all $i\neq j\in [q]$. For any $T\in\goodpre{\secp'}{n-\secp'}{\ell}{\ell+t}$, the following density matrices are equal: 
\[
\rho := \E_{K\gets\set{0,1}^{2m\secp'}}
\left[
\left(\bigotimes_{i=1}^{q} G_K(\bfx^i,\cdot)^{\otimes \ell_i}\otimes I^{\otimes t}\right) \ketbra{T}{T}
\right]
\]
\[
\sigma := \Ex_{(T_1, T_2, \dots, T_q, \hat{T})}
\left[
\bigotimes_{i=1}^{q}\ketbra{T_i}{T_i}\otimes\ketbra{\hat{T}}{\hat{T}}
\right]
\]
where we omit the Hermitian conjugate of the unitary in $\rho$ and identify it as a quantum channel; $(T_1, T_2, \dots,$ $T_q, \hat{T})$ in $\sigma$ are sampled as follows: for $i = 1, 2, \dots, q$, recursively sample an $\ell_i$-subset from $T\setminus(\bigcup_{j=1}^{i-1} T_j)$ uniformly at random and let $\hat{T}:= T \setminus (\bigcup_{j=1}^{q} T_j)$.
\end{lemma}
\begin{proof}
We define the following notation: Let $\ell:\bit^{\ast}\to\N$, such that for all $i\in [m]$, $y\in\bit^{i}$, $\ell(y) = \sum_{\substack{j\in [q]: \\ \bfx^{j}_{[1:i]}=y}}\ell_j$. Then we start by simplifying $\rho$:
\begin{align*}
    & \rho  
    = \E_{K\gets\set{0,1}^{2m\secp'}}
    \left[
    \left(\bigotimes_{j=1}^{q} G_K(\bfx^j,\cdot)^{\otimes \ell_j}\otimes I^{\otimes t}\right) \ketbra{T}{T}
    \right] \\
    & = \E_{K\gets\set{0,1}^{2m\secp'}}
    \left[
    \left(\bigotimes_{j=1}^{q} \left(Z^{\bigoplus_{i=1}^m k^{x^j_i}_i}\otimes I_{n-\secp'}\right)^{\otimes \ell_j}\otimes I^{\otimes t}\right) \ketbra{T}{T}
    \right] \\
    & = \E_{K\gets\set{0,1}^{2m\secp'}}
    \left[
    \left(\bigotimes_{j=1}^{q} \left(Z^{k^{x^j_m}_m}\otimes I_{n-\secp'}\right)^{\otimes \ell_j}\otimes I^{\otimes t}\right)\cdots
    \left[
    \left(\bigotimes_{j=1}^{q} \left(Z^{k^{x^j_1}_1}\otimes I_{n-\secp'}\right)^{\otimes \ell_j}\otimes I^{\otimes t}\right) \ketbra{T}{T}
    \right]\right] \\
    & = \E_{k_m^0,k_m^1\gets\set{0,1}^{\secp'}}
    \left[
    \left(\bigotimes_{j=1}^{q} \left(Z^{k^{x^j_m}_m}\otimes I_{n-\secp'}\right)^{\otimes \ell_j}\otimes I^{\otimes t}\right)\cdots\E_{k_1^0,k_1^1\gets\set{0,1}^{\secp'}}
    \left[
    \left(\bigotimes_{j=1}^{q} \left(Z^{k^{x^j_1}_1}\otimes I_{n-\secp'}\right)^{\otimes \ell_j}\otimes I^{\otimes t}\right) \ketbra{T}{T}
    \right]\right], \\
\end{align*}
where the first equality is by definition of $\rho$, second equality is by definition of $G$, third equality is because $Z^{k_1\oplus k_2} = Z^{k_1}Z^{k_2}$ and fourth equality is by linearity of expectation. We define for $i\in [q]$, the following channels $\cC_i$:
$$\cC_i(\cdot) = \E_{k_i^0,k_i^1\gets\set{0,1}^{\secp'}}
    \left[
    \left(\bigotimes_{j=1}^{q} \left(Z^{k^{x^j_i}_i}\otimes I_{n-\secp'}\right)^{\otimes \ell_j}\otimes I^{\otimes t}\right)\cdot\left(\left(\bigotimes_{j=1}^{q} \left(Z^{k^{x^j_i}_i}\otimes I_{n-\secp'}\right)^{\otimes \ell_j}\otimes I^{\otimes t}\right)\right)^{\dagger}\right],$$
then $\rho = \cC_m(\cC_{m-1}(\ldots\cC_1(\ketbra{T}{T})\ldots)).$

\noindent We define $(\set{T_x}_{x\in\bit^i},\hat{T})\gets\mu_i$ as follow: For all $x\in\bit^i$, sample an $\ell(x)$-subset from $T\setminus(\bigcup_{y=0}^{x} T_y)$ uniformly and let $\hat{T}:= T \setminus (\bigcup_{y=0}^{2^i} T_y)$.

\noindent We start by computing $\cC_1(\ketbra{T}{T})$, by~\Cref{lem:nice_T_gen}, $$\cC_1(\ketbra{T}{T})=\E_{(\set{T_x}_{x\in\bit},\hat{T})\gets\mu_1}\left[\bigotimes_{b\in\bit}\ketbra{T_b}{T_b}\otimes\ketbra{\hat{T}}{\hat{T}}\right].$$
In fact, for all $i\in [q]$,
$$\cC_i(\cC_{i-1}(\ldots\cC_1(\ketbra{T}{T})\ldots))=\E_{(\set{T_y}_{y\in\bit^i},\hat{T})\gets\mu_i}\left[\bigotimes_{y\in\bit^i}\ketbra{T_y}{T_y}\otimes\ketbra{\hat{T}}{\hat{T}}\right].$$

\noindent We can show the above by induction on $i$. Assume that for some $i\in [q]$, 
$$\cC_i(\cC_{i-1}(\ldots\cC_1(\ketbra{T}{T})\ldots))=\E_{(\set{T_y}_{y\in\bit^i},\hat{T})\gets\mu_i}\left[\bigotimes_{y\in\bit^i}\ketbra{T_y}{T_y}\otimes\ketbra{\hat{T}}{\hat{T}}\right],$$ then for $i+1\in [q]$, 
\begin{align*}
& \cC_{i+1}(\cC_{i}(\ldots\cC_1(\ketbra{T}{T})\ldots)) \\
& = \cC_{i+1}\left(\E_{(\set{T_y}_{y\in\bit^i},\hat{T})\gets\mu_i}\left[\bigotimes_{y\in\bit^i}\ketbra{T_y}{T_y}\otimes\ketbra{\hat{T}}{\hat{T}}\right]\right)\\
& = \E_{(\set{T_y}_{y\in\bit^i},\hat{T})\gets\mu_i}\left[\E_{k^0_{i+1},k^1_{i+1}}\left[\bigotimes_{y\in\bit^i}\left(\left(Z^{k^{0}_{i+1}}\otimes I_{n-\secp'}\right)^{\otimes \ell(y0)}\otimes\left(Z^{k^{1}_{i+1}}\otimes I_{n-\secp'}\right)^{\otimes \ell(y1)} \ketbra{T_y}{T_y}\right)\otimes\ketbra{\hat{T}}{\hat{T}}\right]\right]\\
& = \E_{(\set{T_y}_{y\in\bit^{i+1}},\hat{T})\gets\mu_{i+1}}\left[\bigotimes_{y\in\bit^{i+1}}\ketbra{T_y}{T_y}\otimes\ketbra{\hat{T}}{\hat{T}}\right],
\end{align*}
where the first equality is by the induction hypothesis, the second equality is by the definition of $\cC_{i+1}$ and the third equality is by~\Cref{lem:perm_split_gen}.

\noindent Hence, we get 
$$\rho = \cC_m(\cC_{m-1}(\ldots\cC_1(\ketbra{T}{T})\ldots))=\E_{(\set{T_x}_{x\in\bit^m},\hat{T})\gets\mu_m}\left[\bigotimes_{y\in\bit^m}\ketbra{T_y}{T_y}\otimes\ketbra{\hat{T}}{\hat{T}}\right].$$
Ignoring the $y\in\bit^{m}$ for which $\ell(y) = 0$, we get 
$$\rho = \Ex_{(T_1, T_2, \dots, T_q, \hat{T})}
\left[
\bigotimes_{i=1}^{q}\ketbra{T_i}{T_i}\otimes\ketbra{\hat{T}}{\hat{T}}
\right],
$$
where $(T_1, T_2, \dots, T_q, \hat{T})$ are sampled as follows: for $i = 1, 2, \dots, q$, sample an $\ell_i$-subset from $T\setminus(\bigcup_{j=1}^{i-1} T_j)$ uniformly and let $\hat{T}:= T \setminus (\bigcup_{j=1}^{q} T_j)$.
Hence, $\rho=\sigma$.
\end{proof}

\begin{lemma}[Pseudorandomness] \label{lem:prfs-sec}
Let $G$ be as defined above. Let $q,t\in\N$, let $\ell_1,\ldots,\ell_q\in\N$ be such that $\sum_{i=1}^{q}\ell_i = \ell$. Let $\bfx^1,\ldots,\bfx^q\in\bit^{m}$. Let 
\[
\rho := \E_{\substack{K\gets\set{0,1}^{2m\secp'} \\ \ket{\vartheta}\gets\Haar_{n}}}
\left[
\otimes_{i=1}^q G_K(\bfx^i,\ket{\vartheta})^{\otimes \ell_i}\otimes\ketbra{\vartheta}{\vartheta}^{\otimes t}
\right],
\]
and
\[
\sigma := \Ex_{\substack{\forall i\in[q], \ket{\varphi_i}\gets\Haar_n \\ \ket{\vartheta}\gets\Haar_{n}}}
\left[
\otimes_{i=1}^{q}\ketbra{\varphi_i}{\varphi_i}^{\otimes \ell_i}\otimes\ketbra{\vartheta}{\vartheta}^{\otimes t}
\right].
\]
Then $\TD\left(\rho,\sigma\right) = O\left(\frac{(\ell+t)^{2\ell}}{2^{\secp'}}\right)$.
\end{lemma}

\begin{proof}
    We prove this using hybrid arguments:
\paragraph{Hybrid $1$.} 
Sample $T\gets [0:\ell+t]^{2^n}$. 
Sample $K \gets \set{0,1}^{2m\secp'}$. 
Output $(\bigotimes_{j=1}^{q}(Z^{\oplus_{i=1}^m k^{x^j_i}_i}\otimes I_{n-\secp'})^{\otimes \ell_j}\otimes I_{n}^{\otimes t})\ket{T}$.
\paragraph{Hybrid $2$.} 
Sample $T\gets [0:\ell+t]^{2^n}$ uniformly conditioned on $T\in\goodpre{\secp'}{n-\secp'}{\ell}{\ell+t}$. 
Sample $K\gets\set{0,1}^{2m\secp'}$. 
Output $(\bigotimes_{j=1}^{q}(Z^{\oplus_{i=1}^m k^{x^j_i}_i}\otimes I_{n-\secp'})^{\otimes \ell_j}\otimes I_{n}^{\otimes t})\ket{T}$.
\paragraph{Hybrid $3$:}
Sample $T\gets [0:\ell+t]^{2^n}$ uniformly conditioned on $T\in\goodpre{\secp'}{n-\secp'}{\ell}{\ell+t}$.
Sample a uniform for all $j\in [q]$, $\ell_j$-subsets $T_j$ from $T$ such that for any $j\neq j'\in [q]$, $T_j\cap T_{j'} = \emptyset$. Define $\Tilde{T} = \bigcup_{j=1}^q T_j$.
Output $\bigotimes_{j=1}^{q}\ket{T_j}\otimes\ket{T\setminus \Tilde{T}}$.
\paragraph{Hybrid $4$.}
Sample $T\gets [0:\ell+t]^{2^n}$.
For all $j\in [q]$, sample a uniform $\ell_j$-subset $T_j$ from $T\setminus \bigcup_{i=1}^{j-1} T_j$.\footnote{Since $T$ might have collisions, $T_j$ is allowed to contain duplicate elements.}
Output $\bigotimes_{j=1}^{q}\ket{T_j}\otimes\ket{T\setminus \bigcup_{j=1}^q T_i}$.
\paragraph{Hybrid $5$.} 
Sample a collision-free $T$ from $[0:\ell+t]^{2^n}$.
Sample a uniform for all $j\in [q]$, $\ell_j$-subsets $T_j$ from $T$ such that for any $j\neq j'\in [q]$, $T_j\cap T_{j'} = \emptyset$. Define $\Tilde{T} = \bigcup_{j=1}^q T_j$.
Output $\bigotimes_{j=1}^{q}\ket{T_j}\otimes\ket{T\setminus \Tilde{T}}$.
\paragraph{Hybrid $6$.} 
For all $j\in [q]$, sample uniform collision-free $T_j$ from $[0:\ell_j]^{2^n}$ conditioned on $T_j$ and $\bigcup_{i=1}^{j-1} T_j$ have no common elements.
Sample a uniform collision-free $\hat{T}$ from $[0:t]^{2^n}$ conditioned on $\bigcup_{j=1}^{q}T_j$ and $\hat{T}$ have no common elements.
Output $\bigotimes_{j=1}^{q}\ket{T_j}\otimes \ket{\hat{T}}$.
\paragraph{Hybrid $7.y$, for $y\in [0:q-1]$.}
For all $j\in [q-y]$, sample uniform collision-free $T_j$ from $[0:\ell_j]^{2^n}$ conditioned on $T_j$ and $\bigcup_{i=1}^{j-1} T_j$ have no common elements.
For all $j\in [q-y+1:q]$, sample a uniform collision-free $T_j$ from $[0:\ell_j]^{2^n}$.
Sample a uniform collision-free $\hat{T}$ from $[0:t]^{2^n}$.
Output $\bigotimes_{j=1}^{q}\ket{T_j}\otimes \ket{\hat{T}}$.
\paragraph{Hybrid $8$.} 
For all $j\in [q]$, sample $T_j\gets [0:\ell_j]^{2^n}$.
Sample $\hat{T}\gets [0:t]^{2^n}$.
Output $\bigotimes_{j=1}^{q}\ket{T_j}\otimes \ket{\hat{T}}$.

\paragraph{Indistinuishability of Hybrids.} 
\begin{itemize}
    \item By~\Cref{fact:random_type_l_fold_collision}, the trace distance between Hybrid~$1$ and Hybrid~$2$ is $O((t+\ell)^{2\ell}/2^{\secp'})$.
    \item From~\Cref{lem:prfs-type}, the output of Hybrid~$2$ is equivalent to Hybrid~$3$.
    \item By~\Cref{fact:random_type_l_fold_collision}, the trace distance between Hybrid~$3$ and Hybrid~$4$ is $O((t+\ell)^{2\ell}/2^{\secp'})$.
    \item The trace distance between Hybrid~$4$ and Hybrid~$5$ is $O((t+\ell)^{2}/2^{n})$ by collision bound.
    \item Hybrid~$5$ and Hybrid~$6$ are equivalent.
    \item The trace distance between Hybrid~$6$ and Hybrid~$7.0$ is $O(t\ell/2^{n})$.
    \item For $y\in [0:q-2]$, the trace distance between Hybrid~$7.y$ and Hybrid~$7.(y+1)$ is $O(\ell_{q-y}(\sum_{j=1}^{q-y-1}\ell_j)/2^{n})$.
    \item Finally, the trace distance between Hybrid~$7$ and Hybrid~$8$ is $O((t^2+\sum_{j=1}^{q}\ell_j^2)/2^{n})$ by collision bound.
\end{itemize}
This completes the proof.
\end{proof}

\begin{remark}
Note that the above construction is still secure if we set $k_i^{1} = 0$ for all $i\in [2:m]$. This slightly reduces the key length from $2m\secp'$ to $(m+1)\secp'$.
\end{remark}

\section{Quantum Commitments in the CHS model} \label{sec:Commitment}
In this section, we construct a commitment scheme that satisfies poly-copy statistical hiding and statistical sum-biding in the CHS model. The scheme is inspired by the quantum commitment scheme proposed in~\cite{MY21,MNY23}. In contrast to the scheme in~\cite{MY21}, our construction is not of the canonical form~\cite{Yan22}. To achieve binding, similar to~\cite{MNY23}, the receiver needs to perform several SWAP tests. To achieve hiding, our scheme relies on the multi-key pseudorandomness property in~\Cref{lem:prs-multi-sec}. 

\subsection{Construction}
We assume that $n(\secp) \geq \secp + 1$ for all $\secp\in\N$. Our construction, parameterized by the polynomial $p = p(\secp) := \secp$, is shown in~\Cref{fig:commitment}.

\begin{theorem} \label{thm:com}
The construction in~\Cref{fig:commitment} is a quantum commitment in the CHS model.
\end{theorem}

\begin{figure}[h]
   \begin{tabular}{|p{16cm}|}
   \hline \\
\par Commit phase: The sender $C_{\secp}$ on input $b\in\bit$ does the following:
    \begin{itemize}
        \item Use $p$ copies of the common Haar state $\ket{\vartheta}$ to prepare the state $\ket{\Psi_b}_{\sfC\sfR} := \bigotimes_{i=1}^p \ket{\psi_b}_{\sfC_i\sfR_i}$, where 
        \[
            \ket{\psi_0}_{\sfC_i\sfR_i} 
            := \frac{1}{\sqrt{2^\secp}} \sum_{k \in \bit^\secp} (Z^k \otimes I_{n-\secp}) \ket{\vartheta}_{\sfC_i} \ket{k||0^{n-\secp}}_{\sfR_i}
        \]
        and
        \[
            \ket{\psi_1}_{\sfC_i\sfR_i} 
            := \frac{1}{\sqrt{2^n}} \sum_{j \in \bit^n} \ket{j}_{\sfC_i} \ket{j}_{\sfR_i},
        \]
        and $\sfC := (\sfC_1, \sfC_2, \dots, \sfC_p)$ and $\sfR := (\sfR_1, \sfR_2, \dots, \sfR_p)$.
        \item Send register $\sfC$ to the receiver.
    \end{itemize}

\par Reveal phase: 
\begin{itemize}
    \item The sender sends $b$ and register $\sfR$ to the receiver.
    \item The receiver prepares the state $\ket{\Psi_b}_{\sfC'\sfR'} = \bigotimes_{i=1}^p \ket{\psi_b}_{\sfC'_i\sfR'_i}$ by using $p$ copies of the common Haar state $\ket{\vartheta}$, where $\sfC' := (\sfC'_1, \sfC'_2, \dots, \sfC'_p)$ and $\sfR' := (\sfR'_1, \sfR'_2, \dots, \sfR'_p)$ are receiver's registers.
    \item For $i\in[p]$, the receiver performs the SWAP test between registers $(\sfC_i,\sfR_i)$ and $(\sfC'_i,\sfR'_i)$.
    \item The receiver outputs $b$ if all SWAP tests accept; otherwise, outputs $\bot$.
\end{itemize}
\\ 
\hline
\end{tabular}
\caption{Quantum commitment scheme in the CHS model}
\label{fig:commitment}
\end{figure}

\subsection{Proving Hiding and Binding}
Now, we prove~\Cref{thm:com}.
\begin{proof}[Proof of~\Cref{thm:com}]
Clearly, the construction has perfect correctness.
\item \paragraph{Poly-copy statistical hiding.} 
It follows immediately from~\Cref{lem:prs-multi-sec} by setting $\ell = 1$.

\item \paragraph{Statistical sum binding.}
For any (fixed) common Haar state $\ket{\vartheta}$ and $i\in[p]$, it holds that
\begin{align} \label{eq:fidelity}
& F( \Tr_{\sfR_i}( \ketbra{\psi_0}{\psi_0}_{\sfC_i\sfR_i}), \Tr_{\sfR_i}( \ketbra{\psi_1}{\psi_1}_{\sfC_i\sfR_i}) ) \nonumber \\
= & F\left( \underbrace{\frac{1}{2^\secp} \sum_{k \in \bit^\secp} (Z^k \otimes I_{n-\secp}) \ketbra{\vartheta}{\vartheta}_{\sfC_i}(Z^k \otimes I_{n-\secp})}_{=: \rho_0}, \frac{I_{\sfC_i} }{2^n} \right) \nonumber \\
= & 2^{-n} \cdot \Tr( \sqrt{\rho_0} )^2 \nonumber \\
\leq & 2^{-n} \cdot \rank(\sqrt{\rho_0}) \cdot \Tr(\rho_0) \nonumber \\
\leq & 2^{-n} \cdot 2^{\secp} \cdot 1
= 2^{-(n-\secp)},
\end{align}
where the second equality is by the definition of fidelity $F(\rho,\sigma) = \left( \Tr(\sqrt{\sqrt{\rho}\sigma\sqrt{\rho}}) \right)^2$; the first inequality follows from $\Tr(\rho)^2 \leq \rank(\rho) \cdot \Tr(\rho^2)$ for $\rho \succeq 0$; the second inequality is because $\rank(\sqrt{\rho}) = \rank(\rho)$ for $\rho \succeq 0$ and $\rank(X+Y) \leq \rank(X) + \rank(Y)$.

\noindent Let $M^{(b)}_{\sfC\sfR}$ be the POVM operator corresponding to that the receiver outputs $b$ (\ie all the SWAP tests accept), 
\[
M^{(b)}_{\sfC\sfR} := \bigotimes_{i\in[p]} \left( \frac{I_{\sfC_i\sfR_i} + \ketbra{\psi_b}{\psi_b}_{\sfC_i\sfR_i})}{2} \right)
= \Ex_{\cS\subseteq[p]} \left[ \bigotimes_{i\in \cS}\ketbra{\psi_b}{\psi_b}_{\sfC_i\sfR_i} \otimes \bigotimes_{i\notin \cS} I_{\sfC_i\sfR_i} \right],
\]
where $\cS$ is a uniformly random subset of $[p]$. Then the probability that the receiver outputs $b$ is
\begin{align*}
p_b & := \Tr\left(M^{(b)}_{\sfC\sfR} \underbrace{\Tr_\sfE( U^{(b)}_{\sfR\sfE} \ketbra{\Phi}{\Phi}_{\sfC\sfR\sfE} U^{(b)\dagger}_{\sfR\sfE}}_{=: \rho^{(b)}_{\sfC\sfR}} ) \right) \\
& = \Ex_{\cS\subseteq[p]} \left[ \Tr\left( \bigotimes_{i\in \cS}\ketbra{\psi_b}{\psi_b}_{\sfC_i\sfR_i} \otimes \bigotimes_{i\notin \cS} I_{\sfC_i\sfR_i} \cdot \rho^{(b)}_{\sfC\sfR} \right) \right] \\
& = \Ex_{\cS\subseteq[p]} \left[ \underbrace{ F \left( \bigotimes_{i\in \cS}\ketbra{\psi_b}{\psi_b}_{\sfC_i\sfR_i}, \Tr_{\sfC_i\sfR_i: i\notin\cS}(\rho^{(b)}_{\sfC\sfR}) \right)}_{=: p_{b,S}} \right],
\end{align*}
where $\sfE$ is the sender's internal register, $\ket{\Phi}_{\sfC\sfR\sfE}$ is the malicious sender's initial state that might depend on $\ket{\vartheta}$ (we omit the dependence for simplicity), and $U^{(b)}_{\sfR\sfE}$ is the malicious sender's attacking unitary for $b$; we plug in the definition of $M^{(b)}_{\sfC\sfR}$ and use the short-hand notation $\rho^{(b)}_{\sfC\sfR}$ to obtain the second equality.

\noindent For any fixed $\cS \subseteq [p]$, we have
\begin{align*}
& p_{0,\cS} + p_{1,\cS} \\
& = F\left( \bigotimes_{i\in \cS}\ketbra{\psi_0}{\psi_0}_{\sfC_i\sfR_i}, \Tr_{\sfC_i\sfR_i: i\notin\cS}(\rho^{(0)}_{\sfC\sfR}) \right) 
+ F\left( \bigotimes_{i\in \cS}\ketbra{\psi_1}{\psi_1}_{\sfC_i\sfR_i}, \Tr_{\sfC_i\sfR_i: i\notin\cS}(\rho^{(1)}_{\sfC\sfR}) \right) \\
& \leq F\left( \bigotimes_{i\in \cS} \Tr_{\sfR_i} (\ketbra{\psi_0}{\psi_0}_{\sfC_i\sfR_i}), \Tr_{\sfC_i:i\notin \cS}\Tr_{\sfR}(\rho^{(0)}_{\sfC\sfR}) \right) 
+ F\left( \bigotimes_{i\in \cS} \Tr_{\sfR_i} (\ketbra{\psi_1}{\psi_1}_{\sfC_i\sfR_i}), \Tr_{\sfC_i:i\notin \cS}\Tr_{\sfR}(\rho^{(1)}_{\sfC\sfR}) \right) \\
& \leq 1 + F\left( \bigotimes_{i\in \cS} \Tr_{\sfR_i} (\ketbra{\psi_0}{\psi_0}_{\sfC_i\sfR_i}), \bigotimes_{i\in \cS} \Tr_{\sfR_i} (\ketbra{\psi_1}{\psi_1}_{\sfC_i\sfR_i}) \right)^{1/2} \\
& = 1 + \bigotimes_{i\in \cS} F \left( \Tr_{\sfR_i} (\ketbra{\psi_0}{\psi_0}_{\sfC_i\sfR_i}), \Tr_{\sfR_i} (\ketbra{\psi_1}{\psi_1}_{\sfC_i\sfR_i}) \right)^{1/2} 
\leq 1 + 2^{\frac{-|\cS|(n-\secp)}{2}},
\end{align*}
where the first inequality follows from the fact that taking a partial trace won't decrease the fidelity; the second inequality is because $\Tr_{\sfR}(\rho^{(0)}_{\sfC\sfR}) = \Tr_{\sfR}(\rho^{(1)}_{\sfC\sfR})$ and $F(\rho,\xi)+F(\sigma,\xi) \leq 1 + \sqrt{F(\rho,\sigma)}$~\cite{NS03}; the last equality follows from the fact that $F(\bigotimes_i \rho_i, \bigotimes_i \sigma_i) = \prod_i F(\rho_i, \sigma_i)$; the last inequality follows from~\Cref{eq:fidelity}. Finally, we bound the probability $p_0 + p_1$ as follows:
\begin{align*}
p_0 + p_1
& = \Ex_{\cS\subseteq[p]} \left[ p_{0,\cS} + p_{1,\cS} \right] 
\leq 1 + \Ex_{\cS\subseteq[p]} \left[ 2^{\frac{-|\cS|(n-\secp)}{2}} \right]
= 1 + 2^{-p} \cdot \sum_{s = 0}^t \binom{p}{s} 2^{\frac{-s(n-\secp)}{2}} \\
& = 1 + \left( \frac{ 1 + 2^{\frac{-(n-\secp)}{2}} }{2} \right)^p
= 1 + \negl(\secp),
\end{align*}
since we set $n(\secp)\geq \secp+1$ and $p(\secp) = \secp = \omega(\log(\secp))$.
\end{proof}

\section{LOCC Indistinguishability} \label{sec:LOCC}
In this section, we prove our main technical theorem for proving impossibilities and separations in~\Cref{sec:Imp_QCCC_CHS} and~\Cref{sec:QBB_QCCC}.

\subsection{Definitions}
\begin{definition}[LOCC adversaries]
An \emph{LOCC adversary} is a tuple $(\alice,\bob)$, where $\alice$ and $\bob$ are spatially separated, non-uniform, and computationally unbounded quantum algorithms without pre-shared entanglement. In addition, $\alice$ and $\bob$ can only perform local operations on their registers and communicate classically.
\end{definition}

\begin{definition}[LOCC Indistinguishability]
We say that two density matrices $(\rho_{\sfA\sfB}, \sigma_{\sfA\sfB})$ are \emph{$\veps$-LOCC indistinguishable} if for any LOCC adversary $(\alice,\bob)$ with $\alice$ taking as input register $\sfA$ and $\bob$ taking as input register $\sfB$, the probability that $\bob$ outputs $1$ satisfies\footnote{Since $(\alice,\bob)$ are allowed to communicate and we do not care about communication complexity, it is without loss of generality to assume that $\bob$ outputs the bit.}
\[
\left| \Pr[(\alice,\bob)(\rho_{\sfA\sfB}) = 1] - \Pr[(\alice,\bob)(\sigma_{\sfA\sfB}) = 1] \right|
\leq \veps.
\]
If $\veps(\cdot)$ is negligible, then we simply say that $(\rho_{\sfA\sfB}, \sigma_{\sfA\sfB})$ are LOCC indistinguishable.
\end{definition}

\noindent It is well-known that the class of operations having positive partial transpose (PPT) is a strict superset of the class of LOCC operations (see, \eg~\cite{DLT02,EW02,CLMOW14,Har23}). Hence, it suffices to consider the maximum distinguishing advantage over PPT measurements.

\begin{lemma} \label{lem:LOCC_PPT}
For any two density matrices $\rho_{\sfA\sfB}, \sigma_{\sfA\sfB}$ and $\veps \geq 0$, $(\rho_{\sfA\sfB},\sigma_{\sfA\sfB})$ are $\veps$-LOCC indistinguishable if
\[
\sup_{\substack{M_{\sfA\sfB}: 0\preceq M_{\sfA\sfB} \preceq I \\ \land 0\preceq M_{\sfA\sfB}^{\Gamma_\sfB} \preceq I}} |\Tr(M_{\sfA\sfB}(\rho_{\sfA\sfB} - \sigma_{\sfA\sfB}))|
\leq \veps.
\]
\end{lemma}

We can extend the LOCC indistinguishability into the multi-party setting. 

\begin{definition}[$m$-party LOCC adversaries]
An \emph{$m$-party LOCC adversary} is an $m$-tuple $(\sfP_1,\sfP_2,\dots,\sfP_m)$, where each $\sfP_i$ is a non-uniform, computationally unbounded quantum algorithms and every distinct pair $(\sfP_i,\sfP_j)$ is spatially separated and without pre-shared entanglement. In addition, every party can only perform local operations on their registers and communicate classically.
\end{definition}

\begin{definition}[$m$-party LOCC Indistinguishability]
We say that two density matrices $(\rho_{\sfP}, \sigma_{\sfP})$ on register $\sfP = (\sfP_1,\sfP_2,\dots,\sfP_m)$ are \emph{$(m,\veps)$-LOCC indistinguishable} if for any $m$-party LOCC adversary $(\sfP_1,\sfP_2,\dots,\sfP_m)$ with each $\sfP_i$ taking as input register $\sfP_i$, the probability that $\sfP_m$ outputs $1$ satisfies
\[
\left| \Pr[(\sfP_1,\sfP_2,\dots,\sfP_m)(\rho_\sfP) = 1] - \Pr[(\sfP_1,\sfP_2,\dots,\sfP_m)(\sigma_\sfP) = 1] \right|
\leq \veps.
\]
\end{definition}

\subsection{LOCC Haar Indistinguishability}
We first introduce several useful lemmas.
\begin{lemma} \label{lem:type_bipartite}
For any $d\in\N$, any set $T \subseteq [d]$ and any integer $0 \leq x \leq |T|$, the type state $\ket{T}$ can be written as
\[
\ket{T}_{\sfA\sfB} 
= \sum_{X\in\binom{T}{x}} \frac{1}{\sqrt{\binom{|T|}{x}}} \ket{X}_\sfA \otimes \ket{T\setminus X}_\sfB,
\]
where register $\sfA$ contains the first $x$ qudits and register $\sfB$ contains the last $|T| - x$ qudits.
\end{lemma}
\begin{proof}
For every $X\in\binom{T}{x}$, the inner product of $\ket{X}_\sfA \otimes \ket{T\setminus X}_\sfB$ and $\ket{T}_{\sfA\sfB}$ is 
\[
\left( \frac{1}{\sqrt{x!(|T|-x)!}}\sum_{\bfx \in X, \bfy \in T \setminus X} \bra{\bfx}_\sfA \otimes \bra{\bfy}_\sfB \right) 
\left( \frac{1}{\sqrt{T!}}\sum_{\bfv \in T} \ket{\bfv}_{\sfA\sfB} \right)
= \sqrt{\frac{x!(|T|-x)!}{T!}}
= \frac{1}{\sqrt{\binom{|T|}{x}}}.
\]
Moreover, $\ket{X}_\sfA \otimes \ket{T\setminus X}_\sfB$ and $\ket{X'}_\sfA \otimes \ket{T\setminus X'}_\sfB$ are orthogonal for every pair $X\neq X'\in\binom{T}{x}$. Since $\ket{T}$ is normalized, the equality holds.
\end{proof}

\paragraph{Kneser graphs.}
For any $v,k\in\N$, the \emph{Kneser graph} $K(v, k)$ is the graph whose vertices correspond to the $k$-element subsets of the set $[v]$, and two vertices are adjacent if and only if the two corresponding sets are disjoint.

\begin{lemma}[{\cite[Theorem~1]{LW12}}] \label{lem:Kneser}
For any $v,k\in\N$ such that $v\geq 2k+1$, the sum of absolute eigenvalues of the adjacency matrix of $K(v,k)$ (which is equal to its $1$-norm) is 
\[
\frac{2^k(v-1)(v-3)\dots(v-2k+1)}{k!}.
\]
\end{lemma}

The following lemma is the crux for proving~\Cref{thm:LOCC_main}.
\begin{lemma} \label{lem:PPT}
Let 
$\Tilde{\rho}_{\sfA\sfB} := \Ex_{T\gets\binom{[d]}{2t}} \left[ \ketbra{T}{T}_{\sfA\sfB} \right]$
and 
$\Tilde{\sigma}_{\sfA\sfB} := \Ex_{S_A, S_B \gets\binom{[d]}{t}: S_A \cap S_B = \emptyset} \left[ \ketbra{S_A}{S_A}_\sfA \otimes \ketbra{S_B}{S_B}_\sfB \right]$. 
Then we have $\norm{\Tilde{\rho}^{\Gamma_\sfB}_{\sfA\sfB} - \Tilde{\sigma}^{\Gamma_\sfB}_{\sfA\sfB}}_1 \leq O( t^2/d )$.
\end{lemma}

\begin{proof}
By~\Cref{lem:type_bipartite}, we can expand $\Tilde{\rho}$ as follows:
\begin{align*}
\Tilde{\rho}_{\sfA\sfB}
= \frac{1}{\binom{d}{2t}\binom{2t}{t}} \sum_{T\in\binom{[d]}{2t}} \sum_{\substack{X,Y\in\binom{T}{t}}} \ketbra{T\setminus X}{T\setminus Y}_\sfA \otimes \ketbra{X}{Y}_\sfB.
\end{align*}
On the other hand, we have
\begin{align*}
\Tilde{\sigma}_{\sfA\sfB}
= \frac{1}{\binom{d}{t}\binom{d-t}{t}} \sum_{\substack{S_A,S_B\in\binom{[d]}{t}: \\ S_A \cap S_B = \emptyset}} \ketbra{S_A}{S_A}_\sfA \otimes \ketbra{S_B}{S_B}_\sfB
= \frac{1}{\binom{d}{2t}\binom{2t}{t}} \sum_{\substack{S_A,S_B\in\binom{[d]}{t}: \\ S_A \cap S_B = \emptyset}} \ketbra{S_A}{S_A}_\sfA \otimes \ketbra{S_B}{S_B}_\sfB.
\end{align*}
Taking partial transpose \wrt $\sfB$, we have $\Tilde{\sigma}_{\sfA\sfB}^{\Gamma_\sfB} = \Tilde{\sigma}_{\sfA\sfB}$ and
\begin{align*}
\Tilde{\rho}_{\sfA\sfB}^{\Gamma_\sfB}
& = \frac{1}{\binom{d}{2t}\binom{2t}{t}} \sum_{T\in\binom{[d]}{2t}} \sum_{\substack{X,Y\in\binom{T}{t}}} \ketbra{T\setminus X}{T\setminus Y}_\sfA \otimes \ketbra{Y}{X}_\sfB \\
& = \Tilde{\sigma}_{\sfA\sfB}^{\Gamma_\sfB} + \frac{1}{\binom{d}{2t}\binom{2t}{t}} \sum_{T\in\binom{[d]}{2t}} \sum_{\substack{X,Y\in\binom{T}{t}:\\ X\neq Y}} \ketbra{T\setminus X}{T\setminus Y}_\sfA \otimes \ketbra{Y}{X}_\sfB.
\end{align*}
where the second equality is because when $X = Y$,
\[
\sum_{T\in\binom{[d]}{2t}} \sum_{\substack{X\in\binom{T}{t}}} \ketbra{T\setminus X}{T\setminus X}_\sfA \otimes \ketbra{X}{X}_\sfB = \sum_{S_A,S_B\in\binom{[d]}{t}: S_A \cap S_B = \emptyset} \ketbra{S_A}{S_A}_\sfA \otimes \ketbra{S_B}{S_B}_\sfB.
\]
Hence, we have
\begin{align*}
\norm{\Tilde{\rho}_{\sfA\sfB}^{\Gamma_\sfB} - \Tilde{\sigma}_{\sfA\sfB}^{\Gamma_\sfB}}_1 
= \frac{1}{\binom{d}{2t}\binom{2t}{t}} \norm{ \sum_{T\in\binom{[d]}{2t}} \sum_{\substack{X,Y\in\binom{T}{t}:\\ X\neq Y}} \ketbra{T\setminus X}{T\setminus Y}_\sfA \otimes \ketbra{Y}{X}_\sfB }_1.
\end{align*}
Now, we will apply a double-counting argument. Each $(T,X,Y)$ can uniquely correspond to a tuple of mutually disjoint sets $(C,I,X',Y')$ satisfying $C = T \setminus (X \cup Y)$ ($C$ denotes  complement of $X \cup Y$), $I = X \cap Y$ ($I$ denotes intersection), $X' = X \setminus I$ and $Y' = Y \setminus I$. Hence, $T \setminus X = C \uplus Y'$, $Y = I \uplus Y'$, $T \setminus Y = C \uplus X'$, and $X = I \uplus X'$ where $\uplus$ denotes the disjoint union. By further classifying the summands according to $s := |C| = |I| \in \set{0,1,\dots,t-1}$ (note that then $|X'| = |Y'| = t-s$), we have
\begin{align*}
& \norm{\Tilde{\rho}_{\sfA\sfB}^{\Gamma_\sfB} - \Tilde{\sigma}_{\sfA\sfB}^{\Gamma_\sfB}}_1 
=  \frac{1}{\binom{d}{2t}\binom{2t}{t}} \norm{ \sum_{s = 0}^{t-1} \sum_{C\in\binom{[d]}{s}} \sum_{I \in \binom{[d]\setminus C}{s}} 
 \sum_{\substack{X',Y' \in \binom{[d]\setminus (C\uplus I)}{t-s}: \\ X' \cap Y' = \emptyset}} \ket{C\uplus Y'}_\sfA \ket{I\uplus Y'}_\sfB \bra{C\uplus X'}_\sfA \bra{I\uplus X'}_\sfB }_1 \\
& \leq \frac{1}{\binom{d}{2t}\binom{2t}{t}} \sum_{s = 0}^{t-1} \sum_{C\in\binom{[d]}{s}} \sum_{I \in \binom{[d]\setminus C}{s}} 
\Bigg\| \underbrace{\sum_{\substack{X',Y' \in \binom{[d]\setminus (C\uplus I)}{t-s}: \\ X' \cap Y' = \emptyset}} \ket{C\uplus Y'}_\sfA \ket{I\uplus Y'}_\sfB \bra{C\uplus X'}_\sfA \bra{I\uplus X'}_\sfB}_{=: K_{C,I}} \Bigg\|_1, \label{eqn:LOCC1} \tag{$*$}
\end{align*}
where the inequality follows from the triangle inequality. Observe that for every $(C,I)$, the matrix $K_{C,I}$ is isospectral\footnote{Two matrices are isospectral to one another if they have the same set of non-zero eigenvalues, including multiplicities.} to the adjacency matrix of the Kneser graph $K(d-2s, t-s)$. By~\Cref{lem:Kneser}, we continue bounding the above inequality:
\begin{align*}
\eqref{eqn:LOCC1} & = \frac{1}{\binom{d}{2t}\binom{2t}{t}} \sum_{s = 0}^{t-1} \binom{d}{s} \binom{d-s}{s}
\frac{ 2^{t-s}(d-2s-1)(d-2s-3)\ldots (d-2t+1)}{(t-s)!} \\
& = \sum_{s = 0}^{t-1} \frac{2^{t-s} \left( \frac{t!}{s!} \right)^2 }{(t-s)!(d-2s)(d-2s-2)\ldots (d-2t+2)} \\
& \leq \sum_{s = 0}^{t-1} \frac{2^{t-s}\cdot t^{2(t-s)}}{(t-s)!(d-2t+2)^{t-s}}. \label{eqn:LOCC2} \tag{$**$}
\end{align*}
By letting $k:=t-s$, we finally have
\[
\eqref{eqn:LOCC2} = \sum_{k=1}^t \frac{(\frac{2t^2}{d-2t+2})^k}{k!}
\leq \exp(\frac{2t^2}{d-2t+2}) - 1 
= O\left( \frac{t^2}{d} \right).  \qedhere
\]
\end{proof}

\begin{theorem}[LOCC Haar Indistinguishability] \label{thm:LOCC_main}
Let 
$\rho_{\sfA\sfB} := \Ex_{\ket{\psi} \gets \Haar_n} \left[ \ketbra{\psi}{\psi}^{\otimes t}_{\sfA} \otimes \ketbra{\psi}{\psi}_{\sfB}^{\otimes t} \right]$ 
and 
$\sigma_{\sfA\sfB} := \Ex_{\ket{\psi} \gets \Haar_n} \left[ \ketbra{\psi}{\psi}^{\otimes t}_{\sfA} \right]
\otimes \Ex_{\ket{\phi} \gets \Haar_n} \left[ \ketbra{\phi}{\phi}^{\otimes t}_{\sfB} \right]$.
Then $\rho_{\sfA\sfB}$ and $\sigma_{\sfA\sfB}$ are $O( t^2/2^n )$-LOCC indistinguishable.
\end{theorem}
\begin{proof}
Let $\Tilde{\rho}_{\sfA\sfB}$ and $\Tilde{\sigma}_{\sfA\sfB}$ be defined as in~\Cref{lem:PPT}. By the collision bound, both $(\rho_{\sfA\sfB}, \Tilde{\rho}_{\sfA\sfB})$, and $(\sigma_{\sfA\sfB}, \Tilde{\sigma}_{\sfA\sfB})$ are $O(t^2/d)$-close in trace distance, which trivially implies their $O(t^2/d)$-LOCC indistinguishability. Thus, if suffices to show that $\Tilde{\rho}_{\sfA\sfB}$ and $\Tilde{\sigma}_{\sfA\sfB}$ are $O(t^2/d)$-LOCC indistinguishable. From~\Cref{lem:LOCC_PPT}, the LOCC distinguishing advantage of $\Tilde{\rho}_{\sfA\sfB}$ and $\Tilde{\sigma}_{\sfA\sfB}$ can be upper bounded by
\begin{align*}
& \sup_{\substack{M_{\sfA\sfB}: 0\preceq M_{\sfA\sfB} \preceq I \\ \land 0\preceq M_{\sfA\sfB}^{\Gamma_\sfB} \preceq I}} |\Tr(M_{\sfA\sfB}(\Tilde{\rho}_{\sfA\sfB} - \Tilde{\sigma}_{\sfA\sfB}))| \\
& \leq \sup_{\substack{M_{\sfA\sfB}: 0\preceq M_{\sfA\sfB}^{\Gamma_\sfB} \preceq I}} |\Tr(M_{\sfA\sfB}(\Tilde{\rho}_{\sfA\sfB} - \Tilde{\sigma}_{\sfA\sfB}))| \\
& = \sup_{\substack{M_{\sfA\sfB}: 0\preceq M_{\sfA\sfB}^{\Gamma_\sfB} \preceq I}} |\Tr(M_{\sfA\sfB}^{\Gamma_\sfB}(\Tilde{\rho}_{\sfA\sfB}^{\Gamma_\sfB} - \Tilde{\sigma}_{\sfA\sfB}^{\Gamma_\sfB}))| \\
& = \sup_{\substack{M_{\sfA\sfB}: 0\preceq M_{\sfA\sfB} \preceq I}} |\Tr(M_{\sfA\sfB}(\Tilde{\rho}_{\sfA\sfB}^{\Gamma_\sfB} - \Tilde{\sigma}_{\sfA\sfB}^{\Gamma_\sfB}))| \\
& = \frac{1}{2} \norm{ \Tilde{\rho}_{\sfA\sfB}^{\Gamma_\sfB} - \Tilde{\sigma}_{\sfA\sfB}^{\Gamma_\sfB} }_1.
\end{align*}
The first inequality holds because we omit the constraint $0\preceq M_{\sfA\sfB} \preceq I$. The first equality follows from the fact that $\Tr(P_{\sfA\sfB}Q_{\sfA\sfB}) = \Tr(P_{\sfA\sfB}^{\Gamma_\sfB}Q_{\sfA\sfB}^{\Gamma_\sfB})$ for all matrices $P_{\sfA\sfB},Q_{\sfA\sfB}$. Since $\Tilde{\rho}_{\sfA\sfB}^{\Gamma_\sfB}-\Tilde{\sigma}_{\sfA\sfB}^{\Gamma_\sfB}$ is Hermitian and has trace zero, the last equality follows from the variational definition of trace norm. Applying~\Cref{lem:PPT} completes the proof.
\end{proof}

To prove the separations in~\Cref{sec:QBB_QCCC}, we rely on the following generalization of~\Cref{thm:LOCC_main} which states the LOCC indistinguishability when $(\alice,\bob)$ are further given many i.i.d. input instances with different lengths.

\begin{corollary} \label{cor:LOCC_repetition}
For positive integers $s,t,n_1,n_2,\dots,n_s$, define
\[
\rho_{\sfA\sfB} := \bigotimes_{i=1}^s \Ex_{\ket{\psi_i} \gets \Haar_{n_i}} \left[ \left( \ketbra{\psi_i}{\psi_i}^{\otimes t} \right)_{\sfA_i} \otimes \left( \ketbra{\psi_i}{\psi_i}^{\otimes t} \right)_{\sfB_i} \right]
\]
\[
\sigma_{\sfA\sfB} := \bigotimes_{i=1}^s \Ex_{\ket{\psi_i} \gets \Haar_{n_i}} \left[ \left( \ketbra{\psi_i}{\psi_i}^{\otimes t} \right)_{\sfA_i} \right]
\otimes \bigotimes_{i=1}^s \Ex_{\ket{\phi_i} \gets \Haar_{n_i}} \left[ \left( \ketbra{\phi_i}{\phi_i}^{\otimes t} \right)_{\sfB_i} \right],
\]
where $\sfA = (\sfA_1,\sfA_2,\dots,\sfA_s)$ and $\sfB = (\sfB_1,\sfB_2,\dots,\sfB_s)$. Then $\rho_{\sfA\sfB}$ and $\sigma_{\sfA\sfB}$ are $O\left( \sum_{i=1}^s t^2/2^{n_i} \right)$-LOCC indistinguishable.
\end{corollary}
\begin{proof}
For $0\leq k \leq s$, we define the (hybrid) state
\[
\xi_k := \bigotimes_{i = 1}^k \Ex_{\ket{\psi_i} \gets \Haar_{n_i}} \left[ \left( \ketbra{\psi_i}{\psi_i}^{\otimes t} \right)_{\sfA_i} \otimes \left( \ketbra{\psi_i}{\psi_i}^{\otimes t} \right)_{\sfB_i} \right]  \otimes \bigotimes_{j = k+1}^s \left( \Ex_{\ket{\psi_j} \gets \Haar_{n_j}} \left[ \ketbra{\psi_j}{\psi_j}^{\otimes t}_{\sfA_j} \right] \otimes \Ex_{\ket{\phi_j} \gets \Haar_{n_j}} \left[ \ketbra{\phi_j}{\phi_j}^{\otimes t}_{\sfB_j} \right]\right).
\]
Note that $\xi_0 = \rho$ and $\xi_s = \sigma$. By the triangle inequality, we have
\begin{align*}
\sup_{(\alice,\bob)} \left| \Pr[(\alice,\bob)(\rho) = 1] - \Pr[(\alice,\bob)(\sigma) = 1] \right| 
\leq \sum_{k=0}^{s-1} \sup_{(\alice,\bob)}\left| \Pr[(\alice,\bob)(\xi_k) = 1] - \Pr[(\alice,\bob)(\xi_{k+1}) = 1] \right|,
\end{align*}
where the supremum is over all LOCC adversary. We will show that for each $k$,
\begin{multline*}
\sup_{(\alice,\bob)}\left| \Pr[(\alice,\bob)(\xi_k) = 1] - \Pr[(\alice,\bob)(\xi_{k+1}) = 1] \right|
= \sup_{(\alice,\bob)} \bigg| \Pr[(\alice,\bob)\left( \Ex_{\ket{\psi} \gets \Haar_{n_{k+1}}} \left[ \ketbra{\psi}{\psi}^{\otimes t}_{\sfA} \otimes \ketbra{\psi}{\psi}_{\sfB}^{\otimes t} \right] \right) = 1]  \\
- \Pr[(\alice,\bob)\left( \Ex_{\ket{\psi} \gets \Haar_{n_{k+1}}} \left[ \ketbra{\psi}{\psi}^{\otimes t}_{\sfA} \right]
\otimes \Ex_{\ket{\phi} \gets \Haar_{n_{k+1}}} \left[ \ketbra{\phi}{\phi}^{\otimes t}_{\sfB} \right] \right) = 1] \bigg|,
\end{multline*}
which then completes the proof by~\Cref{thm:LOCC_main}. It is easy to see that the LHS is at least as large as the RHS since $(\alice,\bob)$ on the LHS can simply discard all input registers except for $(\sfA_{k+1},\sfB_{k+1})$. To see that the RHS is at least as large as the LHS, for every $(\alice,\bob)$ on the LHS, we define $(\alice',\bob')$ on the RHS based on $(\alice,\bob)$ as follows. For $0\leq i \leq k$, $\alice'$ samples the classical description of i.i.d. $n_i$-qubit Haar states $\ket{\psi_i}$ and sends them to $\bob'$.\footnote{Note that $(\alice',\bob')$ are information-theoretic and thus the description can approximate the Haar state with arbitrarily small error.} They then prepare $t$ copies of the quantum state $\ket{\psi_i}$ according to the description on registers $\sfA_i$ and $\sfB_i$ respectively. For $k+2 \leq j \leq s$, $\alice'$ and $\bob'$ each locally sample $t$ copies of i.i.d. $n_j$-qubit Haar state $\ket{\psi_j}$ and $\ket{\phi_j}$ on registers $\sfA_j$ and $\sfB_j$ respectively. They then embed their input on registers $\sfA_{k+1}$ and $\sfB_{k+1}$, and run $(\alice,\bob)$ respectively. Since the input of $(\alice,\bob)$ is exactly $\xi_k$ or $\xi_{k+1}$, $(\alice',\bob')$ have the same advantage as that of $(\alice,\bob)$.
\end{proof}

Moreover, we have the following corollary regarding the multi-party LOCC indistinguishability.
\begin{corollary} \label{cor:LOCC_multi-party}
Let 
$\rho_\sfP :=  \Ex_{\ket{\psi} \gets \Haar_n} \left[ \bigotimes_{i=1}^m \ketbra{\psi}{\psi}^{\otimes t}_{\sfP_i} \right]$ 
and 
$\sigma_\sfP := \bigotimes_{i=1}^m \Ex_{\ket{\psi_i} \gets \Haar_n} \left[ \ketbra{\psi_i}{\psi_i}^{\otimes t}_{\sfP_i} \right]$ where register $\sfP = (\sfP_1,\sfP_2,\dots,\sfP_m)$. Then $\rho_\sfP$ and $\sigma_\sfP$ are $(m,O( m^2t^2/2^n ))$-LOCC indistinguishable.
\end{corollary}
\begin{proof}
Similar to the proof of~\Cref{cor:LOCC_repetition}, we prove it via a hybrid argument. Without loss of generality, we can assume that $m$ is a power of $2$, \ie $m = 2^r$. Otherwise, by the monotonicity of LOCC indistinguishability, we can instead consider the smallest power of $2$ that is greater than or equal to $m$, which only increases the advantage by a constant factor. Define the following states for $k \in \set{ 0, 1, \dots, r }$:
\[
\xi_k := \bigotimes_{i = 0}^{2^{r-k} - 1} \Ex_{\ket{\psi_i} \gets \Haar_n} \left[ \bigotimes_{j = 1}^{2^k} \ketbra{\psi_i}{\psi_i}^{\otimes t}_{\sfP_{i2^k + j}} \right].
\]
For each $\xi_k$, there are $2^{r-k}$ blocks, each corresponding to a Haar state. Within the $i$-th block, there are $2^k$ parties holding $t$-copies of the same state $\ket{\psi_i}$. By construction, $\xi_r = \rho$ and $\xi_0 = \sigma$. We will show that the LOCC distinguishing advantage between $\xi_k$ and $\xi_{k+1}$ is $O\left( \frac{m2^kt^2}{2^n} \right)$. This would then implies that the LOCC distinguishing advantage between $\rho$ and $\sigma$ is $\sum_{k = 0}^{r-1} O\left( \frac{m2^kt^2}{2^n} \right) = O\left( \frac{m^2t^2}{2^n} \right)$.

To prove the closeness between $\xi_k$ and $\xi_{k+1}$, we introduce sub-hybrids $\xi_{k,\ell}$ for $\ell \in \set{ 0, 1, \dots, 2^{r-k} }$. In $\xi_{k,\ell}$, the first $\ell$ blocks are all ``split in half''. That is, for any $i\in\set{1,2,\dots,\ell}$, in the $i$-th block, the first $2^{k-1}$ parties are holding $t$-copies of $\ket{\psi_{i,0}}$ and the other $2^{k-1}$ parties are holding $t$-copies of $\ket{\psi_{i,1}}$. For any $i\in\set{\ell+1,\ell+2,\dots,2^{r-k}}$, in the $i$-th block, all $2^k$ parties are holding $t$-copies of the same $\ket{\psi_i}$. Hence, the only difference between $\xi_{k,\ell}$ and $\xi_{k,\ell+1}$ is in the $(\ell+1)$-th block --- in the former all $2^k$ parties are holding $t$-copies of the same $\ket{\psi_{\ell+1}}$, whereas in the latter the first $2^{k-1}$ parties are holding $t$-copies of $\ket{\psi_{\ell+1,0}}$ and the other $2^{k-1}$ parties are holding $t$-copies of $\ket{\psi_{\ell+1,1}}$. Now, we can view the first $2^{k-1}$ parties and the other $2^{k-1}$ parties as two entities. By~\Cref{thm:LOCC_main} and setting the number of copies each party receives as $2^{k-1}t$, the LOCC distinguishing advantage between $\xi_{k,\ell}$ and $\xi_{k,\ell+1}$ is $O\left( \frac{(2^{k-1}t)^2}{2^n} \right)$. This implies that the LOCC distinguishing advantage between $\xi_{k}$ and $\xi_{k+1}$ is $O\left( \frac{2^{r-k}(2^{k-1}t)^2}{2^n} \right) = O\left( \frac{2^{r+k}t^2}{2^n} \right) = O\left( \frac{m2^kt^2}{2^n} \right)$ as desired.
\end{proof}

\begin{remark} \label{remark:Harrow}
We compare~\Cref{cor:LOCC_multi-party} with~\cite[Theorem~8]{Har23}. Although both theorems address multi-party LOCC indistinguishability, they are incomparable for the following reasons. \Cref{cor:LOCC_multi-party} is stronger in the sense that each party receives $t$ copies of the states, as opposed to the single-copy setting in~\cite[Theorem~8]{Har23}. Moreover, when $t = 1$, \Cref{cor:LOCC_multi-party} implies an $O(m^2/2^n)$ bound which is better than the $O(m^2/\sqrt{2^n})$ bound given by~\cite[Theorem~8]{Har23}. On the other hand, the statement of~\cite[Theorem~8]{Har23} is more general since their bound holds for a large family of input states. While the input states in~\Cref{cor:LOCC_multi-party} are fixed to $\rho$ and $\sigma$.
\end{remark}

\subsection{An Optimal LOCC Haar distinguisher}
We present an (optimal) LOCC Haar distinguisher with advantage $\Omega(t^2/2^n)$. Hence, the upper bound in~\Cref{thm:LOCC_main} is tight. 
\begin{theorem} \label{thm:LOCC_lower_bound}
There exists an LOCC adversary that distinguishes $\rho_{\sfA\sfB} := \Ex_{\ket{\psi} \gets \Haar_n} \big[ \bigotimes_{i=1}^t \ketbra{\psi}{\psi}_{\sfA_i} \otimes$
$\bigotimes_{i=1}^t \ketbra{\psi}{\psi}_{\sfB_i} \big]$ 
from 
$\sigma_{\sfA\sfB} := \Ex_{\ket{\psi} \gets \Haar_n} \left[ \bigotimes_{i=1}^t \ketbra{\psi}{\psi}_{\sfA_i} \right]
\otimes \Ex_{\ket{\phi} \gets \Haar_n} \left[ \bigotimes_{i=1}^t \ketbra{\phi}{\phi}_{\sfB_i} \right]$ with advantage $\Omega(t^2/2^n)$, where $\sfA = (\sfA_1,\dots,\sfA_t)$ and $\sfB = (\sfB_1,\dots,\sfB_t)$. Moreover, the running time is polynomial in $t$ and $n$.
\end{theorem}
\begin{proof}
The LOCC adversary $(\alice,\bob)$ is defined as follows. For $1\leq i\leq t$, $\alice$ measures register $\sfA_i$ in the computational basis and obtains the outcome $a_i\in\bit^n$. Similarly, $\bob$ measures every $\sfB_i$ in the computational basis and obtains $b_i$. Then $\bob$ sends $b_1,b_2,\dots,b_t$ to $\alice$, and $\alice$ outputs $1$ if there is no collision among $a_1,a_2,\dots,a_t,b_1,b_2,\dots,b_t$. Let $d:=2^n$. The distinguishing advantage can be lower bounded as follows:
\begin{align*}
& \Pr[(\alice,\bob)(\sigma_{\sfA\sfB}) = 1] - \Pr[(\alice,\bob_{\sfA\sfB})(\rho) = 1] \\
& = \Pr_{T_1,T_2\gets[0:t]^d}[T_1,T_2 \text{ are collision-free}\ \land T_1,T_2 \text{ are disjoint}] - \Pr_{T\gets[0:2t]^d}\left[ T \text{ is collision-free} \right] \\
& = \frac{\binom{d}{t}\binom{d-t}{t}}{\binom{d+t-1}{t}^2} - \frac{\binom{d}{2t}}{\binom{d+2t-1}{2t}} \\
& = \frac{\binom{d}{2t}}{\binom{d+2t-1}{2t}} \cdot \left( \frac{\binom{d}{t}\binom{d-t}{t}\binom{d+2t-1}{2t}}{\binom{d+t-1}{t}^2\binom{d}{2t}} - 1 \right) \\
& = \prod_{i=0}^{2t-1} \left( 1 - \frac{2t-1}{d+i} \right) \cdot \left( \prod_{i=0}^{t-1} \left( 1 + \frac{t}{d+i} \right) - 1 \right) \\
& = \left( 1 - O\left(\frac{t^2}{d}\right) \right) \cdot \left( 1 + \Omega\left(\frac{t^2}{d}\right) - 1 \right)
= \Omega\left(\frac{t^2}{d}\right),
\end{align*}
where the first equality follows from the fact that $\sigma_{\sfA\sfB} = \Ex_{T_1,T_2\gets[0:t]^d} \left[ \ketbra{T_1}{T_1}_\sfA \otimes \ketbra{T_2}{T_2}_\sfB \right]$ and $\rho_{\sfA\sfB} = \Ex_{T \gets[0:2t]^d} \left[ \ketbra{T}{T}_{\sfA\sfB} \right]$.  
\end{proof}

\section{Impossibilities of QCCC Primitives in the CHS model} \label{sec:Imp_QCCC_CHS}

\sloppy In this section, we investigate the impossibility of \emph{statistically} secure quantum-computation classical-communication (QCCC) primitives in the CHS model. A recent work by Khurana and Tomer~\cite{KT24} proposed the notion of \emph{one-way puzzles}, which involves a QPT sampler that outputs a classical puzzle-solution pair $(\puz,\sol)$ satisfying a relation, which may not be efficiently computable. In addition, they show that many QCCC primitives imply one-way puzzles. In a very recent work by~Chung, Goldin and Gray~\cite{CGG24}, the authors observed that certain QCCC primitives possess an efficient verification algorithm, and they defined a special class of one-way puzzles called \emph{efficiently verifiable one-way puzzles}. In particular, since we are considering impossibility results, we will focus on the following (fairly weak) form of one-way puzzles in the CHS model.

\begin{definition}[One-way puzzles in the CHS model]
A one-way puzzle is a pair of sampling and verification algorithms $(\Samp, \Ver)$ with the following syntax. Let $q = q(\secp)$ be an arbitrary polynomial and $n = n(\secp) = \omega(\log(\secp))$.

\begin{itemize}
\item $\Samp(1^\secp, \rho) \to (\puz, \sol)$, is a (possibly time-inefficient) quantum algorithm that on input the security parameter and a $2^{nq}$-dimensional quantum state $\rho$ (ideally, $\rho$ will be $q$ copies of an $n$-qubit Haar state), outputs a pair of \emph{classical} strings $(\puz, \sol)$. We refer to $\puz$ as the \emph{puzzle} and $\sol$ as its \emph{solution}. 
\item $\Ver(\puz, \sol, \rho) \to \top$ or $\bot$, is a (possibly time-inefficient) quantum algorithm that on input any pair of classical strings $(\puz, \sol)$ and a $2^{nq}$-dimensional quantum state $\rho$ (ideally, $\rho$ is the same state used to generate the puzzle), outputs either $\top$ (indicating accept) or $\bot$ (indicating reject).

\end{itemize}
These satisfy the following properties.
\begin{itemize}
\item \textbf{Completeness.} The correctness guarantee states that as long as $\Samp$ and $\Ver$ get the same copy of $\rho$, which in turn is $q$ copies of an $n$-qubit Haar state, the output of the sampler will pass the verification with overwhelming probability. That is,
    \[
    \Pr\left[ 
    \Ver(\puz, \sol, \ket{\psi}^{\otimes q}) = \top: 
    \substack{ \ket{\psi}\gets\Haar_n, \\ (\puz, \sol) \gets \Samp(1^\secp, \ket{\psi}^{\otimes q}) } 
    \right] 
    = 1 - \negl(\secp).
    \]
    \item \textbf{Security.} Given $\puz$, it is statistically infeasible to find $\sol$ satisfying $\Ver(\puz, \sol) = \top$, \ie for every unbounded adversary $\alice$,\footnote{Note that the security definition is weak in the sense that the adversary is \emph{not} given any copy of the common Haar state.}
    \[
    \Pr\left[
    \Ver(\puz, \sol', \ket{\psi}^{\otimes q}) = \top:
    \substack{ \ket{\psi}\gets\Haar_n, \\ (\puz, \sol) \gets \Samp(1^\secp, \ket{\psi}^{\otimes q}), \\
    \sol' \gets \alice(\puz) }
    \right] = \negl(\secp).
    \]
\end{itemize}
\end{definition}

\begin{definition}[QCCC key agreements in the CHS model]
A QCCC key agreement in the CHS model is a two-party interactive protocol consisting of a pair of QPT algorithms $(\alice,\bob)$ with their communication being classical. Let $q = q(\secp)$ be an arbitrary polynomial and $n = n(\secp) = \omega(\log(\secp))$. $\alice,\bob$ each take as input the security parameter $1^\secp$ and a $2^{nq}$-dimensional quantum state (ideally, $\alice$ and $\bob$ each obtain $q$ copies of an $n$-qubit Haar state), and outputs classical keys $k_\alice\in\bit$ and $k_\bob\in\bit$ respectively.\footnote{Since we are proving negative results, we assume that the key space of the key agreement is $\bit$, \ie a bit agreement.}
\begin{itemize}
    \item \textbf{Completeness.} There exists a negligible function $\negl$ such that for all $\secp\in\N$,
    \[
    \Pr\left[
    k_\alice = k_\bob: 
    \substack{ \ket{\psi}\gets\Haar_n, \\
    (k_\alice, k_\bob, \tau) \gets \langle \alice(1^\secp, \ket{\psi}^{\otimes q}), \bob(1^\secp, \ket{\psi}^{\otimes q}) \rangle }
    \right] 
    \geq 1 - \negl(\secp),
    \]
    where $\langle\alice,\bob\rangle$ denote the execution of the protocol and $\tau$ is the transcript of the protocol.
    \item \textbf{Statistical Security.} For every computationally unbounded eavesdropper $\eve$, there exists a negligible function $\negl$ such that for all $\secp\in\N$,\footnote{Similarly, we consider a weak security definition in which the eavesdropper is not given any common Haar state.}
    \[
    \Pr\left[
    k_\eve = k_\bob:
    \substack{ \ket{\psi}\gets\Haar_n, \\
    (k_\alice, k_\bob, \tau) \gets \langle \alice(1^\secp, \ket{\psi}^{\otimes q}), \bob(1^\secp, \ket{\psi}^{\otimes q}) \rangle, \\
    k_\eve \gets \eve(1^\secp, \tau) }
    \right] \leq \frac{1}{2} + \negl(\secp).
    \]
\end{itemize}
\end{definition}

\noindent There are various definitions of binding for quantum commitments in literature. Since we are showing impossibility, we focus on sum-binding, which is implied by binding for classical commitments. Similarly, we assume that the input space is $\bit$, \ie a bit commitment.

\begin{definition}[QCCC interactive commitments in the CHS model]
A QCCC commitment in the CHS model is a two-party interactive protocol consisting of a pair of QPT algorithms $(C,R)$, where $C$ is the committer and $R$ is the receiver, with their communication being classical. Let $q = q(\secp)$ be an arbitrary polynomial and $n = n(\secp) = \omega(\log(\secp))$. 
\begin{itemize}
    \item \textbf{Commit Phase:} In the (possibly interactive) commit phase, $C$ takes as input the security parameter $1^\secp$, a bit $b\in\bit$ and a $2^{nq}$-dimensional quantum state $\rho_C$, and $R$ takes as input the security parameter $1^\secp$ and a $2^{nq}$-dimensional quantum state $\rho_R$ (ideally, $C$ and $R$ each obtain $q$ copies of an $n$-qubit Haar state). We denote the execution of the commit phase by $(\sigma_{CR},\tau) \gets \commit\langle C(1^\secp,b,\rho_C), R(1^\secp,\rho_R) \rangle$, where $\sigma_{CR}$ is the joint state of $C$ and $R$ after the commit phase, and $\tau$ denotes the transcript in the commit phase.
    \item \textbf{Reveal Phase:} In the (possibly interactive) reveal phase, the output is $\mu\in\set{0,1,\bot}$ indicating the receiver's output bit or abort. We denote the execution of the reveal phase by $\mu \gets \reveal\langle C, R, \sigma_{CR}, \tau \rangle$.
\end{itemize}
The scheme satisfies the following conditions.
\begin{itemize}
    \item \textbf{Completeness.} There exists a negligible function $\negl$ such that for all $\secp\in\N$,
    \[
    \Pr\left[
    \mu = b: 
    \substack{ \ket{\psi}\gets\Haar_n, \\
    b \gets \bit, \\
    (\sigma_{CR},\tau) \gets \commit\langle C(1^\secp,b,\ket{\psi}^{\otimes q}), R(1^\secp,\ket{\psi}^{\otimes q}) \rangle, \\
    \mu \gets \reveal\langle C, R, \sigma_{CR}, \tau \rangle, \\
    \mu\in\set{0,1,\bot}
    }
    \right] 
    \geq 1 - \negl(\secp).
    \]
    \item \textbf{Statistical Hiding.} For every computationally unbounded malicious receiver $R^*$, there exists a negligible function $\negl$ such that for all $\secp\in\N$,
    \[
    \Pr\left[
    b' = b:
    \substack{ \ket{\psi}\gets\Haar_n, \\
    b \gets \bit, \\
    (\sigma_{CR^*},\tau) \gets \commit\langle C(1^\secp,b,\ket{\psi}^{\otimes q}), R^*(1^\secp,\ket{\psi}^{\otimes q}) \rangle, \\
    b' \gets R^*(\sigma_{R^*}, \tau)
    }
    \right] \leq \frac{1}{2} + \negl(\secp),
    \]
    where $\sigma_{R^*}$ denotes the state obtained by tracing out the committer's part of the state $\sigma_{CR^*}$.
    \item \textbf{Statistical Binding.} For every computationally unbounded malicious committer $C^*$, there exists a negligible function $\negl$ such that for all $\secp\in\N$,
    \[
    \Pr\left[
    \mu = b:
    \substack{ \ket{\psi}\gets\Haar_n, \\
    (\sigma_{C^*R},\tau) \gets \commit\langle C^*(1^\secp,\ket{\psi}^{\otimes q}), R(1^\secp,\ket{\psi}^{\otimes q}) \rangle, \\
    b \gets \bit, \\
    \mu \gets \reveal\langle C^*(b), R, \sigma_{C^*R}, \tau \rangle
    }
    \right] \leq \frac{1}{2} + \negl(\secp).
    \]
\end{itemize}
\end{definition}

\begin{theorem} \label{thm:Imp_QCCC_CHS}
There does not exist primitive $\cP$ in the CHS model where $\cP \in$ \{one-way puzzles, statistically secure QCCC key agreements, statistically hiding and statistically binding QCCC interactive commitments\}.
\end{theorem}
\paragraph{Proof intuition.} The high-level idea is to convert the scheme in the CHS model to a scheme in the \emph{plain} model. In the CHS model, given a pair of algorithms, we define the new pair of (time-inefficient) algorithms to be identical except for their input, which consists of copies of two i.i.d Haar states. Thanks to the LOCC Haar indistinguishability (\Cref{thm:LOCC_main}), the expense of doing so is to only increase the completeness error and security loss by a negligible amount. Therefore, if there were to exist a complete and secure scheme in the CHS model, it would imply the existence of such a scheme in the plain model, contradicting the trivial impossibility.\footnote{The impossibilities in the plain model still hold even when the algorithms of the primitives are time-inefficient.}
\begin{proof}[Proof of~\Cref{thm:Imp_QCCC_CHS}] $ $ \newline
\textbf{One-way puzzles.} Suppose there exists a one-way puzzle $(\Samp,\Ver)$ in the CHS model. We define $(\wt{\Samp},\wt{\Ver})$ as follows. $\wt{\Samp}(1^{\secparam})$ simply samples $q$ copies of a Haar state $\ket{\psi}$ and runs $\Samp(1^{\secparam},\ket{\psi}^{\otimes q})$. $\wt{\Ver}(\puz,\sol)$ is defined similarly; it samples $q$ copies of a Haar state $\ket{\phi}$ and then runs $\Ver(\puz,\sol,\ket{\phi}^{\otimes q})$. It is important to note that $\wt{\Samp}$ and $\wt{\Ver}$ sample the Haar states independently. 

First, we claim that $(\wt{\Samp},\wt{\Ver})$ has negligible completeness error. Otherwise, we construct an LOCC distinguisher $(\alice,\bob)$ with a non-negligible advantage for the task in~\Cref{thm:LOCC_main}. $\alice$ runs $\Samp$ on the security parameter and her input, obtains $(\puz,\sol)$, and sends $(\puz,\sol)$ to $\bob$. Then $\bob$ runs $\Ver$ on $(\puz,\sol)$ and his input, and outputs $1$ if the verification passes. If the input of $(\alice,\bob)$ is $\rho$ (defined in~\Cref{thm:LOCC_main}, \ie each is given $q$ copies of the same Haar state), then the probability of $\bob$ outputting $1$ is equal to the completeness of $(\Samp,\Ver)$. Similarly, if the input is $\sigma$ (\ie each is given $q$ copies of two i.i.d Haar states), then the probability of $\bob$ outputting $1$ is equal to the completeness of $(\wt{\Samp},\wt{\Ver})$. Hence, $(\alice,\bob)$ has a non-negligible advantage by the premise. However, this contradicts~\Cref{thm:LOCC_main}. 

Next, we claim that $(\wt{\Samp},\wt{\Ver})$ satisfies security. Suppose there is an adversary $\wt{\eve}$ that breaks the security of $(\wt{\Samp},\wt{\Ver})$ with a non-negligible advantage of $\wt{\veps} = \wt{\veps}(\secp)$. We claim that $\wt{\eve}$ breaks the security of $(\Samp,\Ver)$ with an advantage of $\veps = \veps(\secp)$ satisfying $|\eps - \wt{\veps}|=\negl(\secp)$, which means that $\wt{\veps}$ is non-negligible as well. Otherwise, suppose $|\eps - \wt{\veps}|$ is non-negligible, we can construct an LOCC distinguisher $(\alice,\bob)$ as follows. 
\begin{itemize}
    \item $\alice$ runs $\Samp$ on the security parameter and her input to obtain $(\puz,\sol)$. It then runs $\wt{\eve}$ on $\puz$ to obtain $\sol'$. Finally, it sends $(\puz,\sol')$ to $\bob$. 
    \item $\bob$ runs $\Ver$ on $(\puz,\sol')$. If the output is $\bot$, it outputs 1. Otherwise, it outputs 0. 
\end{itemize}
\noindent If the input of $(\alice,\bob)$ is $\rho$ (defined in~\Cref{thm:LOCC_main}, \ie each is given $q$ copies of the same Haar state), then the probability of $\bob$ outputting $1$ is equal to $\veps$. Similarly, if the input is $\sigma$ (\ie each is given $q$ copies of two i.i.d Haar states), then the probability of $\bob$ outputting $1$ is $\wt{\veps}$. Again, this contradicts~\Cref{thm:LOCC_main}.

So far, we have shown that $(\wt{\Samp},\wt{\Ver})$ satisfies completeness and security in the plain model. However, such a scheme cannot exist. This is because an unbounded adversary, given a puzzle, can find the solution with the highest probability of passing the verification to break the security. Hence, we conclude that $(\Samp,\Ver)$ is not a one-way puzzle in the CHS model.

The structure of proving the impossibility of key agreements and interactive commitments is very similar. We only describe the LOCC distinguishers and omit the full details.\\

\noindent \textbf{Key agreements.} Suppose $\KA = (P_1,P_2)$ is a statistically secure QCCC key agreement in the CHS model. Define $\wt{\KA} = (\wt{P_1},\wt{P_2})$ such that $\wt{P_1}$ (resp., $\wt{P_2}$) samples $q$ copies of a Haar state and they run $P_1$ (resp., $P_2$). We argue that $\wt{\KA}$ satisfies both completeness and security. 
\par Suppose completeness error of $\wt{\KA}$ is inverse polynomial (in $\secparam$), we define an LOCC adversary $(\alice,\bob)$ as follows. Upon receiving a bipartite state on registers $\sfA$ and $\sfB$,  $\alice$ (resp., $\bob$) runs $P_1$ (resp., $P_2$) on input $1^{\secparam}$ and the register $\sfA$ (resp., $\sfB$). Then, $\alice$ obtains the key $k_{P_1}$ and $\bob$ obtains the key $k_{P_2}$. They perform an extra round of communication to check if $k_{P_1} = k_{P_2}$. Similar to the argument for one-way puzzles, it can be shown that $(A,B)$ can distinguish $\rho$ and $\sigma$ (defined in~\Cref{thm:LOCC_main}) with inverse polynomial probability, which is a contradiction. 
\par Suppose $\wt{\KA}$ is not statistically secure. That is, there exists an eavesdropper $\wt{E}$ that can break the security of $\wt{\KA}$ with inverse polynomial (in $\secparam$) probability. Using $\wt{E}$, we define an LOCC adversary $(A,B)$, who upon receiving a bipartite state on two registers $\sfA$ and $\sfB$ do the following. 
\begin{itemize}
    \item $A$ runs $P_1$ on $1^{\secparam}$ and the register $\sfA$. Similarly, $B$ runs $P_2$ on $1^{\secparam}$ and the register $\sfB$. Denote $\tau$ be the transcript of the protocol. 
    \item $B$ runs $\wt{E}(\tau)$ to obtain $k_{E}$. It then checks if $k_{E}=k_{P_2}$. If so, it outputs 1. Otherwise, it outputs 0. 
\end{itemize}
Similarly, as before, we can show that $(A,B)$ succeeds in distinguishing $\rho$ and $\sigma$ with inverse polynomial probability, a contradiction. 
\par So far, we have shown that $\widetilde{\KA}$ is a key agreement protocol in the plain model that satisfies both completeness and statistical security. However, such a scheme cannot exist which further means that $\KA$ either does not satisfy completeness or security. \\

\noindent \textbf{Interactive Commitments.} Suppose $\Com = (C,R)$ is a statistically hiding and statistically binding QCCC interactive commitment in the CHS model. We define $\wt{\Com} = (\wt{C},\wt{R})$ as follows. Upon receiving the input bit $b\in\bit$, $\wt{C}$ simply samples 
$q$ copies of a Haar state $\ket{\psi}$ and runs $C$ on input $1^\secp$, $b$ and $\ket{\psi}^{\otimes q}$. Similarly, $\wt{R}$ samples $q$ copies of a Haar state $\ket{\phi}$ and runs $R$ on input $1^\secp$ and $\ket{\phi}^{\otimes q}$.
\par Intuitively, $\wt{\Com}$ is at least as secure as $\Com$ since the malicious party in $\wt{\Com}$ has no information about the other party's Haar state as opposed to $\Com$. Suppose $\wt{\Com}$ is not statistically hiding. That is, there exists a malicious receiver $\wt{R}^*$ that can break the statistical hiding of $\wt{\Com}$ with inverse polynomial (in $\secparam$) probability. Using $\wt{R}^*$, we define a malicious receiver $R^*$ that breaks the statistical hiding of $\Com$. $R^*$ simply discards its common Haar states and runs $\wt{R}^*$. Then the distinguishing advantage of $R^*$ is identical to that of $\wt{R}^*$, which is a contradiction. 
Suppose $\wt{\Com}$ is not statistically binding. That is, there exists a malicious receiver $\wt{C}^*$ that can break the statistical binding of $\wt{\Com}$ with inverse polynomial (in $\secparam$) probability. Similarly, discarding the common Haar states and using $\wt{C}^*$ breaks the statistical binding of $\Com$, which is a contradiction.
\par Suppose completeness error of $\wt{\Com}$ is inverse polynomial (in $\secparam$), we define an LOCC adversary $(\alice,\bob)$ as follows. Upon receiving a bipartite state on registers $\sfA$ and $\sfB$,  $\alice$ (resp., $\bob$) runs $C$ (resp., $R$) on input $1^{\secparam}$, a uniform bit $b\in\bit$, and the register $\sfA$ (resp., $\sfB$). Then, $\bob$ obtains $\mu$. They perform an extra round of communication to check if $b = \mu$. Similar to the argument for one-way puzzles, it can be shown that $(A,B)$ can distinguish $\rho$ and $\sigma$ (defined in~\Cref{thm:LOCC_main}) with inverse polynomial probability, which is a contradiction.
\end{proof}

\section{Quantum Black-Box Separation in the QCCC Model}
\label{sec:QBB_QCCC}

\subsection{The Separating Oracle} 
\label{sec:separating_oracle}
As is common in black-box impossibility results, we will define oracles relative to which $\omega(\log(\secp))$-PRSGs exist while QCCC key agreements and interactive commitments do not. We define the oracle $G := \set{\set{G_k}_{k \in \bit^\secp}}_{\secp\in\N}$ as follows. For every $\secp\in\N$ and $k\in\bit^\secp$, the oracle $G_k$ is a Haar isometry that maps any state $\ket{\psi}$ to $\ket{\psi}\ket{\vartheta_k}$, where $\ket{\vartheta_k}$ is a Haar state of length $n(\secp) = \omega(\log(\secp))$. The existence of $\omega(\log(\secp))$-PRSGs relative to $G$ can be proven easily.

\begin{lemma}[$(\secp,\omega(\log(\secp)))$-PRSGs exist relative to $G$] \label{lem:oracle_implies_PRS}
There exists a $(\secp,\omega(\log(\secp)))$-PRSG relative to $G$. In particular, for any polynomial $q(\cdot)$ and any computationally unbounded adversary $\alice^G$ that takes as input $1^\secp$ and asks $q(\secp)$ quantum queries to $G$, the distinguishing advantage is negligible in $\secp$.
\end{lemma}
\begin{proof}[Proof sketch]
The proof is similar to the proof of~\cite[Lemma~30]{Kretschmer21}. The implementation of the PRSG is simply the oracle $G$: on input $k$, outputs the state $\ket{\vartheta_k}$ generated by $G_k$. The security follows from the hardness of the unstructed search problem~\cite{BBBV97}.
\end{proof}

\subsection{Separating QCCC Key Agreements from $(\secp,\omega(\log(\secp)))$-PRSGs}
\begin{definition}[QCCC key agreements relative to oracle]
A \emph{QCCC key agreement relative to an oracle $\cO$} is a two-party interactive protocol consisting of a pair of uniform quantum (possibly time-inefficient) oracle algorithms $(\alice,\bob)$ such that $\alice,\bob$ each take as input the security parameter $1^\secp$, ask $q(\secp)$ queries to the oracle $\cO$ for some polynomial $q$, communicate classically, and output the classical keys $k_\alice\in\bit$ and $k_\bob\in\bit$ respectively. An \emph{$(\veps,p,\delta)$-QCCC key agreement relative to $\cO$} satisfies the following:
\begin{itemize}
    \item \textbf{$\veps$-completeness.} We say that a QCCC key agreement is \emph{$\veps$-complete} if the following holds for all $\secp\in\N$,
    \[
    \Pr\left[
    k_\alice = k_\bob: 
    \substack{ O \gets \cO, \\
    (k_\alice, k_\bob, \tau) \gets \langle \alice^O(1^\secp), \bob^O(1^\secp) \rangle }
    \right] 
    \geq 1 - \veps(\secp),
    \]
    where $\langle\alice,\bob\rangle$ denote the execution of the protocol and $\tau$ is the transcript of the protocol. We anticipate that $\veps$ is negligible.
    \item \textbf{$(p,\delta)$-security.} We say that a QCCC key agreement is \emph{$(p,\delta)$-secure} if for any computationally unbounded eavesdropper $\eve$ that on input $1^\secp$ and transcript $\tau$ and asks at most $p(\secp)$ classical queries to $\cO$, the following holds for all $\secp\in\N$,
    \[
    \Pr\left[
    k_\eve = k_\bob:
    \substack{ O \gets \cO, \\
    (k_\alice, k_\bob, \tau) \gets \langle \alice^O(1^\secp), \bob^O(1^\secp) \rangle, \\
    k_E \gets \eve^O(1^\secp,\tau) }
    \right] 
    \leq \frac{1}{2} + \delta(\secp).
    \]
    We anticipate that for any polynomial $p$, there exists a negligible $\delta$ such that the key agreement is $(p,\delta)$-secure.
\end{itemize}
\end{definition}
\noindent In the plain model, completeness and security are defined similarly in the absence of an oracle. In particular, a QCCC key agreement is an \emph{$(\veps,\delta)$-QCCC key agreement} if it satisfies $\veps$-completeness and $\delta$-security.

\begin{lemma}[Conditional independence] \label{lem:Cond_Indep}
For any two-party interactive QCCC protocol $(\alice,\bob)$ where the party's initial state is a product state, the joint state at the end of each round  $i$ can be written as 
\[
\sum_{t^i} p_{t^i} \ketbra{t^i}{t^i}_\sfT \otimes \rho^{t^i}_{\sfA\sfB},
\]
for some partial transcripts $t^i := (t_1,t_2,\dots,t_i)$ until round $i$ and product states $\rho^{t^i}_{\sfA\sfB}$, where register $\sfT$ is for storing the transcript.
\end{lemma}
\begin{proof}
We prove it by induction on rounds. Initially, the joint state is $\ketbra{\bot}{\bot}_\sfT \otimes \rho^{\bot}_{\sfA}\otimes \rho^{\bot}_{\sfB}$ by the premise, where $\bot$ denotes the empty transcript. Suppose after the $j$-th round, the joint state is $\sum_{t^j} p_{t^j} \ketbra{t^j}{t^j}_{\sfT} \otimes \rho^{t^j}_{\sfA}\otimes \rho^{t_j}_{\sfB}$. In the $(j+1)$-th round (suppose it is $\alice$'s round), $\alice$ will first apply a unitary controlled by $t^j$ of the form $\sum_{t^j}\ketbra{t^j}{t^j}_\sfT\otimes U^{(t^j)}_\sfA$  and then perform the measurement to generate the message $t_{j+1}$ of this round. Then the state $\rho^{t^{j}}_{\sfA}$ becomes $\sum_{t_{j+1}}\ketbra{t_{j+1}}{t_{j+1}}_{\sfT_{j+1}}\otimes \left({(\bra{t_{j+1}}\otimes I)} U^{(t^j)}_\sfA \rho^{t^j}_{\sfA} {\left(U^{(t^j)}_\sfA\right)}^{\dagger}{(\ket{t_{j+1}}\otimes I)}\right)$, where register $\sfT_{j+1}$ is appended to the transcript register. We can write $\left({(\bra{t_{j+1}}\otimes I)} U^{(t^j)}_\sfA \rho^{t^j}_{\sfA} {\left(U^{(t^j)}_\sfA\right)}^{\dagger}{(\ket{t_{j+1}}\otimes I)}\right)$ as $p(t_{j+1}|t^j) \cdot \rho_{\sfA}^{t^j||{t_{j+1}}}$, where $p(t_{j+1}|t^j)$ is the probability of getting the outcome $t_{j+1}$ by measuring $U^{(t^j)}_\sfA \rho^{t^j}_{\sfA} {\left(U^{(t^j)}_\sfA\right)}^{\dagger}$ in the computational basis. Hence we get that the final state is 
\[\sum_{t^j}\sum_{t_{j+1}} p_{t^j} p(t_{j+1}|t^j) \cdot \ketbra{t^j||t_{j+1}}{t^j||t_{j+1}}_{\sfT} \otimes \rho^{t^j||t_{j+1}}_{\sfA}\otimes \rho^{t_j}_{\sfB},
\]
which is still a product state for any $t^{j+1} = t^{j}||t_{j+1}$.
\end{proof}

\begin{lemma}[Impossibility of key agreements in the plain model] \label{lem:Imp_KA_plain}
For any $\veps,\delta:\N\to[0,1]$ and $(\veps,\delta)$-QCCC key agreement in the plain model, it holds that $\veps(\lambda) + \delta(\lambda) \geq 1/2$ for any $\secp\in\N$.
\end{lemma}
\begin{proof}
Let $\KA$ be an $(\veps,\delta)$-QCCC key agreement in the plain model that outputs $(\tau,k_\alice,k_\bob)$. Fix $\secp$ for the rest of the proof. In execution of $\KA$, equivalently, we can first sample $\tau$, then sample $k_\bob$ conditioned on $\tau$, and finally sample $k_\alice$ conditioned on $(\tau,\bob)$. For any fixed $\tau$ in the support, by~\Cref{lem:Cond_Indep}, the joint state of $\alice$ and $\bob$ is a product state. Thus, further fixing $k_\bob$ won't change the marginal distribution of $k_\alice$. In the rest of the proof, we fix $\tau$ and $k_\bob$.

Consider the following eavesdropper $\eve$. Upon receiving the transcript $\tau$, $\eve$ runs the protocol coherently and computes the post-measurement state conditioned on $\tau$. Then $\eve$ sets $k_\eve$ to $k_\alice$ computed from the final joint state. Hence, the distribution of $k_\eve$ is identically distributed to the marginal distribution of $k_\alice$ in $\KA$ conditioned on $\tau$. That is, the probability of $k_\eve = k_\bob$ is equal to that of $k_\alice = k_\bob$. Finally, averaging over $(\tau,k_\bob)$, we have the probability of $k_\eve = k_\bob$ is $\geq 1 - \veps(\secp) = 1/2 + (1/2 - \veps(\secp))$ from the $\veps$-completeness of $\KA$. In other words, $\delta(\secp)$ must be $\geq 1/2 - \veps(\secp)$. Hence, we have $\delta(\secp) + \veps(\secp) \geq 1/2$ for any $\secp\in\N$.
\end{proof}

\begin{theorem}[Quantum state tomography~\cite{OW16}] \label{thm:tomography}
There exists an algorithm $\tomography$ and a polynomial $p_\tomography$ satisfy the following. For any $d\in\N,\Delta, \gamma \in (0,1]$ and $d$-dimensional pure quantum state $\ketbra{\psi}{\psi}$, given $p_\tomography(d,\Delta^{-1},\log(\gamma^{-1}))$ copies of  $\ketbra{\psi}{\psi}$, $\tomography$ outputs the classical description of $\ketbra{\wh{\psi}}{\wh{\psi}}$ satisfying $\TD(\ketbra{\psi}{\psi},\ketbra{\wh{\psi}}{\wh{\psi}}) \leq \Delta$ with probability at least $1 - \gamma$.
\end{theorem}

\begin{lemma}[Compling out $G$ from $\KA^G$] \label{lem:compile_KA}
If QCCC key agreements relative to $G$ (the keyed common Haar state oracle defined in~\Cref{sec:separating_oracle}) exist, then there exists an $(\veps,\delta)$-QCCC key agreement in the plain model such that $\veps(\secp)$ is an inverse polynomial and $\delta(\secp) \leq 0.2$ for sufficiently large $\secp\in\N$.
\end{lemma}
\begin{proof}
Let $\KA^G = (\alice^G,\bob^G)$ be a QCCC key agreement relative to $G = \set{\set{G_k}_{k \in \bit^\secp}}_{\secp\in\N}$ in which $\alice$ and $\bob$ each ask $q(\secp) = \poly(\secp)$ queries with the maximum input length of the queries being $L(\secp) = \poly(\secp)$. Define $\Lambda(\secp) := \ceil{\log(q^{10} + L^{10} + \secp^{10})} = O(\log(\secp))$ and the ``truncated'' oracle $G_\Lambda = \set{\set{G_k}_{k\in\bit^i}}_{i=1}^\Lambda$. We define the following hybrid protocol $\wt{\KA}^{G_\Lambda} = (\wt{\alice}^{G_\Lambda},\wt{\bob}^{G_\Lambda})$:

\begin{mdframed}
$\wt{\KA}^{G_\Lambda}(1^\secp,\wt{\alice}^{G_\Lambda},\wt{\bob}^{G_\Lambda})$:
\begin{enumerate}
\item For every $k\in\bigcup_{i=\Lambda+1}^L\bit^i$, $\wt{\alice}$ and $\wt{\bob}$ samples $\ket{\phi^\alice_k},\ket{\phi^\bob_k}\gets\Haar_{|k|}$ respectively.
\item $(\wt{\alice}^{G_\Lambda},\wt{\bob}^{G_\Lambda})$ runs $(\alice^G, \bob^G)$ on $1^\secp$ by answering the queries as follows: Suppose $\alice$ asks a query $k\in\bigcup_{i=1}^L\bit^i$. If $|k|\leq \Lambda$, then $\wt{\alice}$ asks $k$ to oracle $G_\Lambda$ and forwards the response. Otherwise, $\wt{\alice}$ sends $\ket{\phi^\alice_k}$ to $\alice$. $\wt{\bob}$ answers the queries of $\bob$ similarly by replacing $\ket{\phi^\alice_k}$ with $\ket{\phi^\bob_k}$.
\item $\wt{\alice}$ outputs the key $k_\alice$ generated by $\alice$ and $\wt{\bob}$ outputs key $k_\bob$ generated by $\bob$.
\end{enumerate}
\end{mdframed}
\paragraph{$\wt{\KA}^{G_\Lambda}$ is query-efficient.}
Since $(\wt{\alice}^{G_\Lambda},\wt{\bob}^{G_\Lambda})$ needs to sample Haar states in Step~1, $\wt{\KA}^{G_\Lambda}$ is not time-efficient. However, each of $\wt{\alice}^{G_\Lambda},\wt{\bob}^{G_\Lambda}$ makes at most $q$ queries in Step~2 in $\wt{\KA}^{G_\Lambda}$.

\paragraph{$\wt{\KA}^{G_\Lambda}$ is $1/\poly$-complete.}
First, we prove that $\wt{\KA}^{G_\Lambda}$ satisfies completeness. The idea is similar to the proof of~\Cref{thm:Imp_QCCC_CHS}. Define LOCC distinguisher $(\alice_\LOCC,\bob_\LOCC)$ for the task in~\Cref{cor:LOCC_repetition} with the following parameters: $t = 2q$, $n_{2qi+j} = \Lambda + i + 1$ for $i = 0,1,\dots,L - \Lambda - 1$ and $j = 1,2,\dots,2q$, and thus $s = (L - \Lambda) \cdot 2q$:\footnote{For $k\in [s]$, we represent the $k$-th state by $\ket{\psi_{i+\Lambda+1}^{j}}$ ($\ket{\phi_{i+\Lambda+1}^{j}}$ resp.) where $i,j$ are determined by uniquely writing $k = 2qi+j$ for $i = 0,1,\dots,L - \Lambda - 1$ and $j = 1,2,\dots,2q$.}
\begin{enumerate}
    \item $\alice_\LOCC$ and $\bob_\LOCC$ receive input register.
    \item $\alice_\LOCC$ samples oracle $G_\Lambda$ and sends its description to $\bob_\LOCC$.
    \item $\alice_\LOCC$ and $\bob_\LOCC$ initialize lists $\cL_{\ell} = \set{(1,\bot),(2,\bot),\dots,(2q,\bot)}$ for answering queries of different lengths $\ell = \Lambda + 1,\Lambda + 2,\dots,\lambda$ (let $\cL := \set{\cL_{\ell}}_{\ell\in [\Lambda+1:L]}$), and runs $\KA^{(\cdot)} = (\alice^{(\cdot)},\bob^{(\cdot)})$ on $1^\secp$ by lazy evaluation and jointly maintaining the list $\cL$ as follows: \\
    
    In the $r$-th round (suppose it's $\alice_\LOCC$'s round), upon received the message $t_{r-1}$ and list $\cL$ from $\bob_\LOCC$ in the $(r-1)$-th round, $\alice_\LOCC$ feeds $t_{r-1}$ to $\alice$.\footnote{In the first round (suppose it's $\alice_\LOCC$'s round), $\alice_\LOCC$ simply runs $\alice$ on input the security parameter and $t_0 := \bot$.} Upon receiving $\alice$'s query $x \in \bigcup_{i=1}^L \bit^i$, if $|x|\leq \Lambda$, then $\alice_\LOCC$ uses $G_\Lambda$ to answer the query. Otherwise, $\alice_\LOCC$ checks if $(i,x)$ is in $\cL_{|x|}$ for some $i\in[2q]$ (\ie whether $x$ has already been queried by $\alice$ or $\bob$). If $(i,x)\in\cL_{|x|}$, then $\alice_\LOCC$ answers the query using a copy of $\ket{\psi_{|x|}^i}$. Otherwise, $\alice_\LOCC$ finds the first index $i\in[2q]$ such that $(i,\bot)\in\cL_{|x|}$, updates it into $(i,x)$, and answers the query using a copy of $\ket{\psi_{|x|}^i}$. At the end of the round, $A$ outputs a classical message $t_r$. Then $\alice_\LOCC$ sends $\cL$ and $t_r$ to $\bob_\LOCC$.\footnote{In $\bob_\LOCC$'s round, $\bob$ acts similarly as defined above.}
    \item At the end of the protocol, $\alice, \bob$ outputs the keys $k_\alice, k_\bob$ respectively.
    \item $\alice_\LOCC$ sends $k_\alice$ to $\bob_\LOCC$, and $\bob_\LOCC$ outputs $1$ if $k_\alice = k_\bob$.
\end{enumerate}
Hence, $(\alice, \bob)$ asks at most $2q$ queries in total, $(\alice_\LOCC, \bob_\LOCC)$ perfectly simulates either $\KA^G$ or $\wt{\KA}^{G_\Lambda}$ depending on if they obtained the same states or i.i.d. states. Hence, 
by~\Cref{cor:LOCC_repetition} we have
\[
\left| \Pr_{\KA^G}[k_\alice = k_\bob] - \Pr_{\wt{\KA}^{G_{\Lambda}}}[k_\alice = k_\bob] \right|
\leq O\left( \sum_{n = \Lambda+1}^L 2q \cdot \frac{(2q)^2}{2^n} \right)
\leq O\left( L \cdot \frac{q^3}{2^\Lambda} \right),
\]
which implies $\Pr_{\wt{\KA}^{G_\Lambda}}[k_\alice = k_\bob] \geq \Pr_{\KA^G}[k_\alice = k_\bob] - O\left( L q^3/2^\Lambda \right) = 1 - 1/\poly(\secp)$ for some polynomial $\poly$.

\paragraph{$\wt{\KA}^{G_\Lambda}$ is $0.1$-secure.}
Next, we claim that for any polynomial $p$ and eavesdropper that asks $p(\secp)$ classical queries to $G_\Lambda$, her advantage of finding $k_\bob$ in $\wt{\KA}^{G_\Lambda}$ is at most $0.1$ for sufficiently large $\secp$. For contradiction, suppose there exist a polynomial $p$ and an eavesdropper $\wt{\eve}$ that asks $p(\secp)$ classical queries to $G_\Lambda$ and finds $k_\bob$ with advantage at least $0.1$ for infinitely many $\secp$ in $\wt{\KA}^{G_\Lambda}$. Then we construct following the LOCC distinguisher: $(\alice_\LOCC,\bob_\LOCC)$ first run $\KA^G$ as the previous paragraph and obtains $k_\alice,k_\bob$ and the transcript $\tau$. Then $\bob_\LOCC$ runs $\wt{\eve}$ on input the transcript $\tau$, answers the queries by $G_\Lambda$ defined by themselves (without using any input state), and obtains a key $k_\eve$. $\bob_\LOCC$ outputs $1$ if $k_\bob = k_\eve$. By the same argument, $(\alice_\LOCC, \bob_\LOCC)$ perfectly simulates either $\wt{\eve}^G$ in $\KA^G$ or $\wt{\eve}^{G_\Lambda}$ in $\wt{\KA}^{G_\Lambda}$ depending on if they got the same states or i.i.d. states. Hence, by~\Cref{cor:LOCC_repetition} we have
\[
\left| \Pr_{\KA^G}[k_\bob = k_\eve] - \Pr_{\wt{\KA}^{G_{\Lambda}}}[k_\bob = k_\eve] \right|
\leq O\left( \sum_{n = \Lambda+1}^L 2q \cdot \frac{(2q)^2}{2^n} \right)
\leq O\left( L \cdot \frac{q^3}{2^\Lambda} \right),
\]
which implies $\Pr_{\wt{\KA}^{G_\Lambda}}[k_\bob = k_\eve] \geq \Pr_{\KA^G}[k_\bob = k_\eve] - O\left( L q^3/2^\Lambda \right) \geq 0.1 - O\left( L q^3/2^\Lambda \right)$ for infinitely many $\secp$. However, this contradicts the security of $\KA^G$.
\paragraph{Getting to plain model:}
Finally, define the following protocol $\KA_{\plain}(\alice_\plain,\bob_\plain)$ in the plain model:
\begin{mdframed}
$\KA_{\plain}(1^\secp,\alice_\plain,\bob_\plain)$:
\begin{enumerate}
\item For every $k\in\bigcup_{i=1}^\Lambda\bit^i$, $\alice_\plain$ samples $\ket{\psi_k}\gets\Haar_{|k|}$.
\item For every $k\in\bigcup_{i=1}^\Lambda\bit^i$, $\alice_\plain$ run $\tomography$ (defined in~\Cref{thm:tomography}) on $\ket{\psi_k}$ with parameters $\Delta = 2^{-2\Lambda}$ and $\gamma = 2^{-\secp}$ to obtain the classical description of $\ket{\wh{\psi}_k}$.\footnote{Note that $\alice_\plain$ samples $\ket{\psi_k}$ and thus has its classical description. Performing tomography is merely for the simplicity of proof.}
\item For every $k\in\bigcup_{i=1}^\Lambda\bit^i$, $\alice_\plain$ sends the description of $\ket{\wh{\psi}_k}$ to $\bob_\plain$.
\item $(\alice_\plain,\bob_\plain)$ define the output of the oracle $\wh{G}_\Lambda$ to be $\set{\set{\ket{\wh{\psi}_k}}_{k\in\bit^i}}_{i\in\set{1,\dots,\Lambda}}$.
\item $(\alice_\plain,\bob_\plain)$ runs $\wt{\KA}^{\wh{G}_\Lambda}$ on $1^\secp$ to obtain $(k_\alice,k_\bob)$.
\item $\alice_\plain$ outputs key $k_\alice$ and $\bob_\plain$ outputs key $k_\bob$ respectively.
\end{enumerate}
\end{mdframed}

\paragraph{$\KA_{\plain}$ is $1/\poly$-complete.}
Define the event $\good$ in $\KA_\plain$ as:
\[
\good \equiv \bigwedge_{k \in \bigcup^\Lambda_{i=1} \bit^i} \left[ \TD(\ketbra{\psi_k}{\psi_k}, \ketbra{\wh{\psi}_k}{\wh{\psi}_k}) \leq \Delta \right].
\]
From the guarantee of tomography (\Cref{thm:tomography}) and a union bound, the probability of $\good$ happening is at least $1 - \sum_{i=1}^\Lambda 2^i \cdot \gamma = 1 - \negl(\secp)$. 
Since $\wt{\alice}^{(\cdot)}$ and $\wt{\bob}^{(\cdot)}$ in $\wt{\KA}^{G_\Lambda}$ (\resp $\wt{\KA}^{\wh{G}_\Lambda}$) ask a total of $2q$ queries, one can use $\set{\set{\ket{\psi_k}^{\otimes 2q}}_{k\in\bit^i}}_{i=1}^\Lambda$ (\resp $\set{\set{\ket{\wh{\psi}_k}^{\otimes 2q}}_{k\in\bit^i}}_{i=1}^\Lambda$) to perfectly answer $\alice$'s and $\bob$'s queries. Hence, from the operational definition of trace distance, we have
\begin{align*}
& \left| \Pr_{\KA_{\plain}}[k_\alice = k_\bob] - \Pr_{\wt{\KA}^{G_\Lambda}}[k_\alice = k_\bob] \right| \\
& \leq \Pr[\neg\good] + \Ex\left[ \TD\left( \bigotimes_{i=1}^\Lambda\bigotimes_{k\in\bit^i} \ketbra{\psi_k}{\psi_k}^{\otimes 2q}, \bigotimes_{i=1}^\Lambda\bigotimes_{k\in\bit^i} \ketbra{\wh{\psi}_k}{\wh{\psi}_k}^{\otimes 2q} \right) \mid \good \right] \\
& \leq \Pr[\neg\good] + \sum_{i=1}^\Lambda \sum_{k\in\bit^i} 2q \cdot \Ex\left[ \TD(\ketbra{\psi_k}{\psi_k}, \ketbra{\wh{\psi}_k}{\wh{\psi}_k}) \mid \good \right] \\
& \leq \negl(\secp) + 2q \cdot \sum_{i=1}^\Lambda 2^i \cdot \Delta 
\leq \frac{1}{\poly'(\secp)}
\end{align*}
for some polynomial $\poly'$. Hence, the completeness of $\KA_{\plain}$ is at least
\[
\Pr_{\KA_{\plain}}[k_\alice = k_\bob]
\geq \Pr_{\wt{\KA}^{G_\Lambda}}[k_\alice = k_\bob] - \frac{1}{\poly'(\secp)}
= 1 - \frac{1}{\poly(\secp)} - \frac{1}{\poly'(\secp)}
= 1 - \veps(\secp)
\]
for some inverse polynomial $\veps$, where the first equality is because $\wt{\KA}^{G_\Lambda}$ is $1/\poly(\secp)$-complete for some polynomial $\poly$.

\paragraph{$\KA_{\plain}$ is $0.2$-secure.}
For contradiction, suppose there exists an eavesdropper $\eve_\plain$ that finds $k_\bob$ in $\KA_\plain$ with advantage $0.2$ for infinitely many $\secp$. We construct the following eavesdropper $\wt{\eve}^{G_\Lambda}$ for $\wt{\KA}^{G_\Lambda}$ by using $\eve_\plain$ as follows.
\begin{mdframed}
$\wt{\eve}^{G_\Lambda}(1^\secp, \tau)$:
\begin{enumerate}
    \item For every $k\in\bigcup_{i=1}^\Lambda\bit^i$, ask $p_{\tomography}(2^{|k|},\Delta^{-1},\log(\gamma^{-1}))$ queries to $G_k$ with parameters $\Delta = 2^{-2\Lambda}$ and $\gamma = 2^{-\secp}$ to get $\set{\ket{\psi_k}^{\otimes p_{\tomography}(2^{|k|},\Delta^{-1},\log(\gamma^{-1}))}}_{k\in\bigcup_{i=1}^\Lambda\bit^i}$.
    \item Perform $\tomography$ (defined in~\Cref{thm:tomography}) on every state obtained in the previous step to obtain the description of $\set{\ket{\wh{\psi}_k}}_{k\in\bigcup_{i=1}^\Lambda\bit^i}$.
    \item Run $\eve_\plain$ on input $\tau$ and all the descriptions obtained by tomography, and set $k_\eve$ to the output of $\eve_\plain$.
    \item Output $k_\eve$.
\end{enumerate}
\end{mdframed}
First, $\wt{\eve}^{G_\Lambda}$ makes at most $\sum_{i=1}^\Lambda 2^i \cdot p_{\tomography}(2^i,\Delta^{-1},\log(\gamma^{-1})) = p(\secp)$ queries for some polynomial $p$. Next, in $\wt{\KA}^{G_\Lambda}$, the joint distribution of $G_\Lambda$ and the description $\set{\ket{\wh{\psi}_k}}_{k\in\bigcup_{i=1}^\Lambda\bit^i}$ obtained from tomography in Step~2 of $\wt{\eve}^{G_\Lambda}$ is identically distributed as Steps~1 to 3 in $\KA_\plain$. Now, from the correctness guarantee of $\tomography$, there is a $1-\negl(\secp)$ fraction of $\set{\ket{\psi_k}}_{k\in\bigcup_{i=1}^\Lambda\bit^i}$ and $\set{\ket{\wh{\psi}_k}}_{k\in\bigcup_{i=1}^\Lambda\bit^i}$ such that event $\good$ occurs. By the same argument in the previous paragraph, the distributions of $(\tau,k_\alice,k_\bob)$ generated by $(\wt{\alice}^{G_\Lambda},\wt{\bob}^{G_\Lambda})$ and $(\wt{\alice}^{\wh{G}_\Lambda}, \wt{\bob}^{\wh{G}_\Lambda})$ are $1/\poly'(\secp)$-close in statistical distance. Since $\eve_\plain$ takes as input $\tau$ and $\set{\ket{\wh{\psi}_k}}_{k\in\bigcup_{i=1}^\Lambda\bit^i}$, $\wt{\eve}^{G_\Lambda}$ breaks the security of $\wt{\KA}^{G_\Lambda}$ with advantage at least $0.2 - 1/\poly'(\secp) > 0.1$ for infinitely many $\secp$, which  contradicts the security of $\wt{\KA}^{G_\Lambda}$.
\end{proof}

\begin{lemma} \label{lem:KA_eve}
There does not exist a secure QCCC key agreement relative to $G$.
\end{lemma}
\begin{proof}
It immediately follows from~\Cref{lem:Imp_KA_plain,lem:compile_KA}.
\end{proof}

\begin{theorem} \label{thm:QBB_KA}
There does not exist a quantum fully black-box reduction $(C,S)$ from QCCC key agreements to $(\secp,\omega(\log(\secp)))$-PRSGs such that $C$ only asks classical queries to the PRSG.
\end{theorem}
\begin{proof}
For the sake of contradiction, suppose $(C,S)$ is a fully black-box reduction satisfying the conditions. Let $\cI$ be the implementation of $(\secp,\omega(\log(\secp)))$-PRSGs as stated in the proof of~\Cref{lem:oracle_implies_PRS}. Then $C^\cI$ is a key agreement that satisfies completeness. From~\Cref{lem:KA_eve}, there exists a poly-query adversary $\wt{\eve}$ that breaks the security of the QCCC key agreement $C^\cI$. Then $S^{{\wt{\eve}},\cI}$ by definition breaks the security of the $(\secp,\omega(\log(\secp)))$-PRSG $\cI$ by asking polynomially many queries to $\wt{\eve}$ and $\cI$, thus in total polynomial queries to $G$. However, this contradicts~\Cref{lem:oracle_implies_PRS}.
\end{proof}

\subsection{Separating QCCC Interactive Commitments from $(\secp,\omega(\log(\secp)))$-PRSGs}
\begin{definition}[QCCC interactive commitments relative to oracle]
A QCCC commitment relative to an oracle $\cO$ is a two-party interactive protocol consisting of a pair of uniform QPT oracle algorithms $(C,R)$, where $C$ is the committer and $R$ is the receiver. Let $q = q(\secp)$ be an arbitrary polynomial. Each of $C$ and $R$ can ask $q$ queries to the oracle $\cO$ and are allowed to communicate classically.
\begin{itemize}
    \item \textbf{Commit phase:} In the (possibly interactive) commit phase, $C$ takes as input the security parameter $1^\secp$ and a bit $b\in\bit$, and $R$ takes as input the security parameter $1^\secp$. We denote the execution of the commit phase by $(\sigma_{CR},\tau) \gets \commit\langle C^\cO(1^\secp,b), R^\cO(1^\secp) \rangle$, where $\sigma_{CR}$ is the joint state of $C$ and $R$ after the commit phase, and $\tau$ denotes the transcript in the commit phase.
    \item \textbf{Reveal phase:} In the (possibly interactive) reveal phase, the output is $\mu\in\set{0,1,\bot}$ indicating the receiver's output bit or abort. We denote the execution of the reveal phase by $\mu \gets \reveal\langle C^\cO(1^\secp,b), R^\cO(1^\secp), \sigma_{CR}, \tau \rangle$.
\end{itemize}
The scheme satisfies the following conditions.
\begin{itemize}
    \item \textbf{$\veps$-completeness.} For all $\secp\in\N$,
    \[
    \Pr\left[
    \mu = b: 
    \substack{ O \gets \cO, \\
    b \gets \bit, \\
    (\sigma_{CR},\tau) \gets \commit\langle C^O(1^\secp,b), R^O(1^\secp) \rangle, \\
    \mu \gets \reveal\langle C^\cO(1^\secp,b), R^\cO(1^\secp), \sigma_{CR}, \tau \rangle, \\
    \mu\in\set{0,1,\bot}
    }
    \right] 
    \geq 1 - \veps(\secp).
    \]
    If $\veps$ is negligible, then we simply say that it is complete.
    \item \textbf{Statistical hiding.} For any polynomial $p$ and any computationally unbounded malicious receiver $R^*$ who asks at most $p(\secp)$ classical queries, there exists a negligible function $\negl$ such that for all $\secp\in\N$,
    \[
    \Pr\left[
    b' = b:
    \substack{ O \gets \cO, \\
    b \gets \bit, \\
    (\sigma_{CR^*},\tau) \gets \commit\langle C^O(1^\secp,b), {R^*}^O(1^\secp) \rangle, \\
    b' \gets {R^*}^O(\sigma_{R^*}, \tau)
    }
    \right] \leq \frac{1}{2} + \negl(\secp),
    \]
    where $\sigma_{R^*}$ denotes the state obtained by tracing out the committer's part of the state $\sigma_{CR^*}$.
    \item \textbf{Statistical binding.} For any polynomial $p$ and any computationally unbounded malicious committer $C^*$ who asks $p(\secp)$ classical queries, there exists a negligible function $\negl$ such that for all $\secp\in\N$,
    \[
    \Pr\left[
    \mu = \ch:
    \substack{ O \gets \cO, \\
    (\sigma_{C^*R},\tau) \gets \commit\langle {C^*}^O(1^\secp), R^O(1^\secp) \rangle, \\
    \ch \gets \bit, \\
    \mu \gets \reveal\langle {C^*}^O(\ch), R^O, \sigma_{C^*R}, \tau \rangle
    }
    \right] \leq \frac{1}{2} + \negl(\secp).
    \]
\end{itemize}
\end{definition}

\noindent We need the following lemma regarding total variation distance.
\begin{lemma} \label{lem:tv_distance}
Let $\bfP_{BT},\bfQ_{BT}$ be two discrete distributions over $\bit \times \cT$. Consider the following experiment:
\begin{mdframed}
\begin{multicols}{2}
\noindent $\mathbf{Exp.0:}$
\begin{enumerate}
    \item Sample $(b,\tau) \gets \bfP_{BT}$.
    \item If $\bfQ_{T}(\tau) = 0$,\footnote{$\bfQ_{T}$ denotes the marginal distribution of $\bfQ_{BT}$ on $T$.} then set $b'$ to a uniform bit. Otherwise, set $b'$ to the more likely bit according to $\bfQ_{B \mid T = \tau}$.
    \item Output $(b,b',\tau)$.
\end{enumerate}

\columnbreak

\noindent $\mathbf{Exp.1:}$
\begin{enumerate}
    \item Sample $(b,\tau) \gets \bfP_{BT}$.
    \item Set $b'$ to the more likely bit according to $\bfP_{B \mid T = \tau}$.
    \item Output $(b,b',\tau)$.
\end{enumerate}

\end{multicols}
\end{mdframed}
Then it holds that
\[
\Pr_{\mathbf{Exp.0}}[b = b'] 
\geq \Pr_{\mathbf{Exp.1}}[b = b'] - 3\dtv(\bfP_{BT},\bfQ_{BT}).
\]
\end{lemma}
\begin{proof}
Consider the following hybrid:
\begin{mdframed}
\noindent $\mathbf{Hyb:}$
\begin{enumerate}
    \item Sample $(b,\tau) \gets \bfQ_{BT}$.
    \item If $\bfQ_{T}(\tau) = 0$,\footnote{Since $(b,\tau)$ is sampled from $\bfQ_{BT}$, $\bfQ_{T}(\tau)$ is always $>0$. We write it merely for the clarity of the proof.} then set $b'$ to a uniform bit. Otherwise, set $b'$ to the more likely bit according to $\bfQ_{B \mid T = \tau}$.
    \item Output $(b,b',\tau)$.
\end{enumerate}
\end{mdframed}
Since a randomized function (Step~2 in $\mathbf{Exp.0}$ and $\mathbf{Hyb}$) cannot increase the total variation distance, we have 
\begin{align*} 
\left| \Pr_{\mathbf{Exp.0}}[b = b'] - \Pr_{\mathbf{Hyb}}[b = b'] \right|
\leq \dtv(\bfP_{BT},\bfQ_{BT}),
\end{align*}
which implies
\begin{align} \label{eq:dTV}
\Pr_{\mathbf{Exp.0}}[b = b']
\geq \Pr_{\mathbf{Hyb}}[b = b'] - \dtv(\bfP_{BT},\bfQ_{BT}).
\end{align}
In $\mathbf{Hyb}$, we have
\begin{align} \label{eq:dTV-1}
\Pr_{\mathbf{Hyb}}[b = b']
& = \Ex_{\tau \gets \bfQ_T}\left[ \frac{1}{2} + \dtv(\bfQ_{B \mid  T = \tau}, \bfU_1) \right] \nonumber \\
& = \frac{1}{2} + \sum_{\tau} \bfQ_T(\tau) \cdot \frac{1}{2} \sum_{b\in\bit} \left| \bfQ(b)_{B \mid  T = \tau} - \frac{1}{2} \right| \nonumber \\
& = \frac{1}{2} + \frac{1}{2} \sum_{\tau,b\in\bit} \left| \bfQ(b,\tau)_{BT} - \frac{1}{2} \cdot \bfQ_T(\tau) \right| \nonumber \\
& = \frac{1}{2} + \dtv(\bfQ_{BT}, \bfU_1 \otimes \bfQ_T),
\end{align}
where $\bfU_1$ denotes the uniform distribution on $\bit$.
Similarly, in $\mathbf{Exp.1}$, we have
\begin{align} \label{eq:dTV-2}
\Pr_{\mathbf{Exp.1}}[b = b']
& = \Ex_{\tau \gets \bfP_T}\left[ \frac{1}{2} + \dtv(\bfP_{B \mid  T = \tau}, \bfU_1) \right] \nonumber \\
& = \frac{1}{2} + \dtv(\bfP_{BT}, \bfU_1 \otimes \bfP_T) \nonumber \\
& \leq \frac{1}{2} + \dtv(\bfP_{BT}, \bfQ_{BT}) + \dtv(\bfQ_{BT}, \bfU_1 \otimes \bfQ_T) + \dtv(\bfU_1 \otimes \bfQ_T, \bfU_1 \otimes \bfP_T) \nonumber \\
& = \frac{1}{2} + \dtv(\bfP_{BT}, \bfQ_{BT}) + \dtv(\bfQ_{BT}, \bfU_1 \otimes \bfQ_T) + \dtv(\bfQ_T, \bfP_T) \nonumber \\
& \leq \frac{1}{2} + \dtv(\bfP_{BT}, \bfQ_{BT}) + \dtv(\bfQ_{BT}, \bfU_1 \otimes \bfQ_T) + \dtv(\bfQ_{BT}, \bfP_{BT}),
\end{align}
where the first inequality follows from the triangle inequality. From~\Cref{eq:dTV-1,eq:dTV-2}, we have
\begin{align} \label{eq:dTV-3}
\Pr_{\mathbf{Hyb}}[b = b']
\geq \Pr_{\mathbf{Exp.1}}[b = b'] - 2\dtv(\bfP_{BT},\bfQ_{BT}).
\end{align}
Hence, combining~\Cref{eq:dTV,eq:dTV-3}, we have
\[
\Pr_{\mathbf{Exp.0}}[b = b'] 
\geq \Pr_{\mathbf{Exp.1}}[b = b'] - 3\dtv(\bfP_{BT},\bfQ_{BT}). \qedhere
\]
\end{proof}

\begin{lemma}
There does not exist a QCCC interactive commitment relative to $G$.
\end{lemma}
\begin{proof}
For the sake of contradiction, suppose $\Com^G = (C^G,R^G)$ is a QCCC interactive commitment relative to $G$, where $q(\secp) = \poly(\secp)$ is the number of queries asked by $C$ and $R$ respectively and $L(\secp) = \poly(\secp)$ is the maximum input length of the queries. Define the function $\Lambda(\secp) := \ceil{\log(q^{10} + L^{10} + \secp^{10})} = O(\log(\secp))$ and the truncated oracle $G_\Lambda = \set{\set{G_k}_{k\in\bit^i}}_{i=1}^\Lambda$. The proof consists of two major parts. First, we will show that $\Com^G$ can be converted to a QCCC interactive commitment $\wt{\Com}^{G_\Lambda}$ relative to $G_\Lambda$. Next, we will show that any QCCC interactive commitment relative to $G_\Lambda$ cannot satisfy completeness, statistical hiding, and statistical binding simultaneously.

\paragraph{Converting $\Com^G$ to $\wt{\Com}^{G_\Lambda}$.}
We define the following scheme $\wt{\Com}^{G_\Lambda} = (\wt{C}^{G_\Lambda},\wt{R}^{G_\Lambda})$ relative to $G_\Lambda$:
\begin{mdframed}
$\wt{\Com}^{G_\Lambda}(1^\secp,\wt{C}^{G_\Lambda},\wt{R}^{G_\Lambda})$:
\begin{enumerate}
\item For every $k\in\bigcup_{i=\Lambda+1}^L\bit^i$, $\wt{C}$ and $\wt{R}$ samples $\ket{\phi^C_k},\ket{\phi^R_k}\gets\Haar_{|k|}$ respectively.
\item On input $1^\secp$ and $b$, $\wt{C}^{G_\Lambda}$ runs $C^{(\cdot)}(1^\secp,b)$ by answering the queries as follows. Suppose $C$ asks a query $k\in\bigcup_{i = 1}^L\bit^i$. If $|k|\leq \Lambda$, then $\wt{C}$ ask $k$ to oracle $G_\Lambda$ and forward the response. Otherwise, $\wt{C}$ sends $\ket{\phi^C_k}$ to $C$. On input $1^\secp$, $\wt{R}^{G_\Lambda}$ runs $R^{(\cdot)}(1^\secp)$ by answering $R$'s queries similarly, except that it replaces $\ket{\phi^C_k}$ with $\ket{\phi^R_k}$.
\end{enumerate}
\end{mdframed}
\paragraph{$\wt{\Com}^{G_\Lambda}$ is $1/\poly$-complete.}
This is similar to proving the completeness of $\wt{\KA}^{G_\Lambda}$ in the proof of~\Cref{lem:compile_KA}.

\paragraph{$\wt{\Com}^{G_\Lambda}$ is statistically hiding and statistically binding.}
Intuitively, $\wt{\Com}^{G_\Lambda}$ is at least as secure as $\Com^G$ because the malicious party cannot obtain any information about the Haar states of length greater than $\Lambda$ held by the other party via asking queries. To prove statistical hiding, suppose there exists a malicious receiver $(\wt{R}^*)^{G_\Lambda}$ that breaks the statistical hiding of $\wt{\Com}^{G_\Lambda}$ by asking polynomially many queries to $G_\Lambda$, then we construct a malicious receiver $(R^*)^G$ that breaks the statistical hiding of $\Com^G$ by using $(\wt{R}^*)^{(\cdot)}$. $(R^*)^G$ simply runs $\wt{R}^*$ by answering its queries with $G$. Since the distributions of the (honest) committer $C$ in $\Com^G$ and $\wt{C}$ in $\wt{\Com}^{G_\Lambda}$ are identical, the advantage of $(R^*)^G$ is equal to that of $(\wt{R}^*)^{G_\Lambda}$. This contradicts the premise that $\Com^G$ is statistically hiding.

Similarly, to prove statistical binding, suppose there exists a malicious committer $(\wt{C}^*)^{G_\Lambda}$ that breaks the statistical binding of $\wt{\Com}^{G_\Lambda}$ by asking polynomially many queries to $G_\Lambda$, then we construct a malicious committer $(C^*)^G$ that breaks the statistical binding of $\Com^G$ by using $(\wt{C}^*)^{(\cdot)}$. $(C^*)^G$ simply runs $\wt{C}^*$ by answering its queries with $G$. Since the distributions of the (honest) receiver $R$ in $\Com^G$ and $\wt{R}$ in $\wt{\Com}^{G_\Lambda}$ are identical, the advantage of $(C^*)^G$ is equal to that of $(\wt{C}^*)^{G_\Lambda}$. This contradicts the premise that $\Com^G$ is statistically binding. \\

In the rest of the proof, we will show that a commitment scheme relative to $G_\Lambda$ cannot satisfy completeness, statistical hiding, and statistical binding at the same time. Intuitively, this is because the output length of $G_\Lambda$ is short, so each party can approximate the whole oracle by performing tomography using polynomially many queries. Hence, the scheme can be reduced to the plain model, modulo the error introduced by tomography.

\paragraph{QCCC commitments do not exist relative to $G_\Lambda$.} We will show that
there does not exist a complete, statistically hiding, and statistically binding QCCC interactive commitment relative to $G_\Lambda$. Toward contradiction, suppose $\wt{\Com}^{G_\Lambda} = (\wt{C}^{G_\Lambda}, \wt{R}^{G_\Lambda})$ is such a scheme. Consider the following malicious receiver $R^*$ (for brevity, we omit the tilde $\wt{\cdot}$ in the rest of the proof) with classical oracle access to $G_\Lambda$:
\begin{mdframed}
\textbf{$R^*$ in Hiding Experiment}:
\begin{enumerate}
\item $R^*$ runs the commit phase honestly with $C$ who commits to $\bfb$ (where $b$ was sampled uniformly at random by $C$) and obtains the transcript $\tau$. 
\item $R^*$ performs $\tomography$ (defined in~\Cref{thm:tomography}) with parameters $\Delta = 2^{-2\Lambda}$ and $\gamma = 2^{-\secp}$ on every output state of $G_\Lambda$ to obtain the description, denoted by $\wh{G}_\Lambda$.
\item If $\tau$ and $\wh{G}_\Lambda$ are not consistent, then $R^*$ output a uniform bit $\bfb'$. Otherwise, $R^*$ outputs the more likely bit $\bfb'$ from the distribution conditioned on $(\tau,\wh{G}_\Lambda)$.
\end{enumerate}
\end{mdframed}
For efficiency, $R^*$ asks polynomially many queries in Step~2. For every fixed $(G_\Lambda,\wh{G}_\Lambda)$, we denote by $p^{R^*}_{G_\Lambda,\wh{G}_\Lambda}$ the probability that $R^*$ guess the committed bit correctly.
\paragraph{Analyze $R^*$.}
The structure of the proof is similar to proving the completeness of $\KA_\plain$ in~\Cref{lem:compile_KA}. Define the event $\good$ in the hiding experiment as
\[
\good \equiv \bigwedge_{k \in \bigcup^\Lambda_{i=1} \bit^i} \left[ \TD(\ketbra{\psi_k}{\psi_k}, \ketbra{\wh{\psi}_k}{\wh{\psi}_k}) \leq \Delta \right].
\]
We now consider any pair $G_\Lambda = \set{\set{\ket{\psi_k}}_{k\in\bit^i}}_{i=1}^\Lambda$ and $\wh{G}_\Lambda = \set{\set{\ket{\wh{\psi}_k}}_{k\in\bit^i}}_{i=1}^\Lambda$ such that event $\good$ occurs. Let $\bfD_{BT\mid G_\Lambda}$ (\resp $\bfD_{BT\mid \wh{G}_\Lambda}$) denote the distribution of $(b,\tau)$ in the honest commit phase of $\wt{\Com}$ conditioned on oracle being $G_\Lambda$ (\resp $\wh{G}_\Lambda$). Since $\wt{C}^{(\cdot)}$ and $\wt{R}^{(\cdot)}$ in $\wt{\Com}^{G_\Lambda}$ (\resp $\wt{\Com}^{\wh{G}_\Lambda}$) ask a total of $2q$ queries, one can use $\set{\set{\ket{\psi_k}^{\otimes 2q}}_{k\in\bit^i}}_{i=1}^\Lambda$ (\resp $\set{\set{\ket{\wh{\psi}_k}^{\otimes 2q}}_{k\in\bit^i}}_{i=1}^\Lambda$) to perfectly answer $\wt{C}$'s and $\wt{R}$'s queries. From the operational definition of trace distance, we have
\[
\dtv(\bfD_{BT\mid G_\Lambda},\bfD_{BT\mid \wh{G}_\Lambda}) 
\leq 2q \cdot \sum_{i=1}^\Lambda 2^i \cdot \Delta
= \frac{1}{\poly(\secp)}
\]
for some polynomial $\poly$. \\

\noindent Define the quantity $p^{R^*}_{G_\Lambda}$ which is equal to the success probability of $R^*$ conditioned on $G_\Lambda$ without tomography error, \ie
\begin{align} \label{eq:hiding}
p^{R^*}_{G_\Lambda}
& := \frac{1}{2} + \Ex_{\tau\gets \bfD_{\tau \mid G_\Lambda}}\left[ \dtv(\bfD_{B \mid G_\Lambda, T = \tau}, \bfU_1) \right] \nonumber \\
& = \frac{1}{2} + \sum_\tau \Pr_\commit[\tau \mid G_\Lambda] \cdot  \dtv(\bfD_{B \mid G_\Lambda, T = \tau}, \bfU_1),
\end{align}
where $\commit$ denotes the honest commit phase of $\wt{\Com}$. By~\Cref{lem:tv_distance} (setting $\bfP \equiv \bfD_{BT\mid G_\Lambda}$ and $\bfQ \equiv \bfD_{BT\mid \wh{G}_\Lambda}$), we have
\begin{align} \label{eq:fixed_hiding_advantage}
p^{R^*}_{G_\Lambda,\wh{G}_\Lambda} 
\geq p^{R^*}_{G_\Lambda} - 3\dtv(\bfD_{BT\mid G_\Lambda},\bfD_{BT\mid \wh{G}_\Lambda}) 
= p^{R^*}_{G_\Lambda} - \frac{3}{\poly(\secp)}.
\end{align}
Finally, after averaging over $(G_\Lambda,\wh{G}_\Lambda)$, the probability $p_{R^*\text{win}}$ that $R^*$ guess the committed bit correctly satisfies
\begin{align} \label{eq:average_hiding_advantage}
p_{R^*\text{win}}
& := \Ex_{G_\Lambda,\wh{G}_\Lambda}[ p^{R^*}_{G_\Lambda,\wh{G}_\Lambda} ] \nonumber \\
& = \Ex_{G_\Lambda} \left[ \Ex_{\wh{G}_\Lambda} \left[ p^{R^*}_{G_\Lambda,\wh{G}_\Lambda} \mid G_\Lambda \right] \right] \nonumber \\
& \geq \Ex_{G_\Lambda} \left[ \Pr_{\wh{G}_\Lambda}[\good \mid G_\Lambda] \cdot \Ex_{\wh{G}_\Lambda}\left[ p^{R^*}_{G_\Lambda,\wh{G}_\Lambda} \mid G_\Lambda \land \good \right] \right] \nonumber \\
& \geq \Ex_{G_\Lambda} \left[ (1 - \negl(\secp)) \cdot \left( p^{R^*}_{G_\Lambda} - \frac{3}{\poly(\secp)} \right) \right] \nonumber \\
& = (1 - \negl(\secp)) \cdot \left( \Ex_{G_\Lambda} \left[ p^{R^*}_{G_\Lambda} \right] - \frac{3}{\poly(\secp)} \right).
\end{align}
The second inequality follows from~\Cref{eq:fixed_hiding_advantage} and the following reason: by the correctness guarantee of $\tomography$ (\Cref{thm:tomography}) and a union bound, the probability of $\good$ happening conditioned on any $G_\Lambda$ is at least $1 - \sum_{i=1}^\Lambda 2^i \cdot \gamma = 1 - \negl(\secp)$. \\

\noindent Next, consider the following malicious committer $C^*$:
\begin{mdframed}
\textbf{$C^*$ in Binding Experiment}:
\begin{enumerate}
    \item $C^*$ commits to a uniform bit $\bfb$, runs the commit phase with $R$ honestly, and generates the transcript $\tau$. The joint state of $C^*$ and $R$ after the commit phase is $\rho_{\bfb,G_\Lambda,\tau}^{\sfC} \otimes \sigma_{G_\Lambda,\tau}^{\sfR}$.\footnote{From~\Cref{lem:Cond_Indep}, the joint state is a product state. Moreover, fixing $(G_\Lambda, \tau)$ already determines the state of $R$. So it is independent of $\bfb$ after conditioned on $(G_\Lambda, \tau)$.}
    \item $C^*$ performs $\tomography$ (defined in~\Cref{thm:tomography}) with parameters $\Delta = 2^{-2\Lambda}$ and $\gamma = 2^{-\secp}$ on every output state of $G_\Lambda$ to obtain the description, denoted by $\wh{G}_\Lambda$.
    \item Upon receiving the challenge bit $\ch$, $C^*$ computes the description of the joint state conditioned on $(\ch,\tau,\wh{G}_\Lambda)$, denoted by $\rho_{\ch,\wh{G}_\Lambda,\tau}^{\sfC\sfR}$. If $(\ch,\wh{G}_\Lambda,\tau)$ is inconsistent, then $C^*$ aborts.\footnote{Note that it is equivalently to setting $\rho_{\ch,\wh{G}_\Lambda,\tau}^{\sfC\sfR}$ to the zero matrix in terms of calculating $C^*$'s success probability.}
    \item $C^*$ runs the reveal phase honestly on input $\ch$ and state $\rho_{\ch,\wh{G}_\Lambda,\tau}^\sfC$.
\end{enumerate}
\end{mdframed}
For efficiency, $C^*$ asks polynomially many queries in Step~2. For every fixed $(G_\Lambda,\wh{G}_\Lambda)$, the probability that $C^*$ successfully opens to $\ch$ is
\begin{align*}
& p^{C^*}_{G_\Lambda,\wt{G}_\Lambda}
:= \sum_\tau \sum_{b, ch \in \bit} \Pr_\commit[\tau \mid G_\Lambda] \cdot \Pr_\commit[\bfb = b \mid \tau, G_\Lambda] \cdot \Pr_\commit[\ch = ch \mid \bfb = b,\tau,G_\Lambda] \\
& \cdot \Pr[\reveal\langle {C}^{G_\Lambda}(ch), R^{G_\Lambda}, \rho_{ch,\wh{G}_\Lambda,\tau}^\sfC \otimes \sigma_{G_\Lambda, \tau}^\sfR, \tau \rangle = ch] \\
& = \sum_\tau \sum_{b, ch \in \bit} \Pr_\commit[\tau \mid G_\Lambda] \cdot \Pr_\commit[\bfb = b \mid \tau, G_\Lambda] \cdot \frac{1}{2} \cdot \Pr[\reveal\langle {C}^{G_\Lambda}(ch), R^{G_\Lambda}, \rho_{ch,\wh{G}_\Lambda,\tau}^\sfC \otimes \sigma_{G_\Lambda, \tau}^\sfR, \tau \rangle = ch]
\end{align*}
since $\ch$ is sampled uniformly and independently.

\paragraph{Analyze $C^*$.}
Define the event $\good$ in the same way as in the hiding experiment. For every fixed $(G_\Lambda,\wh{G}_\Lambda)$ such that $\good$ happens, 
consider the following two classical-quantum states corresponding to the joint state of $C$ and $R$ right after the honest commit phase of $\wt{\Com}$ conditioned on the oracle being $G_\Lambda$ and $\wh{G}_\Lambda$ respectively:
\[
\Psi_{G_\Lambda} := \sum_{b\in\bit}\sum_{\tau} \frac{1}{2} 
\cdot \Pr_\commit[\tau \mid \bfb = b, G_\Lambda] \cdot \ketbra{b}{b}_\sfB 
\otimes \rho_{b,G_\Lambda,\tau}^\sfC
\otimes \sigma_{G_\Lambda,\tau}^\sfR
\otimes \ketbra{\tau}{\tau}_\sfT,
\]
\[
\Psi_{\wh{G}_\Lambda} := \sum_{b\in\bit}\sum_{\tau} \frac{1}{2} 
\cdot \Pr_\commit[\tau \mid \bfb = b, \wh{G}_\Lambda] \cdot \ketbra{b}{b}_\sfB 
\otimes \rho_{b,\wh{G}_\Lambda,\tau}^\sfC
\otimes \sigma_{\wh{G}_\Lambda,\tau}^\sfR
\otimes \ketbra{\tau}{\tau}_\sfT,
\]
where register $\sfB$ is the committer's private register for storing the input and register $\sfT$ is the public register for storing the transcript. Similar to the previous section, since $\wt{C}^{(\cdot)}$ and $\wt{R}^{(\cdot)}$ in $\wt{\Com}^{G_\Lambda}$ (\resp $\wt{\Com}^{\wh{G}_\Lambda}$) ask a total of $2q$ queries, one can use $\set{\set{\ket{\psi_k}^{\otimes 2q}}_{k\in\bit^i}}_{i=1}^\Lambda$ (\resp $\set{\set{\ket{\wh{\psi}_k}^{\otimes 2q}}_{k\in\bit^i}}_{i=1}^\Lambda$) to perfectly answer $\wt{C}$'s and $\wt{R}$'s queries. From the correctness guarantee of $\tomography$, we have
\begin{align} \label{eq:binding_1}
\TD(\Psi_{G_\Lambda},\Psi_{\wh{G}_\Lambda})
\leq 2q \cdot \sum_{i=1}^\Lambda 2^i \cdot \Delta
= \frac{1}{\poly(\secp)}.
\end{align}
In order to analyze the success probability of $C^*$ conditioned on $(G_\Lambda,\wh{G}_\Lambda)$, we define the following state
\[
\Psi_{G_\Lambda,\wh{G}_\Lambda} := \sum_{b\in\bit}\sum_{\tau} \frac{1}{2} 
\cdot \Pr_\commit[\tau \mid \bfb = b, G_\Lambda] \cdot \ketbra{b}{b}_\sfB 
\otimes \rho_{b,\wh{G}_\Lambda,\tau}^\sfC
\otimes \sigma_{G_\Lambda,\tau}^\sfR
\otimes \ketbra{\tau}{\tau}_\sfT.
\]
We claim that 
\begin{align} \label{eq:binding_hyb}
\TD(\Psi_{G_\Lambda},\Psi_{G_\Lambda,\wt{G}_\Lambda})
\leq \frac{2}{\poly(\secp)}.
\end{align}
To prove~\Cref{eq:binding_hyb}, we introduce the following hybrid state:
\[
\Psi_{\Hyb} := \sum_{b\in\bit}\sum_{\tau} \frac{1}{2} 
\cdot \Pr_\commit[\tau \mid \bfb = b, \wh{G}_\Lambda] \cdot \ketbra{b}{b}_\sfB 
\otimes \rho_{b,\wh{G}_\Lambda,\tau}^\sfC
\otimes \sigma_{G_\Lambda,\tau}^\sfR
\otimes \ketbra{\tau}{\tau}_\sfT.
\]
By the triangle inequality, we can bound~\Cref{eq:binding_hyb} as
\begin{align} \label{eq:binding_2}
\TD(\Psi_{G_\Lambda},\Psi_{G_\Lambda,\wh{G}_\Lambda})
\leq \TD(\Psi_{G_\Lambda},\Psi_{\Hyb}) + \TD(\Psi_{\Hyb},\Psi_{G_\Lambda,\wh{G}_\Lambda}).
\end{align}
For the first term in~\Cref{eq:binding_2}, we have
\begin{align*}
\TD(\Psi_{G_\Lambda},\Psi_{\Hyb}) 
& = \sum_{b\in\bit} \sum_{\tau} \frac{1}{2} 
\cdot  
\TD\left( \Pr_\commit[\tau \mid \bfb = b, G_\Lambda] \cdot \rho_{b,G_\Lambda,\tau}^\sfC, 
\Pr_\commit[\tau \mid \bfb = b, \wh{G}_\Lambda] \cdot \rho_{b,\wh{G}_\Lambda,\tau}^\sfC \right) \\
& = \TD(\Tr_{\sfR}(\Psi_{G_\Lambda}), \Tr_{\sfR}(\Psi_{\wh{G}_\Lambda})) \\
& \leq \TD(\Psi_{G_\Lambda},\Psi_{\wh{G}_\Lambda}) \\
& = \frac{1}{\poly(\secp)},
\end{align*}
where the first two equalities are because $\TD(\bigoplus_i A_i, \bigoplus_i B_i) = \sum_i \TD(A_i, B_i)$ and the inequality is because the trace distance won't increase under partial trace; the inequality follows from~\Cref{eq:binding_1}. 
Similarly, For the first term in~\Cref{eq:binding_2}, we have
\begin{align*}
\TD(\Psi_{\Hyb},\Psi_{G_\Lambda,\wh{G}_\Lambda}) 
& = \sum_{b\in\bit} \sum_{\tau} \frac{1}{2} 
\cdot \TD\left( \Pr_\commit[\tau \mid \bfb = b, G_\Lambda], 
\Pr_\commit[\tau \mid \bfb = b, \wh{G}_\Lambda] \right) \\
& = \TD(\Tr_{\sfC\sfR}(\Psi_{G_\Lambda}), \Tr_{\sfC\sfR}(\Psi_{\wh{G}_\Lambda})) \\
& \leq \TD(\Psi_{G_\Lambda},\Psi_{\wh{G}_\Lambda}) \\
& = \frac{1}{\poly(\secp)}.
\end{align*}
Thus, the proof of~\Cref{eq:binding_hyb} is complete. \\

\noindent Define the quantity $p^{C^*}_{G_\Lambda}$ which is equal to the success probability of $C^*$ conditioned on $G_\Lambda$ without tomography error:
\begin{align*}
p^{C^*}_{G_\Lambda} 
:= \sum_\tau \sum_{ch \in \bit} \Pr_\commit[\tau \mid G_\Lambda] \cdot \frac{1}{2} \cdot \Pr[\reveal\langle {C}^{G_\Lambda}(ch), R^{G_\Lambda}, \rho_{ch,G_\Lambda,\tau}^\sfC \otimes \sigma_{G_\Lambda, \tau}^\sfR, \tau \rangle = ch].
\end{align*}
Thus, from the operational definition of trace distance and~\Cref{eq:binding_hyb}, we have
\begin{align*}
|p^{C^*}_{G_\Lambda,\wh{G}_\Lambda} - p^{C^*}_{G_\Lambda}|
\leq \TD(\Psi_{G_\Lambda},\Psi_{G_\Lambda,\wh{G}_\Lambda})
\leq \frac{2}{\poly(\secp)},
\end{align*}
which implies
\begin{align} \label{eq:fixed_binding_advantage}
p^{C^*}_{G_\Lambda,\wh{G}_\Lambda}
\geq p^{C^*}_{G_\Lambda} - \frac{2}{\poly(\secp)}.
\end{align}
By a similar argument to that of~\Cref{eq:average_hiding_advantage}, the probability $p_{C^*\text{win}}$ that $C^*$ successfully opens to $\ch$ satisfies
\begin{align} \label{eq:average_binding_advantage}
p_{C^*\text{win}}
:= \Ex_{G_\Lambda,\wh{G}_\Lambda}[ p^{C^*}_{G_\Lambda,\wh{G}_\Lambda} ]
\geq (1 - \negl(\secp)) \cdot \left( \Ex_{G_\Lambda} \left[ p^{C^*}_{G_\Lambda} \right] - \frac{2}{\poly(\secp)} \right).
\end{align}

\paragraph{Trade-off between completeness, hiding, and binding of commitments.}
Suppose $\wt{\Com}^{G_\Lambda}$ satisfies $\veps$-completeness. In other words,
\begin{align} \label{eq:completeness}
& p_{\complete} := \nonumber \\
& \Ex_{G_\Lambda} \left[ \sum_\tau \sum_{b \in \bit} 
\Pr_\commit[\tau \mid G_\Lambda] 
\cdot \Pr_\commit[\bfb = b \mid \tau, G_\Lambda] 
\cdot \Pr[\reveal\langle {C}^{G_\Lambda}(b), R^{G_\Lambda}, \rho_{b,G_\Lambda,\tau}^\sfC \otimes \sigma_{G_\Lambda, \tau}^\sfR, \tau \rangle = b] \right] \nonumber \\
& \geq 1 - \veps.
\end{align}
Now, for any fixed $(G_\Lambda,\tau)$ in the support of the honest commit phase of $\wt{\Com}^{G_\Lambda}$, define the success probabilities of $R^*$ and $C^*$ conditioned on $(G_\Lambda,\tau)$:
\[
p^{R^*}_{G_\Lambda, \tau} 
:= \frac{1}{2} + \dtv(\bfD_{B \mid G_\Lambda, \tau}, \bfU_1)
= \frac{1}{2} + \frac{1}{2}\sum_{b \in \bit}\left| \Pr_\commit[\bfb = b \mid \tau, G_\Lambda] - \frac{1}{2} \right|,
\]
\[
p^{C^*}_{G_\Lambda, \tau} := 
\sum_{ch \in \bit} \frac{1}{2} \cdot \Pr[\reveal\langle {C}^{G_\Lambda}(ch), R^{G_\Lambda}, \rho_{ch,G_\Lambda,\tau}^\sfC \otimes \sigma_{G_\Lambda, \tau}^\sfR, \tau \rangle = ch].
\]
W.L.O.G, suppose $\Pr_\commit[\bfb = 0 \mid \tau, G_\Lambda] = \frac{1}{2} + \eta$ and $\Pr_\commit[\bfb = 1 \mid \tau, G_\Lambda] = \frac{1}{2} - \eta$ for some $\eta\in[0,0.5]$ (the opposite case can be proven symmetrically). Thus, it holds that
\begin{align*}
p^{R^*}_{G_\Lambda, \tau} = \frac{1}{2} + \eta.
\end{align*}
A straightforward calculation yields
\begin{align} \label{eq:fixed_trade_off}
& p^{R^*}_{G_\Lambda, \tau} + p^{C^*}_{G_\Lambda, \tau} \nonumber \\
& = \frac{1}{2} + \eta + \sum_{ch \in \bit} \frac{1}{2} \cdot \Pr[\reveal\langle {C}^{G_\Lambda}(ch), R^{G_\Lambda}, \rho_{ch,G_\Lambda,\tau}^\sfC \otimes \sigma_{G_\Lambda, \tau}^\sfR, \tau \rangle = ch] \nonumber \\
& \geq \frac{1}{2} + \eta \cdot \bigg( \Pr[\reveal\langle {C}^{G_\Lambda}(0), R^{G_\Lambda}, \rho_{0,G_\Lambda,\tau}^\sfC \otimes \sigma_{G_\Lambda, \tau}^\sfR, \tau \rangle = 0] - \Pr[\reveal\langle {C}^{G_\Lambda}(1), R^{G_\Lambda}, \rho_{1,G_\Lambda,\tau}^\sfC \otimes \sigma_{G_\Lambda, \tau}^\sfR, \tau \rangle = 1] \bigg) \nonumber \\
& + \sum_{ch \in \bit} \frac{1}{2} \cdot \Pr[\reveal\langle {C}^{G_\Lambda}(ch), R^{G_\Lambda}, \rho_{ch,G_\Lambda,\tau}^C \otimes \rho_{G_\Lambda, \tau}^R, \tau \rangle = ch] \nonumber \\
& = \frac{1}{2} + \left( \frac{1}{2} + \eta \right) \Pr[\reveal\langle {C}^{G_\Lambda}(0), R^{G_\Lambda}, \rho_{0,G_\Lambda,\tau}^\sfC \otimes \sigma_{G_\Lambda, \tau}^\sfR, \tau \rangle = 0] \nonumber \\
& + \left( \frac{1}{2} - \eta \right) \Pr[\reveal\langle {C}^{G_\Lambda}(1), R^{G_\Lambda}, \rho_{1,G_\Lambda,\tau}^\sfC \otimes \sigma_{G_\Lambda, \tau}^\sfR, \tau \rangle = 1] \nonumber \\
& = \frac{1}{2} + \sum_{b \in \bit} \Pr_\commit[\bfb = b \mid \tau, G_\Lambda] 
\cdot \Pr[\reveal\langle {C}^{G_\Lambda}(b), R^{G_\Lambda}, \rho_{b,G_\Lambda,\tau}^\sfC \otimes \rho_{G_\Lambda, \tau}^\sfR, \tau \rangle = b].
\end{align}
By averaging over $(G_\Lambda,\tau)$ in~\Cref{eq:fixed_trade_off} and recalling the definition of $p_{\complete}$ in~\Cref{eq:completeness}, we have
\begin{align} \label{eq:average_trade_off}
\Ex_{G_\Lambda}[ p^{R^*}_{G_\Lambda} ] + \Ex_{G_\Lambda}[ p^{C^*}_{G_\Lambda} ]
= \Ex_{G_\Lambda, \tau}[ p^{R^*}_{G_\Lambda, \tau} + p^{C^*}_{G_\Lambda, \tau} ]
\geq \frac{1}{2} + p_{\complete} 
\geq \frac{3}{2} - \veps.
\end{align}
Finally, combining~\Cref{eq:average_hiding_advantage,eq:average_binding_advantage,eq:average_trade_off}, $p_{R^*\text{win}}$, $p_{C^*\text{win}}$, and $\veps$ satisfy
\begin{align*}
\frac{p_{R^*\text{win}} + p_{C^*\text{win}}}{1 - \negl(\secp)} + \frac{5}{\poly(\secp)}
\geq \frac{3}{2} - \veps.
\end{align*}
After rearranging, we have
\[
\left( p_{R^*\text{win}} - \frac{1}{2} \right) + \left( p_{C^*\text{win}} - \frac{1}{2} \right) + (1 - \negl(\secp)) \cdot \veps 
\geq \frac{1}{2} - \frac{3}{2} \negl(\secp) - \frac{5(1 - \negl(\secp))}{\poly(\secp)}.
\]
Therefore, at least one of $\set{p_{R^*\text{win}} - 1/2, p_{C^*\text{win}} - 1/2, \veps}$ is non-negligible. That is, $\wt{\Com}^{G_\Lambda}$ cannot satisfy completeness, statistical hiding, and statistical binding simultaneously.
\end{proof}

\begin{theorem} \label{thm:QBB_Com}
There does not exist a quantum fully black-box reduction $(C,S)$ from QCCC interactive commitments to $(\secp,\omega(\log(\secp)))$-PRSGs such that $C$ only asks classical queries to the PRSG.
\end{theorem}
\begin{proof}
It is essentially the same as the proof of~\Cref{thm:QBB_KA}.
\end{proof}

\begin{remark} \label{remark:QBB}
\par We compare our results with existing results. Note that our impossibility results only rule out implementations that ask classical queries to the PRSG. There exist applications that need to query a PRSG/PRFSG in superposition, \eg quantum bit commitments~\cite{MY21}, quantum PKEs~\cite{BGH+23}, etc. However, all of them require quantum communication. It is less obvious how this would be helpful in the QCCC setting. We leave the generalization of the impossibility results as an open problem.

\par Next, since PRS generators can be constructed from one-way functions in a black-box way~\cite{JLS18}, one might wonder whether~\Cref{thm:QBB_KA} is already implied by the classical separation result between key agreements and one-way functions~\cite{IR89,BR09}. In other words, can we prove~\Cref{thm:QBB_KA} by using a (classical) random oracle? We pointed out that all currently known constructions of PRS generators from one-way functions~\cite{JLS18,BS19,BrakerskiS20,GB23,JMW23} require \emph{quantum} oracle access. The impossibility of QCCC key agreements in the quantum random oracle model was studied in~\cite{ACC+22}, where they ruled out \emph{perfectly-complete} key agreements based on a conjecture. However, \Cref{thm:QBB_KA} separates imperfectly-complete key agreements from $\omega(\log(\secp))$-PRSGs without relying on any conjecture. Hence, the two results are incomparable.
\end{remark}

\subsection{Extending the Separation Results} \label{sec:Extend_QBB_PRFS}
We observe that our technique can also separate QCCC key agreements and commitments from \emph{classically accessible $(\secp,m,n)$-PRFSGs with $n = \omega(\log(\secp))$} and $m$ being arbitrary. Recall that currently there is no construction of long-input PRFSGs (\ie $m = \omega(\log(\secp))$) from PRSGs. Hence, the separation might be strictly stronger. To prove it, we strengthen the separating oracle by increasing the number of oracles as $G = \set{\set{G_{k,x}}_{k,x\in\bit^\secp}}_{\secp\in\N}$. In this way, $G$ can support answering the classical query on key $k$ and input $x$. The rest of the proof is identical to the case of $(\secp,\omega(\log(\secp)))$-PRSGs.

\section*{Acknowledgements}
This work is supported by the National Science Foundation under Grant No. 2329938 and Grant No. 2341004.

\printbibliography

\appendix

\section{Related Work}

\subsection{Quantum Pseudorandomness: State of the Art}
\label{sec:qpseudorandomness}
We present the state of the art of the pseudorandomness notions in the quantum world. We will only restrict our attention to two notions relevant to this work. The open problems will be {\em italicized}.  

\paragraph{Pseudorandomnes State Generators (PRSGs).} The concept of pseudorandom state generators (PRSGs) was introduced in a seminal work by Ji, Liu and Song~\cite{JLS18}. Roughly speaking, it states that any computationally bounded adversary cannot distinguish whether it receives many copies of a state produced using a pseudorandom state generator on a uniform key versus many copies of a single Haar state. We summarise the state of the art of PRSGs below. We use the notation  $(\secparam,n)$-PRSG to denote a PRSG with $\secparam$ being the key length and $n$ being the output length. The number of copies of the state given to the adversary is denoted to be $t$. Unless otherwise stated, $t$ will be an arbitrary polynomial in $\secparam$ that is not fixed ahead of time. If $t$ is indeed fixed ahead of time then we denote such a notion by $(\secparam,n,t)$-PRSGs. 

\begin{itemize}
    \item $n > \secparam$ (stretch): It is known that $(\secparam,n)$-PRSGs exist assuming one-way functions~\cite{JLS18,BS19,BrakerskiS20} or even pseudorandom unitaries\footnote{An efficiently computable keyed circuit is a pseudorandom unitary if any adversary cannot distinguish whether it has oracle access to the keyed circuit or a Haar unitary. }~\cite{JLS18,metger2024pseudorandom,CBBFDHX24}. Even to design $(\secparam,n,1)$-PRSG, we need computational assumptions and in fact, $(\secparam,n,1)$-PRSG is implied by multi-copy PRSGs with output length $\Omega(\log(\secparam))$~\cite{GJMZ23}. {\em However, it is not known if stretch $(n,\secparam)$-PRSGs exist under weaker assumptions}, although we do have some candidates inspired from random circuits~\cite{AQY21}. There is some evidence to believe that stretch PRSGs might be weaker than any existing classical cryptographic assumption~\cite{Kretschmer21,lombardi2023one}.     
    \item $n \leq \secparam$: This can be broken down into three parameter regimes:
    \begin{itemize}
        \item $n < c \cdot \log(\secparam)$, for some $c \in \mathbb{R}$: $(n,\secparam)$-PRSGs exists unconditionally~\cite{BrakerskiS20}.
        \item $n \in \Omega(\log(\secparam))$: for $n\geq \log(\secparam)$, it was shown~\cite{AGQY22} that $(\secparam,n)$-PRSGs cannot be unconditionally secure. However, assuming one-way functions, $(n,\secparam)$-PRSGs was shown to exist~\cite{JLS18,BS19,BrakerskiS20} or even pseudorandom unitaries~\cite{JLS18,metger2024pseudorandom}. {\em Designing $(\secparam,n)$-PRSGs from weaker assumptions is an interesting direction.} There seems to be a separation between $n=\Theta(\log(\secparam))$ and $n=\Omega(\log(\secparam))$ as shown in~\cite{ananth2023pseudorandom,bouaziz2024quantum,coladangelo2024black}. On the other hand, when $t$ is known ahead of time, $(\secparam,n,t)$-PRSGs with statistical security, where $\secparam$ could be much larger than $n$, are implied by state designs. 
    \end{itemize}
\end{itemize}

\paragraph{Pseudorandom Function-Like State Generators (PRFSGs).} The notion of pseudorandom function-like state generators (PRFSGs) was introduced in the work of~\cite{AQY21} as a quantum  analogue of pseudorandom functions. Unlike pseudorandom state generators, in the case of PRFSG, we can use the same key to generate many pseudorandom states, indexed by classical strings. We summarise the state of the art of PRFSGs below. We use the notation $(\secparam,m,n)$-PRFSG to denote a PRFSG with $\secparam$ being the key length,  $m$ being the input length and $n$ being the output length. The number of copies of the state given to the adversary is denoted to be $t$. Unless otherwise stated, $t$ will be an arbitrary polynomial in $\secparam$ and not fixed ahead of time. If $t$ is indeed fixed ahead of time then we denote such a notion by $(\secparam,m,n,t)$-PRFSGs. 
\begin{itemize} 
\item $m = O(\log(\secparam))$: It is known that $(\secparam,m,n)$-PRFSGs, for some $n$, exist based on PRSGs. 
\item $m = \omega(\log(\secparam))$: While we know how to construct $(\secparam,m,n)$-PRFSGs from one-way functions~\cite{AGQY22}, {\em it is not yet known that stretch $(\secparam,m,n)$-PRFSGs exist assuming PRSGs}.   
\end{itemize} 
In the case when $t$ is known ahead of time, unitary designs can be used to achieve statistically secure PRFSGs. 

\subsection{Comparison with~\cite{chen2024power} and~\cite{AGL24}}
\label{sec:cgg24}
\noindent The common Haar state model was concurrently introduced by~\cite{chen2024power} and an earlier version of this work~\cite{AGL24}. Even though the main theme -- studying feasibility and separations in the CHS model -- was common among both the works, there were two main differences. Firstly,~\cite{chen2024power} showed the feasibility of 1-copy PRSGs whereas~\cite{AGL24} showed the feasibility of bounded-copy PRSGs with simplified construction and its analysis. Secondly,~\cite{chen2024power} showed a separation between 1-copy PRS and unbounded-copy PRS which is unique to their work. 
\par Subsequent to both~\cite{chen2024power} and~\cite{AGL24}, we improved upon~\cite{AGL24} to show that even bounded-query PR\underline{F}SGs exist in the CHS model. We also demonstrate optimality, in terms of the query bound, of our construction. We also added separation results in the revised version (\Cref{sec:LOCC},~\Cref{sec:Imp_QCCC_CHS} and~\Cref{sec:QBB_QCCC}). 

\section{Alternative Proof of~\Cref{lem:Z_haar_indis}} \label{app:simpleproof}

\begin{proof}[Proof sketch of~\Cref{lem:Z_haar_indis}]
The first part of the proof is the same as in~\cite{Col23}. Here we introduce the required notations and omit the details. Let $d := 2^n$ and
\begin{align*}
\sigma 
& := \sum_{x\in\bit^n} \rho_x
= \sum_{x\in\bit^n} \Ex_{\ket{\psi}\gets \Haar(2^n)} \left[ (Z^x \otimes I^{\otimes m}) \ketbra{\psi}{\psi}^{\otimes m+1} (Z^x \otimes I^{\otimes m}) \right] \\
& = \Ex_{\Vec{t}\in\cI_{d,m+1}} \sum_{x\in\bit^n} \left[ (Z^x \otimes I^{\otimes m}) \ketbra{s(\Vec{t})}{s(\Vec{t})} (Z^x \otimes I^{\otimes m}) \right] \\
& = \frac{d}{\binom{d + m}{m+1}} \cdot \sum_{\Vec{t}\in\cI_{d,m+1}} \sum_{j\in\bit^n} (\ketbra{j}{j} \otimes I^{\otimes m}) \ketbra{s(\Vec{t})}{s(\Vec{t})} (\ketbra{j}{j} \otimes I^{\otimes m}) \\
& = \frac{d}{\binom{d + m}{m+1}} \cdot \sum_{j\in\bit^n} \sum_{0\le r\le m} \sum_{\Vec{t}\in T_{j,r}^m} \frac{r+1}{m+1} \ketbra{j}{j} \otimes \ketbra{s(\Vec{t})}{s(\Vec{t})}.
\end{align*}
So we have
\begin{align*}
\sigma^{-1/2}
= \sqrt{\frac{\binom{d + m}{m+1}}{d}} \cdot \sum_{j\in\bit^n} \sum_{0\le r\le m} \sum_{\Vec{t}\in T_{j,r}^m} \sqrt{\frac{m+1}{r+1}} \ketbra{j}{j} \otimes \ketbra{s(\Vec{t})}{s(\Vec{t})}.
\end{align*}
Note that $\sigma^{-1/2}$ is PSD with the largest eigenvalue $\norm{\sigma^{-1/2}} = \sqrt{\binom{d + m}{m+1}(m+1)/d}$ (when $r = 0$). In~\cite{Col23}, the main technicality is to show Equation~(28):
\[
\Ex_{x\gets\bit^n} \Tr( \rho_x\sigma^{-1/2}\rho_x\sigma^{-1/2} )
\leq C' \cdot \left( \frac{m}{d} + \frac{m^7}{d^3} \right),
\]
where $C' > 0$ is some constant. Here, we provide an alternative and simpler proof. Since $\sigma^{-1/2}$ and $\rho_x$ are both PSD, the matrix $\sigma^{-1/2}\rho_x\sigma^{-1/2}$ is PSD as well. As $\rho_x$ is a density matrix, we have
\[
\Tr( \rho_x \cdot \sigma^{-1/2}\rho_x\sigma^{-1/2} )
\leq \norm{ \sigma^{-1/2}\rho_x\sigma^{-1/2} }.
\]
Then we use the submultiplicativity of the operator norm to obtain
\begin{align*}
& \norm{ \sigma^{-1/2}\rho_x\sigma^{-1/2} } \\
& \leq \norm{\sigma^{-1/2}} 
\cdot \norm{Z^x \otimes I^{\otimes m}} 
\cdot \norm{\Ex_{\Vec{t}\in\cI_{d,m+1}}\ketbra{s(\Vec{t})}{s(\Vec{t})}} 
\cdot \norm{Z^x \otimes I^{\otimes m}} 
\cdot \norm{\sigma^{-1/2}} \\
& = \norm{\sigma^{-1/2}}^2 \cdot \norm{\Ex_{\Vec{t}\in\cI_{d,m+1}}\ketbra{s(\Vec{t})}{s(\Vec{t})}} \tag{\text{unitaries have a unit operator norm}} \\
& = \frac{\binom{d + m}{m+1}\cdot (m+1)}{d} \cdot \frac{1}{\binom{d + m}{m+1}}
= \frac{m+1}{d}.
\end{align*}
Hence, it holds that
\[
\Ex_{x\gets\bit^n} \Tr( \rho_x\sigma^{-1/2}\rho_x\sigma^{-1/2} )
\leq \frac{m+1}{d}. \qedhere
\]
\end{proof}

\end{document}